\documentclass[12pt,a4,twoside]{book}

\usepackage{framed}
\usepackage{amsmath,amsfonts,amssymb} 
\usepackage{bm} 
\usepackage{theorem}

\newtheorem{theorem}{Theorem}
\newtheorem{lemma}{Lemma}

\usepackage{makeidx} 
\makeindex 

\usepackage{abbrevs}

\usepackage{nomencl} 
\renewcommand{\nomgroup}[1]{%
\ifthenelse{\equal{#1}{S}}{\vspace*{\baselineskip}\item[\textbf{Symbols}]\vspace*{\baselineskip}}{%
\ifthenelse{\equal{#1}{A}}{\vspace*{\baselineskip}\item[\textbf{Abbrivations}]\vspace*{\baselineskip}}{}}} 

\usepackage{ifpdf}
\ifpdf
\usepackage[pdftex]{graphicx}
\else
\usepackage{graphicx}
\fi

\usepackage{multicol}        		

\usepackage{sectsty}

\usepackage{units}		
\usepackage{mathrsfs}	

\usepackage{fancyhdr}
\usepackage{titlesec}
\usepackage{titletoc}

\usepackage{appendix}

\sectionfont{\sffamily \mdseries}
\subsectionfont{\sffamily \mdseries}
\subsubsectionfont{\sffamily }

\chapternumberfont{ \sffamily \mdseries}

\makeatletter 
\def\@seccntformat#1{\protect\makebox[0pt][r]{\sf  \csname
 the#1\endcsname $\,|$ }} 
\makeatother

\titleformat{\chapter}[display] 
{\bfseries \Large} 
{\filleft\MakeUppercase{\chaptertitlename}  \fontsize{60pt}{4ex} \selectfont \thechapter} 
{4ex} 
{\titlerule 
\vspace{2ex}%
\filright} 
[\vspace{2ex}%
\titlerule
]

\renewcommand{\chaptermark}[1]{\markboth{Chapter \thechapter.\ {#1}}{}}
\renewcommand{\sectionmark}[1]{\markright{\thesection\ \boldmath {#1}\unboldmath}}
\newcommand{\bea}{\begin{eqnarray}}
\newcommand{\eea}{\end{eqnarray}}
\newcommand{\beq}{\begin{equation}}
\newcommand{\eeq}{\end{equation}}
\newcommand{\nn}{\nonumber}

\newcommand{\C}[1]{{\mathcal{#1}}}
\newcommand{\B}[1]{{\mathbf{#1}}}

\newcommand{\BB}[1]{{\mathbb{#1}}}

\newcommand{\half}{{\frac{1}{2}}}

\newcommand{\quarter}{{\frac{1}{4}}}
\newcommand{\third}{{\frac{1}{3}}}

\newcommand{\threehalves}{{\frac{3}{2}}}

\newcommand{\ket}[1]{\,\vert #1\rangle}
\newcommand{\bra}[1]{\langle #1\vert\,}
\newcommand{\abs}[1]{\vert #1\vert}
\newcommand{\braket}[2]{\langle #1\vert\,#2\rangle}
\newcommand{\avg}[1]{\left\langle #1 \right\rangle}

\newcommand{\Tr}{\mathrm{Tr}}
\newcommand{\res}{\mathrm{res}}
\newcommand{\sgn}{\mathrm{sgn}}

\newcommand{\tW}{{\widetilde W}}

\newcommand{\Oh}{{\C O}}
\newcommand{\Z}{{\C Z}}

\newcommand{\Tinv}{\frac{1}{T}}
\newcommand{\fl}[1]{#1}

\newenvironment{proof}[1][Proof]{\begin{trivlist}
\item[\hskip \labelsep {\bfseries #1}]}{\end{trivlist}}

\newenvironment{property}[2][Property]{\begin{trivlist}
\item[\hskip \labelsep {\bfseries #1 #2 (P#2):}]}{\end{trivlist}}

\usepackage{makeidx}
\usepackage{comment}
\usepackage{subfigure}
\makeindex

\begin{document}
%

\newpage \thispagestyle{empty}

${}$

\vspace{4cm}
\noindent
 {\Huge APPLICATIONS OF RANDOM GRAPHS}\\[0.5cm]
  {\Huge TO 2D QUANTUM GRAVITY}\\[1.5cm]

\noindent
 {\Large{Max R. Atkin}}\\
 {\Large{Christ Church College, Oxford}}\\

\vspace{7cm}

\noindent
Thesis submitted in fulﬁlment of the requirements for the degree of Doctor of Philosophy at the University of Oxford.\\
\noindent
\\
Trinity term 2011

\newpage \thispagestyle{empty}

\noindent
Max Atkin\\
Rudolf Peierl's Centre for Theoretical Physics\\
Oxford University\\
Kebel Road\\
Oxford\\
United Kingdom

\vspace{11cm}

\noindent
\copyright \, Max Atkin 2011\\\\
All rights reserved. No part of this publication may be reproduced in any form without prior written permission of the author.






\newpage \thispagestyle{empty}
\frontmatter
\addcontentsline{toc}{chapter}{Abstract}
\renewcommand{\chaptermark}[1]{\markboth{ ABSTRACT}{}}
\chapter*{Abstract}

The central topic of this thesis is two dimensional Quantum Gravity and its properties. The term Quantum Gravity itself is ambiguous as there are many proposals for its correct formulation and none of them have been verified experimentally. In this thesis we consider a number of closely related approaches to two dimensional quantum gravity that share the property that they may be formulated in terms of random graphs.
In one such approach known as Causal Dynamical Triangulations, numerical computations suggest an interesting phenomenon in which the effective spacetime dimension is reduced in the UV. In this thesis we first address whether such a dynamical reduction in the number of dimensions may be understood in a simplified model. We introduce a continuum limit where this simplified model exhibits a reduction in the effective dimension of spacetime in the UV, in addition to having rich cross-over behaviour.

In the second part of this thesis we consider an approach closely related to causal dynamical triangulation; namely dynamical triangulation. Although this theory is less well-behaved than causal dynamical triangulation, it is known how to couple it to matter, therefore allowing for potentially multiple boundary states to appear in the theory. We address the conjecture of Seiberg and Shih which states that all these boundary states are degenerate and may be constructed from a single principal boundary state. By use of the random graph formulation of the theory we compute the higher genus amplitudes with a single boundary and find that they violate the Seiberg-Shih conjecture. Finally we discuss whether this result prevents the replacement of boundary states by local operators as proposed by Seiberg.


\addcontentsline{toc}{chapter}{Acknowledgments}
\renewcommand{\chaptermark}[1]{\markboth{ ACKNOWLEDGEMENTS}{}}
\chapter*{Acknowledgments}

I would like to thank my supervisor, John Wheater, for the significant time he has spent helping me and the numerous interesting projects he suggested. This thesis would not exist without his help. I would also like to give a very large thank you to Stefan Zohren who has provided significant encouragement and motivation in the final year of my DPhil. Both Stefan and John have been of invaluable help during the writing of this thesis, both suggesting numerous improvements and raising interesting questions.

I would like to acknowledge the support of STFC studentship PfPA/S/S/2006/04507 in providing me with the funding to undertake this research.

I would like to thank the senior tutors at Christ Church college, Axel Kuhn, Derek Stacey and Alan Merchant for providing me with the opportunity to teach and partake in admissions. This was a life-saver; especially when funds were short! My enjoyment of my time at Oxford was also greatly enhanced by my office mates, Sesh, Shaun, Seung Joo and Tom, who I thank for providing constant interesting conversation, whether related to physics or not. 

My family have also been a source of encouragement throughout my DPhil and have always been there to offer support. Finally, I want to thank Lauren for being there throughout it all and her endless willingness to lend an ear. It is to my family and Lauren to which I dedicate this thesis. 

\renewcommand{\chaptermark}[1]{\markboth{ {#1}}{}}
\addcontentsline{toc}{chapter}{Contents}
\tableofcontents 
\addcontentsline{toc}{chapter}{List of Figures}
\listoffigures 

\mainmatter
\renewcommand{\chaptermark}[1]{\markboth{Chapter \thechapter.\ {#1}}{}}
\renewcommand{\sectionmark}[1]{\markright{\thesection\ \boldmath {#1}\unboldmath}}
\lhead[\fancyplain{}{\thepage}]         	{\fancyplain{}{\rightmark}}
\chead[\fancyplain{}{}]                 	{\fancyplain{}{}}
\rhead[\fancyplain{}{\leftmark}]       	{\fancyplain{}{\thepage}}

\chapter{Introduction \label{Chap:Introd}}
Perhaps the most ambitious question that one may try to answer is ``What is reality?'', where by reality we mean the sensation of being alive, of experience. A priori there are many possibilities; a subscriber to idealism would believe one's experiences are entirely generated by one's mind. If, on the other hand, one subscribes to materialism, then what constitutes reality is a representation, generated by one's brain, of an independent physical universe. The quest to understand reality then becomes a matter of understanding the true nature of this physical universe. The 20th century saw many advancements in our understanding of the universe in this respect. The development of quantum mechanics fundamentally altered what is considered the true nature of the matter and forces, while the introduction of general relativity, which describes the dynamics of spacetime and how this manifests itself as gravity, banished the perception of space and time as being a static background stage on which dynamical processes took place.

However, the revolution in our view of the universe caused by quantum mechanics and general relativity was only half completed. General relativity assumes that all matter residing in spacetime is classical. However, we know such matter is fundamentally quantum mechanical and so we must, at the very least, construct a theory of gravity in which the matter is allowed to assume its true quantum mechanical form. 

There have been attempts to pursue the most conservative line of attack; keep the dynamics of spacetime as they are in general relativity, or at least classical, but modify how matter couples to gravity in order to allow for quantum mechanical rather than classical sources. However, there are a number of compelling arguments to suggest this does not produce a consistent theory \cite{Carlip:2008zf}. A more natural possibility is that gravity itself is quantum mechanical in nature. This avoids the aforementioned problems of coupling classical and quantum mechanical systems in addition to the philosophical appeal of all phenomenon in the universe being fundamentally quantum. Such a theory and the attempts to construct it are known as quantum gravity.

The are further compelling reasons that suggest a quantum mechanical formulation of general relativity is necessary. Most prominently, general relativity inevitably forms singularities in the form of black holes or the initial big bang singularity \cite{Bojowald:2007ky}, at such times the effect of quantum mechanics is expected to become important at the very least in the matter sector and most likely in the gravitational sector too.

Given that quantum mechanics provides the recipe of quantisation to produce a quantum mechanical theory starting with a classical one, the obvious way to construct quantum gravity is to apply this procedure to the theory of general relativity. In doing so one is immediately faced with a problem; what degrees of freedom should be quantised? The answer depends on what one considers to be the fundamental quantity whose dynamics general relativity describes. Furthermore, once we have decided upon the form the fundamental degrees of freedom take we must also choose a method of quantisation. For our purposes there two possible methods of quantisation, canonical and path-integral, and each may be applied to the theory resulting from our choice of degrees of freedom. As one might imagine this leads to a large number of different approaches to quantum gravity; all of which have been pursued at one point or another. We now consider some common answers to the above choices.
\\

\noindent {\emph{The spacetime metric $g$ is fundamental.}} In this approach we quantise the metric of spacetime. In a canonical approach this involves rewriting general relativity in terms of a spatial metric together with functions that describe how the spatial metric evolves between each spatial hypersurface. The resulting theory may be quantised leading to the Wheeler-de-Witt equation \cite{Bojowald:2007ky,QuantumGravity}. A path integral approach to the same problem requires one to make sense of integrals of the following form,
\beq
\int [dg] e^{i S_{\mathrm{EH}}[g]},
\eeq
where
\beq
\label{dDQG}
S_{\mathrm{EH}}[g] = \int_{\C{M}} d^dx \sqrt{g}\left( \kappa R[g] + \mu\right),
\eeq
where $\mu$ is a constant known as the cosmological constant and $\kappa$ is a constant satisfying $\kappa \propto 1/G$, where $G$ is Newton's constant. Additionally one must decide exactly which geometries the integral sums over.
\\

\noindent {\emph{The perturbation to flat space, $h$ defined by $g = \eta +h$, is fundamental.}} In this view, the vacuum of general relativity is Minkowski space and we quantise the perturbations around it. This reduces the problem of quantum gravity to a problem of quantising a spin-2 quantum field theory in flat space. One can approach this using canonical or path-integral methods, however in both cases one finds that the resulting theory is perturbatively non-renormalisable \cite{hooft}. Although this would appear to rule out treating general relativity as anything other than an effective theory one suggestion due to Weinberg is that the coupling of the theory may in fact run to a non-Gaussian fixed point in the UV \cite{weinberg}. This would render the theory UV complete while still appearing non-renormalisable at the level of perturbation theory. A second possibility arises in the program of string theory in which all fundamental objects are treated as quantised strings. In this approach the spectrum of the string contains a spin-2 state and the S-matrix for the scattering of this state reduces to that of the S-matrix obtained for graviton scattering when treating $h$ as fundamental \cite{Polchinski1,GSW1}.
\\

\noindent {\emph{No variables in general relativity are fundamental.}} A more radical possibility is that general relativity doesn't contain the correct microscopic degrees of freedom at all. If this is the case then general relativity only becomes a description of spacetime in the long distance limit and attempting to quantise general relativity would be as misguided as attempting to understand the short distance regime of water by quantising the Navier-Stokes equation! Historically the problem with this possibility is that it was very hard to know where to begin as one must invent a theory in which spacetime and gravity appear as low energy concepts; in effect solving the problem of quantum gravity in a moment of brilliant insight. However, the advent of string theory has changed this situation. String theory provides a way to construct order by order a perturbative expansion in which each term is finite and correctly reproduces the low energy scattering of gravitons found in the linearised theory. Furthermore, this is achieved by quantising a theory which makes no assumptions about the classical dynamics of spacetime. The main problem however is the question of how to define the theory from which the perturbative series of various string theories arise. The solution to this problem has been tentatively named M-theory \cite{Polchinski2}, however a full definition of it has yet to be given.

\section{Thesis Outline}
As one can see, the field of quantum gravity is vast and one can only hope, in a thesis such as this, to consider a small piece of it. In this thesis we will be particularly interested in the dimension of spacetime and the allowed boundary conditions of two dimensional quantum gravity. Quantum gravity in two dimensions is special for a number of reasons, one being that a number of the approaches listed above coincide. In particular the approach in which the metric $g$ is fundamental coincides with the string theory approach and so by studying one of these theories we may study the other. 

In chapter \ref{ChapBack} we introduce the theory of pure two dimensional quantum gravity in the metric-is-fundamental approach. We review how this theory may be obtained as a scaling limit of a discrete theory known as Dynamical Triangulation (DT) and introduce the matrix model and combinatorial techniques that may be used to compute certain observables such as the disc-function and string susceptibility. We also introduce Causal Dynamical Triangulation (CDT) models and review their relation to DT.

In chapter \ref{ChapDim} we consider the final observable not considered in chapter \ref{ChapBack}; the dimension of spacetime. In a theory of quantum gravity one expects that spacetime at very short distances is violently fluctuating. Furthermore, in string theory one might expect to also begin to probe compact extra dimensions. Both of these effects would lead the dimension of spacetime to deviate from its infrared value. We consider this question in the context of the approach in which the metric is fundamental. We argue that a-priori it is not necessarily equal to the dimension of the underlying discretisation and introduce the concept of the Hausdorff dimension as a tool to characterise fractal structures. We review the calculation showing that the Hausdorff dimension of DT is four and then argue that the Hausdorff dimension in CDT is two. We then argue that the Hausdorff dimension is inadequate to fully distinguish between smooth and fractal spaces and so we introduce a second definition of dimension, known as the spectral dimension, which is based on the properties of the Laplacian operator on the fractal structure. We argue that the spectral dimension can be obtained by considering random walks on the triangulation. Finally we review the recently observed phenomenon of dimensional reduction, in which the spectral dimension of the UV is lower then IR. By considering random graphs derived from CDT we then begin the search for a simple random graph model in which we might capture the phenomenon of dimensional reduction.

Chapter \ref{ChapComb} consists mainly of work published in \cite{Atkin:2011ak}. We briefly introduce random combs and review some known results. We give a definition of the spectral dimension and then explain how this definition can be extended to show different spectral dimensions at long and short distance scales. We then introduce a simple model which we prove does in fact exhibit a spectral dimension that is different in the UV and IR. We generalise these results to combs in which teeth of any length may appear with a probability governed by a power law and examine the possibility of intermediate scales in which the spectral dimension differs from both its UV and IR values. Finally we analyse the case of a comb in which the tooth lengths are controlled by an arbitrary probability distribution and show that continuum limits exist in which the short distance spectral dimension is one while the long distance spectral dimension can assume values in one-to-one correspondence with the positions of the real poles of the Dirichlet series generating function for the probability distribution. 

In chapter \ref{ChapString} we first review how matter coupled to gravity may be realised using matrix models. We introduce the Ising matrix model and also consider multicritical points. We then move on to quickly review minimal models, string theory, Liouville theory and minimal string theory. The notion of a boundary condition in string theory is promoted to a dynamical object known as a brane. The different boundary conditions then correspond to different types of brane. If one wants to obtain a non-perturbative description of string theory then the spectrum of branes is important. We pursue the question, that arises due to a conjecture of Seiberg and Shih \cite{Seiberg:2003nm}, of how many distinct branes exist in $(p,q)$ minimal string. We review the fact that this conjecture fails for cylinder amplitudes as found in \cite{Gesser:2010fi} and conjecture a way in which it could be fixed. The final portion of this chapter appears in \cite{Atkin:2010yv}.

Chatper \ref{ChapBrane} consists of the remainder of the work in \cite{Atkin:2010yv} in which we apply matrix model techniques to the question of whether the conjecture of the previous chapter holds. In particular we consider the matrix model which describes a $(3,4)$ minimal string theory in the presence of a boundary magnetic field. We show how a general class of matrix models, which includes the one describing the $(3,4)$ string may be solved, and by varying the boundary magnetic field we reproduce the conformal boundary states of the $(3,4)$ model. Using our general solution to this model we then compute the disc-with-handle amplitude for all conformal boundary conditions and compute their deviation from the Seiberg-Shih relation. We argue that these results show that our previous conjecture does not hold. This makes it very difficult to see how the Seiberg-Shih relations could possibly be true. Finally we consider a different approach to testing the Seiberg-Shih relation based on expanding boundary states in terms of local operators. We find that the expansion in local operators appears to be valid for only certain boundary types. We support this by reproducing a recent calculation of Belavin in which the one-point function on the torus was computed.

\chapter{Discrete Approaches to 2D Quantum Gravity\label{ChapBack}}
\label{Backgroundnew}
In the previous chapter we saw that there were a variety of approaches to quantum gravity which could be classified according to which degrees of freedom were considered to be fundamental. In two dimensions a number of these approaches coincide; in particular string theory and approaches in which the metric is considered fundamental. We will therefore restrict our attention to these in the remainder of this thesis. In this chapter we will consider the simplest case of the metric-is-fundamental approach, in which no matter is present and the only dynamical quantity is that of gravity itself. We first review why in two dimensions some of the problems of higher dimensional theories are alleviated before going on to review the few simple observables that exist in this theory. We then motivate a different approach to computing the path integral based on a discretistion of the contributing geometries and show that there exists a scaling limit which allows us to compute some of the observables.

In the metric-is-fundamental approach we must evaluate the partition function,
\beq
Z = \int \frac{[dg]}{\mathrm{Diff}} e^{i S_{\mathrm{EH}[g]}},
\eeq
where we have divided by the volume of the group of diffeomorphisms of the spacetime manifold, since this is a gauge symmetry in general relativity. The usual strategy to evaluate path integrals is to Wick rotate to a Euclidean theory. This has the effect of converting the oscillatory measure to one which is exponentially damped. An immediate problem arises in the case of gravity, in that the resulting Euclidean action is not bounded from below and so it appears the integral does not exist. A further problem is that it is unclear that a Wick rotation from the Euclidean theory back to Lorentzian theory even exists. 

Thankfully, when the spacetime $\C M$ is two dimensional some of these problems are alleviated. Due to the Gauss-Bonnet theorem the curvature term in the action is entirely topological and so the problem of the unboundedness of the action is no longer present. 
However, restricting our attention to two dimensions does not solve the problem of what level of causality violations to allow i.e. how to Wick rotate between the two theories. We will come back to this point later in this chapter. For now we will put no causality constraint on the geometries appearing in the integral and hence we shall sum over all Euclidean geometries. Application of the Gauss-Bonnet theorem to \eqref{dDQG} then gives,
\beq
\label{2DQG}
Z = e^{4\pi \kappa \chi} \int \frac{[dg]}{\mathrm{Diff}} \exp \left( -\int_{\C M} d^2 \sigma \sqrt{g(\sigma) } \mu \right),
\eeq
where $\chi$ is the Euler characteristic of $\C M$ and is given by $\chi = 2 - 2h$, where $h$, the genus, is the number of handles of $\C M$. In order to perform this integral we must define the measure $[dg]$. This is a difficult technical problem whose solution results in a theory in which the overall scale factor of the measure acquires dynamics. The theory describing the dynamics of the scale factor is itself a difficult theory to solve. We therefore postpone our discussion of this approach until Chapter \ref{ChapString}. However, it is still worth considering the observables, within the continuum framework, that one could measure in a theory of pure quantum gravity. 
There are a number of observables that have been commonly considered in the literature. These consist of,

\begin{itemize}
\item[1.] Integrated correlation functions of local operators are the most obvious observables, although these also proved the most difficult to compute. In the most general case, when the manifold $\C M$ has $n$ boundaries, the action given in \eqref{2DQG} must be extended to include a boundary action,
\beq
\label{2DQGboundary}
Z = e^{4\pi \kappa \chi} \int \frac{[dg]}{\mathrm{Diff}} \exp \left( -\int_{\C M} d^2 \sigma \sqrt{g(\sigma) } \mu - \sum_i \int_{\partial M_i} ds \sqrt{h} {\mu_B}_i \right),
\eeq
where $h$ is metric induced on the boundary curve and ${\mu_B}_i$ are boundary cosmological constants and $\chi = 2 - 2h - b$, where $b$ is the number of boundaries. By utilising the symmetries of the theory any integrated correlation function has the form $\mu^{(2-\gamma)\chi/2} F({\mu_B}_i/\sqrt{\mu})$ \cite{Teschner:2001rv,Nakayama:2004vk}. The quantity $\gamma$ is known as the string susceptibility \cite{Ginsparg:1993is} and is an observable that has been used to compare different approaches to two dimensional quantum gravity.
\item[2.] The Hartle-Hawking wave function. This is the partition function for the appearance of a universe ``from nothing'' and evolving into some prescribed state. Geometrically this corresponds to summing over all geometries with a disc topology with a fixed boundary condition. Throughout this thesis we will refer to this quantity as the disc-function. More generally one can consider multi-loop amplitudes in which the spacetime manifold has a number of boundary loops; a disc would be a 1-loop amplitude and a 2-loop amplitude is often known as a cylinder amplitude;
\item[3.] The dimension of spacetime. Living inside a particular universe one may probe the dimension of the spacetime by performing scattering experiments or even simply measuring the electro-static force between two charges. If there exist changes to the dimension of spacetime on the length scales such experiments probe, then we should be able to measure them. This observable may seem trivial as it appears we set the dimension spacetime in our theory at the outset, however, when probing the short distance scale of spacetime we expect that spacetime itself will become less smooth due to quantum gravity effects (such as microscopic black holes, possibly wormholes etc). This could easily alter the effective dimension of spacetime at these scales even if the underlying theory is of fixed dimension \cite{Carlip:2009aa}.
\end{itemize}

In order to compute these observables we will consider an alternative method of computing the above integral in which observables are calculated in a different, but related, theory with a far simpler measure. This will come at the expense of having a non-trivial dictionary relating the observables in the original theory to the one with a simple measure. We will review such an approach in the next section.

\section{Dynamical Triangulation}
The related theory which we will consider is motivated by considering lattice regularisations of flat space quantum field theories. If we take the lattice regularised theory as our starting point then one can obtain a continuum theory, corresponding to the original quantum field theory, by taking the lattice spacing to zero while also appropriately scaling the other parameters in the theory. In the case of gravity there exists a number of approaches to defining a lattice regularisation, however we shall utilise only one here. It is known as Dynamical Triangulations (DT) and it has a long history which is reviewed in \cite{Ginsparg:1993is,EynardReview}. In this case the spacetime geometry is encoded in the lattice, which is canonically composed of triangles. 
The lattice links themselves are always of fixed length \footnote{Links of fixed length is in contrast with the earlier approach of Regge calculus in which the link lengths were integrated over while keeping the triangulation fixed.}. For an example of a triangulation see fig \ref{dual}. The benefit of this discrete model is that if we wish to compute a partition function, then rather than having to integrate over geometries we must sum over the triangulations. Because the latter is a sum, the issue of defining an integration measure on the space of geometries does not arise. If we are able to compute the partition function of the DT model then the hope is that there exists a continuum limit in which we obtain the original partition function \eqref{2DQG}. 
\begin{figure}[t]
  \begin{center}
    \includegraphics[width=10cm]{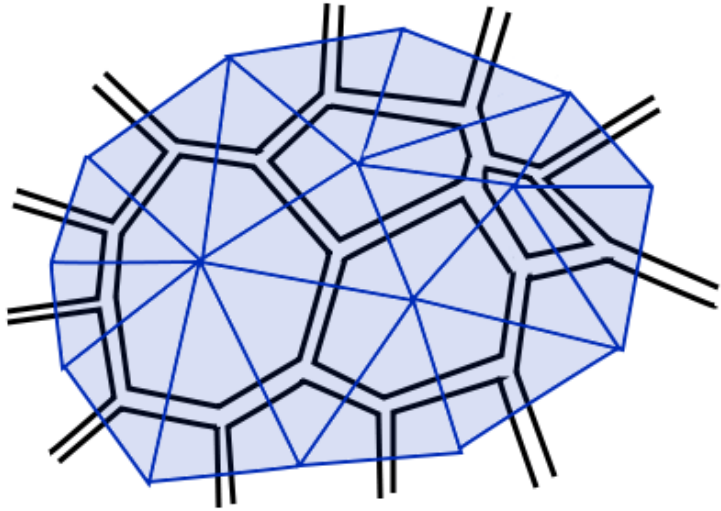}
    \caption{A triangulation (grey) together with its dual graph (black).}
    \label{dual}
  \end{center}
\end{figure}
In order to compute the partition function we will set each link to be of length $a$ meaning each triangle has an area proportional to $a^2$. The discretised version of the integral \eqref{2DQG} is then,
\beq
\label{ZDT}
Z^\mathrm{DT}_h = \sum_{T \in\mathrm{Tri}_h} e^{-\Lambda A(T)} = \sum_{T \in\mathrm{Tri}_h} e^{-a^2 \Lambda |T|} = \sum^\infty_{n=0} \C{N}_h(n) g^{n},
\eeq
where $\mathrm{Tri}_h$ is the set of triangulations of genus $h$, $|T|$ is the number of triangles in the triangulation $T$, $\C{N}_h(n)$ is the number of triangulations of a genus $h$ manifold containing $n$ triangles and $g = e^{-a^2 \Lambda}$. From the final equality we see that the evaluation of the partition function has been reduced to computing the generating function for $\C{N}_h(n)$, which is a graph counting problem. However, the graphs we must count are rather complicated; in particular the order of each vertex is unbounded and there is a non-local restriction that each face of the graph is triangular. Both of these problems may be solved in a simple way by considering the dual graph of the triangulation. The dual graph is constructed by associating a vertex to each face of the triangulation and a link joining any two vertices whose corresponding faces share an edge. The resulting dual graphs have the simple property that each vertex is of order three with no non-local restrictions. Furthermore since each graph has a unique dual we need only count the dual graphs. The dual graph construction is illustrated in fig \ref{dual}.

From the figure it is amusing to note that dual graph looks very much like a Feynman diagram composed of many vertices. In fact there exists a very useful class of tools, known as matrix models, which exploit this fact and are able to compute generating functions for the graph counting function $\C{N}_h(n)$ for a variety of graphs, including those with labels. Recall that in quantum field theory the perturbative computation of the partition function is organised such that the contribution of each process to the overall amplitude is encoded in its Feynman graph; in particular the overall power of the coupling associated to each graph is given by the number of vertices it contains. If one were to find a way of computing the partition function then the perturbative expansion could be recovered by expanding the partition function around zero coupling. If the quantum field theory was sufficiently simple then one could use the resulting expansion to extract information about the number of graphs contributing at each order of the coupling.



\subsection{Random Graphs and Matrix Models}
If we are to develop a quantum field theory that will solve our graph counting problem then it must also have the property that its perturbative expansion is organised by the topology of the graph. This requirement restricts the possible candidates to matrix valued field theories as such theories will produce a perturbative series, in the inverse size of the matrix, in which only graphs of a particular topology contribute at a given order. 


We must also make the theory simple enough that one can extract the information concerning the graphs. The simplest matrix valued theory is one in which the spacetime dimension is zero; such a theory is known as a matrix model. An example of such a model, which also has the property that it produces graphs in which all vertices have order four is,
\beq
\label{pureQGMM}
\Z^{MM}(g) = \int [dM] e^{-N \Tr \left[ \half M^2 + \frac{g}{4} M^4 \right]},
\eeq
where the metric field $M$ is an $N \times N$ hermitian matrix and the measure $[dM]$ is defined to be the flat Lebesgue measure,
\beq
\label{MMmeasure}
\prod^N_{i=0} dM_{ii}  \prod_{i<j}d\mathrm{Re}\left[M_{ij}\right]d\mathrm{Im}\left[M_{ij}\right].
\eeq
If we treat the quartic term in the action as a perturbation to the free theory then we may construct a Feynman graph expansion for $Z^{MM}$. The contribution to the total amplitude from a given graph is $N^{2-2h} g^n$, where $h$ is the genus of the graph and $n$ is the number of vertices it contains. We therefore see that the expansion of $\C Z^{MM}$ around $g=0$ coupling will yield the series,
\bea
\Z^{MM}(g) &=& \sum^\infty_{h=0} N^{2-2h} Z^{(h)}(g), \\
\Z^{(h)}(g) &=& \sum^\infty_{n=0} \tilde{\C{N}}_h(n) g^n.
\eea
The function that counts the number of graphs with $n$ vertices, $\tilde{\C{N}}_h(n)$, differs from the one that appeared in the expression for $Z^{DT}$ since the perturbative expansion includes disconnected graphs. In order to compute $\C{N}_h(n)$ we must instead consider the free energy,
\beq
\log\Z^{(h)}(g) = \sum^\infty_{n=0} \C{N}_h(n) g^n.
\eeq
We see that we have actually done much better than merely being able to compute $\C{N}_h(n)$; we have reproduced the $Z^{\mathrm{DT}}$ partition function exactly! \footnote{The observant reader will note that the relation of $\log\Z^{(h)}(g)$ to $Z^{\mathrm{DT}}$ is not, in fact, exact, due to the above matrix model producing not triangulations but quadrangulations. We could have considered a matrix model with a cubic potential and so obtained an exact equivalence with $Z^{\mathrm{DT}}$. However because the matrix model with cubic and quartic potentials have the same continuum limit we will continue to consider the quartic model for reasons of pedagogy; the quartic model is simpler to solve.}

\subsection{Solving the Matrix Model: Loop equations}
There exist a number of approaches to computing $\Z^{MM}(g)$, among them the saddle point method \cite{Makeenkobook,Ginsparg:1993is,EynardReview}, orthogonal polynomials \cite{Ginsparg:1993is,EynardReview} and loop equations \cite{EynardReview,Eynard:2002kg}. It is the last of these three that we will review in this section. In order to keep our discussion general we will consider a generalised one hermitian-matrix model, given by 
\beq
\label{1MM}
Z^{MM} = \int [dM] e^{-N \Tr V(M) },
\eeq
where $V$ is a polynomial potential. A loop equation is defined as any equation that may be derived from an infinitesimal change in the integration variable which preserves the domain of integration. For the matrix model defined previously, the most arbitrary infinitesimal change in the integration variable would be $M \rightarrow M + \epsilon f(M)$ where $f$ is an arbitrary function satisfying $f(M^\dagger) = f(M)^\dagger$. In the case of \eqref{1MM} the effect of this change of variable is,
\beq
Z^{MM} = \int [dM] (1 + \epsilon J(M)+ \C O(\epsilon^2)) e^{-N \Tr V(M) - \epsilon K(M) + \C O(\epsilon^2)},
\eeq
where $J$ is the trace of the Jacobian matrix for the transformation and $K$ is the change in the potential, i.e. in this case $K = N \Tr V'(M)f(M)$. By equating the above equation to \eqref{1MM} and considering the order $\epsilon$ terms we obtain,
\beq
\label{LEJK}
\avg{J} = \avg{K}.
\eeq
A central role in the loop equation method is played by the function defined by
\beq
\label{WMdef}
W_M(x) = \frac{1}{N}\avg{\Tr\frac{1}{x-M}}\equiv \sum^\infty_{n=0} \frac{w_n(g)}{x^{n+1}},
\eeq
where $w_n(g)  = \frac{1}{N}\avg{\Tr M^n}$ and $W_M(x)$ is known as the resolvent. Its importance is due to the fact that the free energy may be computed from it. Furthermore, a very definite interpretation may attached to it. Note that the Feynman graphs contributing to $w_n(g)$ are those which have $n$ external legs. Such graphs may be interpreted as corresponding to a discretisation of a disc and so we see the resolvent acts as a generating function for such amplitudes. We therefore have,
\beq
\label{wkdef}
w_k(g) = \sum^\infty_{n = 0} \C{N}(n,k) g^n,
\eeq
where $\C{N}(n,k)$ denotes the number of triangulations containing $n$ triangles and $k$ boundary links. Obviously such graphs are the ones which contribute to the disc function observable discussed in the previous section and therefore one might suspect that the continuum disc function can be extracted from the resolvent. This indeed turns out to be the case and we will compute the disc function in the next section. Finally it is worth emphasising at this point that although the resolvent function appears naturally in the matrix model approach, the technique of constructing a generating function for particular classes of amplitudes, for example in this case disc-amplitudes, will be of more general use, as we will see when considering the combinatorial approach to loop equations.

In order to use \eqref{LEJK} to calculate $W_M$ we first need to specify the change of variable we will use. We make the choice $f(M) = 1/(x-M)$. By analysing the effect of the change of variable on the measure \eqref{MMmeasure}, one can show \cite{Eynard:2002kg} that if $f(M) = 1/(x-M)$ then,
\beq
J = \Tr\frac{1}{x-M}\Tr\frac{1}{x-M}.
\eeq
Using this result in \eqref{LEJK} we get,
\bea
\frac{1}{N^2}\avg{\Tr\frac{1}{x-M}\Tr\frac{1}{x-M}} &=& \frac{1}{N}\avg{\Tr\left[V'(M)\frac{1}{x-M}\right]},\\
&=& \frac{V'(x)}{N}\avg{\Tr\left[\frac{1}{x-M}\right]} - Q(x),
\eea
where $Q(x) = \frac{1}{N}\avg{\Tr \right[ \frac{V'(x) - V'(M)}{x-M} \left]}$ is a polynomial in $x$. We now note that the function on the left hand side has contributions from connected and disconnected Feynman diagrams. For such amplitudes we can in general write  $\avg{\Tr A \Tr B} = \avg{\Tr A} \avg{\Tr B} + \avg{\Tr A \Tr B}_c$, where the $c$ stands for ``connected''. We therefore get,
\bea
\label{LEWM}
W_M(x)^2 + \frac{1}{N^2}\avg{\Tr\frac{1}{x-M}\Tr\frac{1}{x-M}}_c &=& V'(x)W_M(x) - Q(x).
\eea
This equation is the loop equation associated with the change of variable defined by $f(M) = 1/(x-M)$. Let us define the two-loop function as,
\beq
W_{M;M}(x_1,x_2) = \avg{\Tr\frac{1}{x-M}\Tr\frac{1}{x-M}}_c.
\eeq
We see that since the diagrams contributing to this amplitude are cylindrical then it goes as $\sim 1$ as $N \rightarrow \infty$.  Furthermore, we expect that both $W_M$ and $W_{M;M}$ will have a large $N$ expansion of the form,
\bea
W_{M}(x) &=& \sum^\infty_{n=0} N^{-2n} W^{(n)}_{M}(x), \\
W_{M;M}(x_1,x_2) &=& \sum^\infty_{n=0} N^{-2n} W^{(n)}_{M;M}(x_1,x_2).
\eea
If we substitute these expansions into the loop equation then we find,
\beq
\label{LEWM0}
W^{(0)}_M(x)^2 = V'(x)W^{(0)}_M(x) - Q^{(0)}(x).
\eeq
We see that in the large $N$ limit we end up with a quadratic equation for $W^{(0)}_M(x)$ in which $Q^{(0)}(x)$ contains a number of unknown constants. Note that this equation defines an algebraic curve in the variables $x$ and $W^{(0)}_M$. This is known as the spectral curve of the matrix model.

Up to this point we have made an assumption that the partition function for the matrix model admits an expansion in inverse powers of $N$. This assumption is in fact sometimes incorrect, as oscillatory terms appear in the large $N$ expansion. The exact condition under which a large $N$ expansion and therefore a topological expansion exists is given in \cite{Eynard:2002kg} to be when the spectral curve is genus zero. We see that for the spectral curve of the matrix model considered here, the genus zero condition is equivalent to there being a single branch cut. By enforcing this condition we may compute the unknown constants appearing in $Q^{(0)}(x)$.
The solution to \eqref{LEWM0} is
\beq
\label{WM0sol}
W^{(0)}_M(x) = \half \left[ V'(x) - \sqrt{V'(x)^2 - 4 Q^{(0)}(x) } \right].
\eeq
Note the discriminant of this is a polynomial in $x$. We may therefore enforce the genus zero condition by requiring that all but two of the roots of the discriminant are double roots. For the particular potential in \eqref{pureQGMM}, we therefore must be able to write,
\beq
\label{1cutWM0}
W^{(0)}_M(x) = \half \left[ V'(x) - (\alpha(g)  x^2-A(g)^2)\sqrt{x^2-B(g)^2} \right],
\eeq
where $\alpha(g)$, $A(g)$ and $B(g)$ may be determined by requiring that the expression in \eqref{1cutWM0} reproduces the large $x$ expansion of \eqref{WM0sol}. Because $W_M(x)$ was defined using a series valid for large $x$ we know that there exists at least one branch of the solution which has the property that $W^{(0)}_M(x) \sim 1/x$ as $x \rightarrow \infty$. If one expands \eqref{1cutWM0} about $x=\infty$, then for generic values of $\alpha(g)$, $A(g)$ and $B(g)$ the leading order term will be of higher order than $1/x$. Requiring that the leading order term is $1/x$ gives the following expressions for $\alpha(g)$, $A(g)$ and $B(g)$,
\bea
\alpha(g) &=& g \nn \\
A(g)^2 &=& 1 + \half B(g)^2 g, \nn \\
\frac{1}{16} B(g)^2 (4 + 3 B(g)^2 g) &=& 1.
\eea

\subsection{The Continuum Limit}
In order to obtain a result that corresponds to a continuum theory we must take what is known as a scaling limit of the above computation. This corresponds to tuning the coupling of the theory to the value at which it undergoes a second order phase transition. At such a phase transition the average number of triangles in the triangulation diverges and by rescaling their size along with the couplings in theory we can derive the continuum theory which describes the model at this critical point. We see from \eqref{ZDT} that the number of triangles in the partition will diverge when $g$ reaches the radius of convergence of \eqref{ZDT}. This radius of convergence can be easily computed by considering the quantity $w_2(g)$ which can be obtained from the coefficient of $x^{-3}$ in the large $x$ expansion of $W^{(0)}_M(x)$. 
The radius of convergence for $w_2(g)$ can be read off from the position of its branch points as a function of $g$.

Having found the critical point of $g$, $g_{c}$, it remains to fix how we scale the parameters in our theory. If we refer back to \eqref{ZDT} then we expect the bare cosmological constant $\Lambda$ to be related to the continuum cosmological constant $\mu$ by,
\beq
\Lambda = a^{-2} \tilde{\Lambda} + \mu
\eeq
which implies $g = g_c e^{-a^2\mu}$, where $g_c = e^{-\tilde{\Lambda}}$. As it stands, if we were to substitute this scaling relation in to the equation for the resolvent then this would correspond to taking a limit in which the number of triangles in the bulk diverges while also shrinking each triangle to zero size. Such a limit would yield a continuous geometry in the bulk. However, note that we have not scaled the parameter $x$, which controls the number of links on the boundary of the disc. Leaving the parameter $x$ unscaled would produce geometries which, although continuous in the bulk, have a boundary of finite size in lattice units. From the perspective of the bulk such a boundary would appear infinitesimal and therefore look like a local operator insertion. We will denote the geometry which has a boundary of length $n$ in lattice units by $w_n(g)$.

The resolvent provides one way of finding a scaling limit in which the number of boundary links can be made to diverge. By noting that the resolvent has a radius of convergence in $x$ determined by the position of the branch point $B(g)$, we can see that by tuning $x$ to the value $B(g_c)$, the average number of boundary links will diverge. We therefore enforce a scaling of $x = x_c e^{a \zeta}$, where $\zeta$  will become the boundary cosmological constant. Inserting these scaling forms for $g$ and $x$ into the solution for $W_M^{(0)}(x)$ and taking $a \rightarrow 0$, we obtain,

\beq
\label{smallaseries}
W_M^{(0)} = \frac{\sqrt{2}}{3} - a\frac{\zeta}{\sqrt{2}}+\frac{a^{3/2} \mu^{3/4}}{3}\left[\left(z+\sqrt{z^2-1}\right)^{3/2}+\left(z-\sqrt{z^2-1}\right)^{3/2} \right] +\C O(a^{3/2}),
\eeq
where we have introduced the dimensionless quantity $z = \zeta/\sqrt{\mu}$.

A peculiarity of this scaling limit is that the universal quantity that corresponds to the continuum theory amplitude is not the leading contribution to $W_M^{(0)}$ as $a$ goes to zero. The obvious way to identity the universal result is to compute the same quantity using different potentials in the matrix model and then note which term in the scaling limit remains unchanged. A way to short-cut this procedure is to conjecture that the universal part of the above expansion corresponds to the first term non-analytic in the bulk and boundary cosmological constant. One can motivate this by arguing that non-analyticities only appear in the thermodynamic limit, i.e. when the surface and boundary are of infinite area and length in lattice units. This conjecture is borne out by comparison to continuum calculations. We therefore have 
\beq
\tilde{W}_M^{(0)}(\zeta,\mu) = \frac{\mu^{3/4}}{3}\left[\left(z+\sqrt{z^2-1}\right)^{3/2}+\left(z-\sqrt{z^2-1}\right)^{3/2} \right],
\eeq
where we have introduced the notation of adding a tilde to a quantity in order to denote its continuum version. Note that this continuum quantity is related to the term appearing in \eqref{smallaseries} by a wavefunction renormalisation which removes the overall factor of $a^{3/2}$.

This completes our computation of the continuum resolvent, and hence the disc function observable, using the matrix model formulation of DT. In particular we have shown that there exists a non-trivial limit of the discrete equations. From this we may also extract the string susceptibility in the case of a disc by removing the marked point on the boundary. The mark may be removed by integrating the above expression w.r.t $\zeta$ to obtain, $\gamma = -\half$. Another question one might ask is whether there exists a different disc function obtained by scaling $k$ in the equation $w_k(g)$ while tuning $g$ to its critical value. This would yield a disc function in which the boundary length, rather than boundary cosmological constant was fixed. This quantity will turn out to be useful in the future and so it is important to understand how it is related to the disc function. Given that the boundary cosmological constant has dimensions of $a^{-1}$ we should expect the boundary length to scale like $k = \lambda/a$, we therefore define the continuum disc-function,
\beq
\bar{W}_M^{(0)}(\lambda,\mu) = \lim_{a\rightarrow 0} a^{\eta} x_c^{-\lambda/a} w_{\lambda/a} (g_c e^{-a^2 \mu}),
\eeq
where $\eta$ and $x_c$ are chosen to give a non-trivial limit. With this definition it is reasonably straightforward to see that $\bar{W}_M^{(0)}(\lambda,\mu)$ is related to $\tilde{W}_M^{(0)}(\zeta,\mu)$ by a Laplace transform,
\beq
\label{Wfixedlmark}
\tilde{W}_M^{(0)}(\zeta,\mu) = \int^\infty_0 d\lambda e^{-\lambda \zeta} \bar{W}_M^{(0)}(\lambda,\mu).
\eeq
Note that the above relation applies to boundaries which carry a marked point. If we were to remove the marked point then the relationship would become,
\beq
\label{Wfixedl}
\tilde{\omega}_M^{(0)}(\zeta,\mu) = \int^\infty_0 \frac{d\lambda}{\lambda} e^{-\lambda \zeta} \bar{\omega}_M^{(0)}(\lambda,\mu),
\eeq
where we have used an $\omega$ to denote amplitudes with non-marked boundaries.

The matrix model approach is particularly useful if one wants to compute higher genus amplitudes such as discs with handles. This is due to the fact that the loop equations of the matrix model produces relations between quantities which include all genus contributions to an amplitude. This was seen in \eqref{LEWM}, which was then used to derive an equation satisfied by the genus zero contribution \eqref{LEWM0}. Obviously, the loops equation \eqref{LEWM} may also be used to obtain expressions for the higher genus corrections to $W_M^{(0)}$ and this will be of great interest in later chapters.

The question of whether this continuum theory we have uncovered is at all related to the original continuum theory \eqref{2DQG}, is a question we leave until Chapter \ref{ChapString}. However, it is worth jumping to the punchline; matrix model computations do indeed match the continuum formulation whenever two quantities have been computed in both. It is therefore worthwhile spending some time investigating what we may learn in the discrete formulations of two dimensional quantum gravity.

\subsection{Combinatorial Interpretation of the Loop Equation}
We have seen that the loop equations as obtained from the matrix model are an efficient way to compute the disc function and that the matrix model is particularly suited to computing higher genus amplitudes. However, it is useful to consider a distinct method of obtaining the loop equations which, although it is not as useful for computing the higher genus amplitudes, gives greater understanding of the content of the equations. This method consists of performing a more direct combinatorial analysis of the triangulations by considering the way one may build up a triangulation of a given size from smaller triangulations. This eventually leads to a recursive equation for the number of triangulations of a given size which is equivalent to the loop equations.

\begin{figure}[t]
  \begin{center}
    \includegraphics[width=10cm]{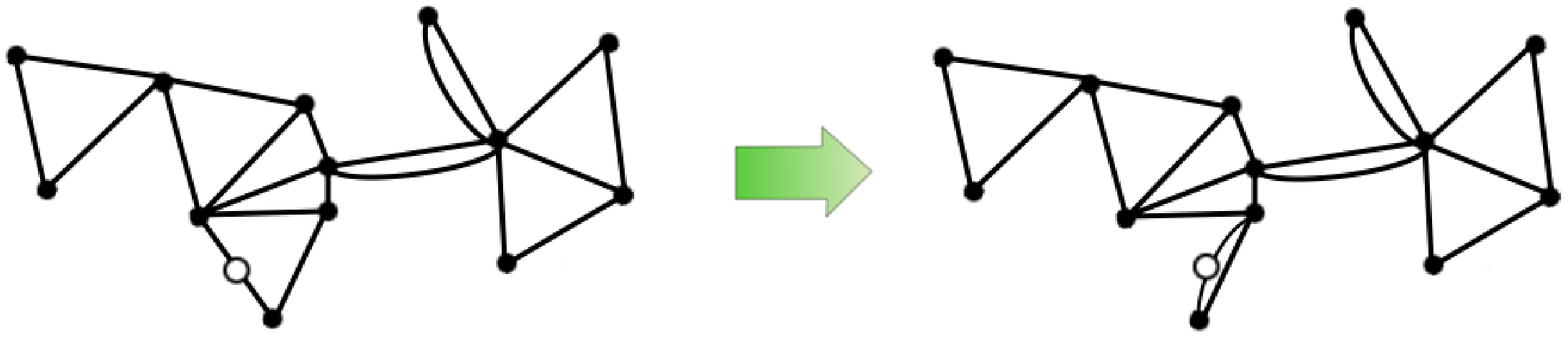}
    \caption{Move 1 relates an unrestricted triangulation (left), with a mark denoted by the white circle, to another unrestricted triangulation with less triangles.}
    \label{unrestrictTriMove1}
    \includegraphics[width=10cm]{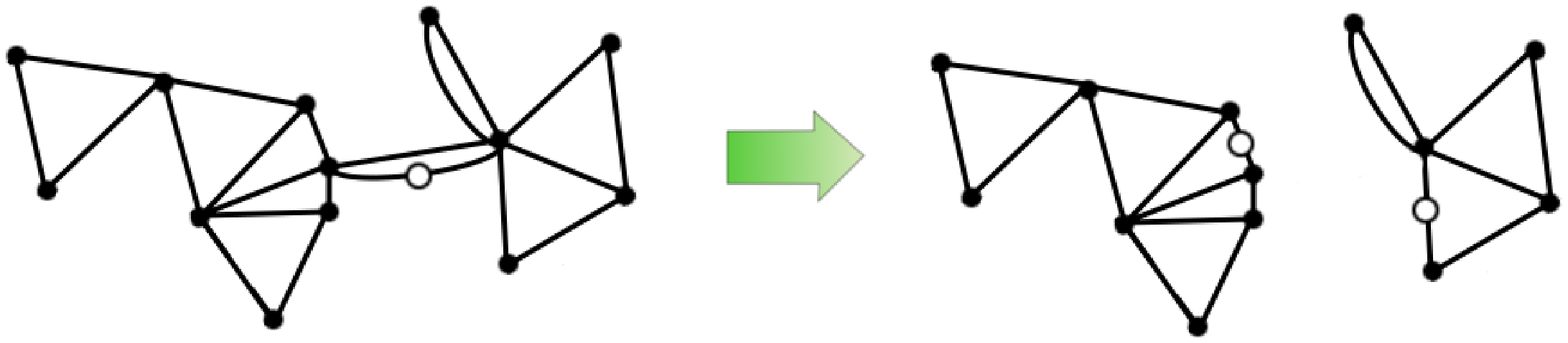}
    \caption{Move 2 relates an unrestricted triangulation (left), with a mark denoted by the white circle, to two disconnected unrestricted triangulations.}
    \label{unrestrictTriMove2}
  \end{center}
\end{figure}

Rather than trying to compute the combinatorics of the triangulations exactly it is useful to consider a slightly generalised class of triangulations known as unrestricted triangulations in which double links are allowed and triangles do not necessarily have to join along an edge. We also require one of the boundary edges to carry a mark. Such a triangulation is shown in fig \ref{unrestrictTriMove1}. Following the arguments of \cite{Carroll:1995nj} we will see that the continuum limit of these triangulations fall into the same universality class as the usual triangulations in DT and indeed give rise to the same loop equations. 

We can classify the graphs with $n$ triangles and $k$ boundary edges into two sets; those in which the marked point occurs on a link which forms a side of a triangle and those on which the marked point resides on one of a pair of double links. All graphs in the first class may be put into a one to one correspondence with graphs with $n-1$ triangles and $k+1$ edges. This is achieved by the map defined by adding double links to each side of the marked triangle which is an unmarked boundary and then removing the marked triangle. The mark is then transferred to the most anti-clockwise new boundary edge. This is shown in fig \ref{unrestrictTriMove1}.
 
All graphs in the second class may be put in one to one correspondence with pairs of triangulations which are obtained by cutting the triangulation along the marked double link and marking the new triangulations as shown in fig \ref{unrestrictTriMove2}.

This the means that the number of graphs with $n\geq 1$ triangles and $k \geq 1$ edges satisfies
\beq
\C{N}(n,k) = \C{N}(n-1,k+1) + \sum^{n}_{m = 0} \sum^{k-2}_{l=0} \C{N}(m,l)\C{N}(n-m,k-2-l),
\eeq
which when substituted into \eqref{wkdef} and \eqref{WMdef} yields the equation,
\beq
W_M^{(0)}(z)^2 +  (g z^2 - z) W_M^{(0)}(z) + 1 -g(w_1(g)+z) = 0.
\eeq
This equation is the same as \eqref{LEWM0} with a particular cubic potential. We therefore see that the combinatorial approach has yielded the correct loop equations. In the next chapter we will use the combinatorial approach to obtain an expression for quantities not easily computable in the matrix model approach, for instance the dimension of spacetime.

\section{Causal Dynamical Triangulation}

\begin{figure}[t]
  \begin{center}
    \includegraphics[width=6cm]{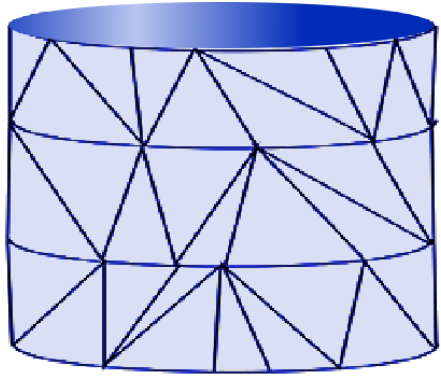}
    \caption{A causal triangulation.}
    \label{CDTtri}
  \end{center}
\end{figure}

Moving to two dimensional gravity alleviated a number of technical problems in the metric-is-fundamental approach, however the problem of what level of causality violations to allow remained. In the last section we considered the case when no causality conditions were enforced, technically this was because we summed over all Euclidean metrics before continuing back to a Lorentzian metric. We found that we could compute the partition function by considering a discretised version of the path integral. In this section we will consider the possibility of computing the Lorentzian path integral directly by means of triangulations that are inherently Lorentzian. This will also allow us to explore the role causality violations play in the theory.

\begin{figure}[t]
\centering 
\parbox{7cm}{
      \includegraphics[scale=0.6]{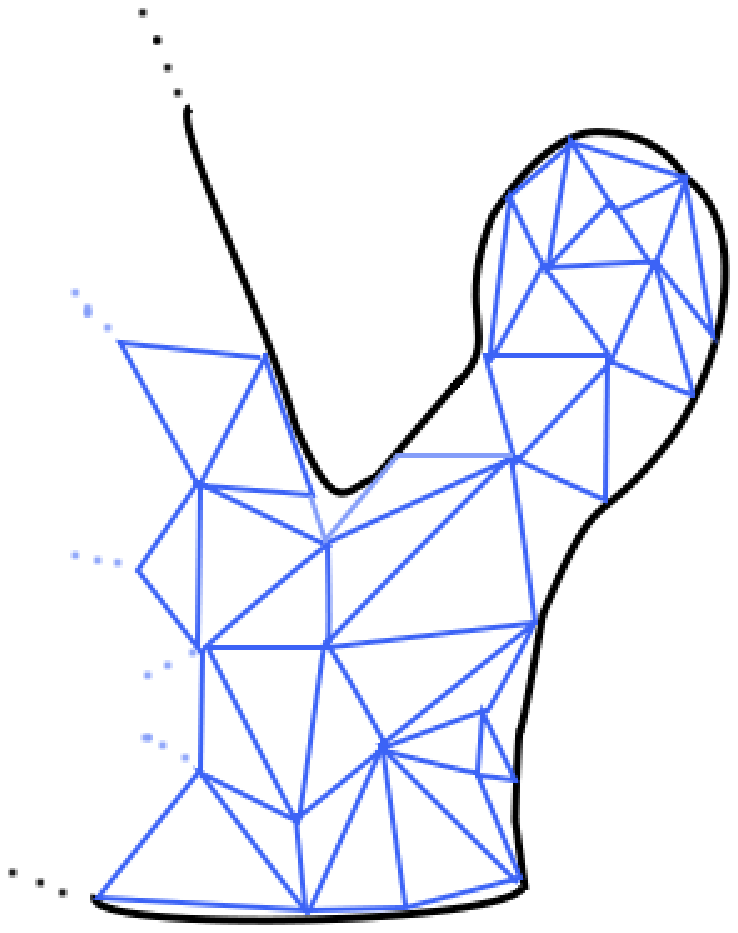}
\caption{A triangulation of a baby universe. If we successively remove triangles, then the final link between the baby universe and the parent will be a double link.}
\label{babyuniverse}}
\qquad 
\begin{minipage}{7cm}
      \includegraphics[scale=0.55]{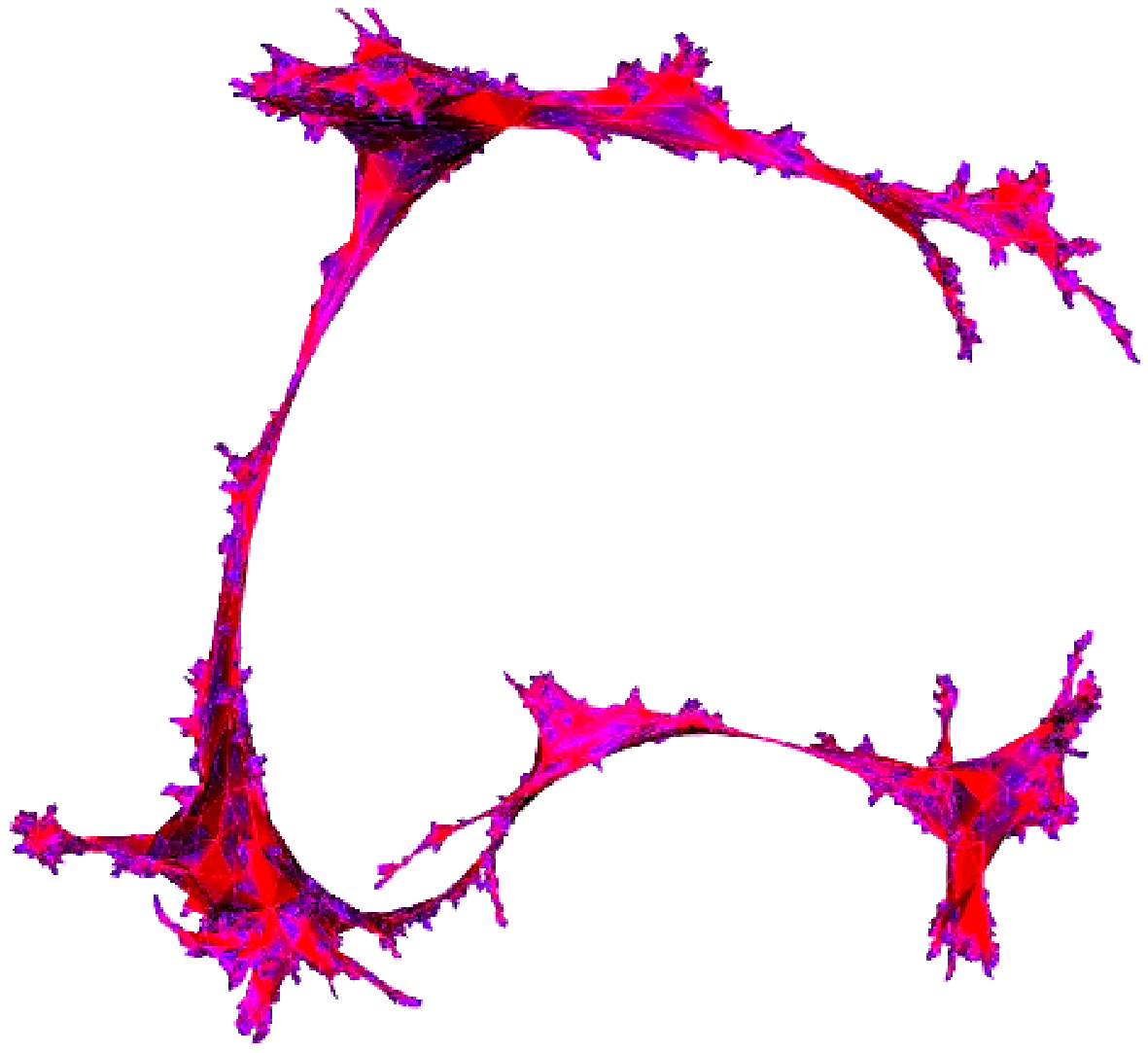}
\caption{A typical geometry in dynamical triangulation exhibiting the lack of a smooth continuum limit \cite{DTsim}.}%
\label{DTuniverse} 
\end{minipage} 
\end{figure}

Following \cite{Ambjorn:1998} we will construct a Lorentzian triangulation by requiring all triangles composing it to be flat Minkowski space, with one spacelike and two timelike edges. Furthermore we restrict the topology of spacetime to be cylindrical with spacelike boundaries. We define the set of triangulations that will be summed over by giving an algorithm for their construction. Suppose the initial boundary is composed of $l$ vertices joined by spacelike links, then we may attach future pointing timelike links to each of these vertices to obtain the set of vertices that, by joining them with spacelike links, form the next spatial slice. Furthermore, we require that each vertex has at least one timelike link ending on the same vertex as the rightmost timelike link of the previous vertex. This result of this procedure is shown in fig \ref{CDTtri}. 

Again the problem of computing a disc function can be solved using a matrix model technique \cite{Zohren:2009dj} or by the combinatorial approach. Although results from CDT will motivate later work in this thesis we will not have cause to compute any actual CDT observables. For this reason we will merely review the results of these computations and refer the interested reader to the literature for the technical details.

One can obtain a qualitative understanding of the difference between DT and CDT by considering which geometries are excluded from the path integral in CDT. From the construction for the causal triangulations given above we can see that there is no point at which the spatial slice can change topology. In terms of the combinatorial interpretation this means that the casual triangulations contain no double links. In the triangulations discussed in DT the double links were able to join two disconnected triangulations to form a larger one. If one were to view the evolution of a DT universe such a joining of two disconnected triangulations would correspond to part of the universe budding off from the main one to form a baby universe fig \ref{babyuniverse}. Together with topology changes, baby universe production is the only type of geometry excluded in CDT which is present in DT \cite{Ambjorn:1998}. 

\begin{figure}[t]
  \begin{center}
    \includegraphics[width=4cm]{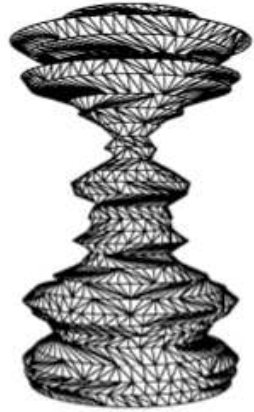}
    \caption{A typical geometry in CDT \cite{Ambjorn:2009rv}.}
    \label{CDTuniverse}
  \end{center}
\end{figure}

The difference between DT and CDT then depends on how much difference baby universe production makes to the resulting continuum theory. It can be shown that the effect of baby universe production is dramatic; it completely dominates the continuum limit of DT, with a baby universe forming at every point in the space \cite{Ambjorn:1998}. The resulting spacetime is highly non-smooth and fractal like as can be seen in fig \ref{DTuniverse}. In contrast in CDT there exist no baby universes in the continuum limit. This leads to spacetimes such as the one shown in fig \ref{CDTuniverse}. 

These result suggest we should consider some method by which the ``fractalness'' of space can be measured. In fact the observable we have yet to consider, the dimension of spacetime, is perfect for this role and will be discussed in detail in the next chapter. 

\chapter{The Dimension of Spacetime\label{ChapDim}}
In the previous chapter we introduced the matrix model method, by which the disc function and string susceptibility of DT may be computed. Of course there was one observable discussed in the previous chapter that we did not analyse using the matrix model approach; the dimension of the spacetime. This observable is in some sense the most interesting, as if we were a being living in such a universe the dimension is a far more accessible observable than the disc function.

Finding the dimensionality of the spacetime may seem trivial as we defined the theory using two-dimensional building blocks; indeed we even refer to the theory as two dimensional quantum gravity. However, we will see that the situation is more subtle than this. To give an idea of why this is the case one can consider the paths contributing to the propagator of a particle in the path integral formulation of quantum mechanics. Although all paths are built piecewise from one dimensional lines the entire path becomes ``fuzzy'' as we take the continuum limit. Is the resulting fuzzy path a one or two dimensional object?  One would be tempted to answer that it somehow has a dimension in-between these two values. This notion of fractional dimension can be made precise in a number of ways and in this chapter we will see how such concepts may be applied to the quantum gravity models defined in the previous chapter.

\section{The Hausdorff Dimension $d_H$}
One manner by which the dimensionality of a space can be defined is by observing how the volume of a ball scales with its radius. In smooth $n$-dimensional space we have $V(r) \propto r^n$, which motivates the definition of the Hausdorff dimension for a metric space. The Hausdorff dimension is defined such that if the Hausdorff dimension of the space is $d_H$, then for all points the following holds;
\beq
|B(r)| \sim r^{d_H} \qquad r \rightarrow 0,
\eeq
where $|B(r)|$ is the volume of the ball of radius $r$. In the models of quantum gravity we are considering we are not interested in the dimensionality of a single space but rather the expectation value for the dimensionality of spacetime. We therefore must consider, 
\beq
\label{dhDEF}
\avg{|B(r)|} \sim r^{d_H} \qquad r \rightarrow 0.
\eeq
The quantity $B(r)$ in the above equation may be computed for a given geometry i.e. metric, contributing to the path integral and the above statement is a statement about its mean behaviour after summing over all geometries. We will now review a number of ways that this quantity may be computed in DT and CDT models.

\subsection{Computing $d_H$ for DT and CDT}
We will begin by assuming that the Hausdorff dimension is not affected by the global topology of the spacetime; this allows one to consider a particular amplitude for which the calculation is easier. In order to compute the Hausdorff dimension we need to understand how to compute the average area of a spacetime of fixed topology in addition to restricting the geometries included in the amplitude to be of a certain ``radius''. To compute the mean area we need only differentiate with respect to the bulk cosmological constant, as this will insert the identity operator into the bulk. To enforce a fixed ``radius'' requires finer control of the computation than allowed by matrix models and therefore we must instead turn to the combinatorial approach. In order to make the notion of ``radius'' precise we define the geodesic distance between two vertices to be the number of links on the shortest path through the graph connecting the vertices. The distance of a point from a boundary is defined as the minimum geodesic distance between that point and a vertex on the boundary and the distance between two boundaries is defined if all points on one boundary have the same geodesic distance to the other boundary. 

The canonical amplitude to consider when computing the Hausdorff dimension is the two point correlation function on the sphere with the restriction that the geodesic distance between the operators is fixed to $t$ \cite{Watabiki:1993ym,Kawai:1993cj}. This means our computation differs slightly from the one that would be performed if we wanted to reproduce the definition of $d_H$ given in \eqref{dhDEF}. Instead we will compute how the total volume of the spacetime depends on the distance between two maximally separated points as the distance is made large.

Following \cite{Ambjorn:1995dg,Zohren:2009dj}, this amplitude may be obtained by considering the more general amplitude, $G(l_1,l_2,t)$, corresponding to a cylinder amplitude in which the geodesic distance between the entrance and exit loops, of length $l_1$ and $l_2$ respectively, is fixed to be $t$. The two point amplitude may then be obtained by considering the case when both entrance and exit loops shrink to zero size. Once we obtain the continuum two-point function $\tilde{G}(\tau,\mu)$ for operators separated by a distance $\tau$ we can then compute the Hausdorff dimension $d_H$, defined implicitly by,
\beq
\label{ComputedH}
\avg{V(\tau)} = \lim_{\tau \rightarrow \infty} -\frac{1}{\tilde{G}(\tau,\mu)}\frac{\partial \tilde{G}(\tau,\mu)}{\partial \mu} \sim \tau^{d_H},
\eeq
where $V(\tau)$ is the total volume of spacetime containing two points separated by a distance of $\tau$.
For triangulations containing $n$ triangles in which the entrance and exit loops are of length $l_1$ and $l_2$ respectively and are separated by a distance $t$ we may derive a recursion relation for the number, $\C{N}(n, l_1, l_2, t)$, of such triangulations by the same combinatorial methods used for the disc function \cite{Watabiki:1993ym}. There are two possible ways to construct such a triangulation from smaller triangulations depending on whether we add a new triangle or double link, as detailed in fig \ref{unrestrictTriMove1} and fig \ref{unrestrictTriMove2}. This leads to the recursion relation,
\beq
\label{recursiveN2loop}
\C{N}(n, l_1, l_2, t) = \C{N}(n-1, l_1+1, l_2, t) +2\sum^n_{m=0}\sum^{l_1-2}_{k=0} \C{N}(m, k, t) \C{N}(n-m, l_1-2-k, l_2, t),
\eeq
where the factor of two in the final sum is due to the fact that the exit loop may appear in either of the two disjoint triangulations being joined. In terms of $\C{N}(n, l_1, l_2, t)$ the two loop amplitude $G(l_1,l_2,t)$ is given by,
\beq
G(l_1,l_2,t) = \sum^\infty_{n=0}\C{N}(n, l_1, l_2, t) g^n.
\eeq
Consider a graph with $n$ triangles and $l_1$ boundary links in the entrance loop. If we were to add $l_1$ more triangles to the entrance loop evenly along the boundary, thereby adding a new layer to the triangulation, we would have increased the distance separating the entrance and exit loops by one. Therefore, if we were to add only a single triangle to the graph then we would have increased $t$ by roughly $1/l_1$. Combining this with \eqref{recursiveN2loop}, leads to,
\bea
G(l_1,l_2,t+1/l_1) &=& G(l_1,l_2,t) + \frac{1}{l_1} \partial_t G(l_1,l_2,t)\\ 
&=& g G(l_1 + 1,l_2,t) + 2 \sum^\infty_{n=0} w_n(g) G(l_1-n-2,l_2,t).
\eea
Finally, if we now introduce the 2-loop resolvent defined as,
\beq
W_{M;M}(z_1,z_2,t) = \sum^\infty_{n=0}\sum^\infty_{l_1=0}\sum^\infty_{l_2=0} \C{N}(n, l_1, l_2, t) g^n z_1^{-l_1-1}z_2^{-l_1-1},
\eeq
we have
\beq
\label{WMMteqn}
\frac{\partial}{\partial t} W_{M;M}(z_1,z_2,t) = \frac{\partial}{\partial z_1}\left[ (z_1-g z_1^2 - 2 W_M^{(0)}(z_1))W_{M;M}(z_1,z_2,t)\right],
\eeq
where $W_M(z_1)$ is the disc function as obtain using the combinatorial approach in the last chapter. To find the appropriate initial condition for $W_{M;M}(z_1,z_2,t)$ we first consider the initial condition for $G(l_1, l_2, t)$. We require $G(l_1, l_2, 0) = \delta_{l_1 l_2}$, which implies $W_{M;M}(z_1,z_2,0) = (z_1 z_2-1)^{-1}$. To obtain a equation for continuum quantities we must take a scaling limit. This requires us to decide how $t$ will scale in addition to the second boundary parameter $z_2$. One can see that in order for the initial condition to have a non-trivial scaling limit we must set $z_1 = z_c \exp(a\zeta_1)$, as usual, but for $z_2$ we set $z_2 = z_c^{-1}\exp(a\zeta_2)$. This gives,
\beq
\label{WMMinitcond}
\tilde{W}_{M;M}(\zeta_1,\zeta_2,0) = \frac{1}{\zeta_1+\zeta_2},
\eeq
where we have set $W_{M;M}(z_1,z_2,t) = a^{-1} \tilde{W}_{M;M}(\zeta_1,\zeta_2,\tau)$ in order for the scaling dimension of both sides of the equation to agree. Using this scaling behaviour for $z_1$, $z_2$ and $W_{M;M}$ in \eqref{WMMteqn} we find that the only non-trivial limit occurs when $t = a^{-1/2} \tau$, resulting in the equation,
\beq
\partial_\tau \tilde{W}_{M;M}(\zeta_1, \zeta_2, \tau) = -\partial_{\zeta_1} \left( \tilde{W}_M^{(0)}(\zeta) \tilde{W}_{M;M}(\zeta_1, \zeta_2, \tau)\right).
\eeq
This equation can be solved for $\tilde{W}_{M;M}(\zeta_1,\zeta_2,\tau)$, which together with the initial condition \eqref{WMMinitcond} gives, 
\beq
\tilde{W}_{M;M}(\zeta_1,\zeta_2,\tau) = \frac{1}{\zeta_1+\zeta_2}\frac{W_M^{(0)}(f(\zeta_1,\tau))}{W_M^{(0)}(\zeta_1)},
\eeq
where $f$ satisfies,
\beq
\frac{d f(\zeta,\tau)}{d \tau} = -W_M(f(\zeta,\tau)).
\eeq
The boundaries may be shrunk to points by taking the two cosmological constants to infinity and extracting the leading order coefficient, giving the result,
\beq
\tilde{G}(\tau,\mu) = \left(\alpha \mu\right)^{3/4} \frac{\cosh\left[\left(\alpha \mu\right)^{1/4} \tau/2\right]}{\sinh^3\left[\left(\alpha \mu \right)^{1/4}\tau/2\right]},
\eeq
where $\alpha$ is an unimportant numerical constant. Using the above equation in \eqref{ComputedH} gives $d_H = 4$ \footnote{A naive application of \eqref{ComputedH} will not in fact produce $d_H = 4$ due the presence of a separate length scale in the problem; that associated with the cosmological constant. A more transparent calculation of $d_H$ would involve determining the maximal distance between points as a function of $\mu$ and then eliminating $\mu$ from an expression of the unconstrained total volume. In our case we simply note that from dimensional analysis $\mu \sim \tau^{-4}$ yielding the stated result for $d_H$.}. This is quite a striking result as it shows the effective dimension of the space is far from what might be expected to arise from two dimensional building blocks. We should also emphasise that the above result does not depend on considering a long or short distance limit and therefore that $d_H = 4$ on all scales. 

\begin{figure}[t]
  \begin{center}
    \includegraphics[width=6cm]{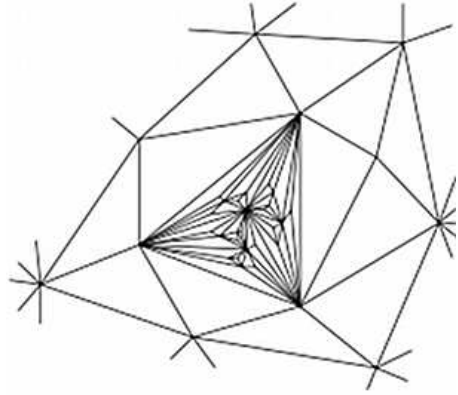}
    \caption{By drawing a baby universe in the lattice we see how this gives rise to a fractal structure \cite{Zohren:2009dj}.}
    \label{planarbabyuniverse}
  \end{center}
\end{figure}

One can avoid some of the above calculation and obtain the Hausdorff dimension via more qualitative means by noting that once the scaling dimension of the geodesic distance has been decided, then, given that we expect a power law relation between the geodesic distance and area, we can obtain the exponent from dimensional analysis. In the above computation we found that $\tau = a^{1/2} t$ and since area has dimensions of $a^2$ this means $A\sim \tau^4$. It is interesting to note that the scaling relation $\tau = a^{1/2} t$ indicates that each link traversed has an effective size of $a^{1/2}$. This is reminiscent of a random walk and indicates that the internal links of the graph form a highly random lattice, such as that shown in fig \ref{planarbabyuniverse}. Intuitively, this random structure is caused by there being baby universes everywhere on the lattice, meaning any path through the lattice must pass through many of them. This insight allows one to give a reasonable conjecture for the Hausdorff dimension in CDT. If one considers the triangulations arising in CDT then due to the sliced structure we do not have any baby universes and we expect that the shortest path through the triangulation will scale in the same way as the boundary length, i.e. with dimensions of $a$. We therefore expect that the area and geodesic distance will be related by $A \sim \tau^2$, giving $d_H = 2$. This qualitative argument is confirmed by an exact computation using combinatorial methods \cite{Ambjorn:1998,Zohren:2009dj}.


\section{The Spectral Dimension $d_s$}
We saw in the last chapter that due to baby universe production we might intuitively expect to obtain a continuum theory from DT that describes the dynamics of a fractal spacetime. The last section gave concrete evidence for this in that the mean Hausdorff dimension of the continuum spacetime was four despite the lattice theory being defined using two dimensional simplexes. However, the fact that the Hausdorff dimension is exactly four does raise the question of how we can be sure that the continuum theory is indeed fractal; perhaps it is a smooth four dimensional manifold. This question arises because the Hausdorff dimension is in fact not that sensitive to the fractal structure of a space. One way to see this is to note it is the same for different graphs with greatly differing structure, for example, the square lattice and the comb graph obtained by removing all its horizontal links except one. The difference is due to the decreased connectivity of the comb graph.

In this section we will introduce another method by which the dimension of a space can be defined based on the properties of a propagating particle. This will allow us to distinguish the DT continuum limit from a smooth manifold. 

Consider the rate at which the force of an electric charge decreases with distance, it clearly depends on dimension. This is related to the form the particle propagator takes in the corresponding field theory. Both of these things are determined by the Laplacian operator of the spacetime on which they are defined and so one would expect that the dimension of the space could be extracted from the Laplacian. A simple way to do this is to look at some dynamical process whose equations of motion involves the Laplacian and look for properties of its solutions that are dependent on the dimension of the manifold. If the dynamical process can be defined on spaces other than manifolds then these dimension dependent observables can be used to define the dimension for the space.

The usual choice for the dynamical process is that of heat diffusion. We define the heat kernel $K(x,x';\sigma)$ by,
\beq
\left(\partial_\sigma - \Delta_x\right) K(x,x';\sigma) = 0,
\eeq
together with the boundary condition $\sigma \geq 0$ and $ K(x,x';0) = \delta(x-x')$. On a smooth $d$ dimensional manifold the heat kernel has the property that as $\sigma \rightarrow 0$,
\beq
K(x,x;\sigma) \sim \frac{1}{(4\pi \sigma)^{d/2}},
\eeq
from which we can read off the dimension of the manifold. The dimension computed in this way is referred to as the spectral dimension of the space. 

The question of generalising the spectral dimension to less smooth spaces now amounts to generalising the process of heat diffusion. This is in fact why the heat diffusion process was chosen; qualitatively, the heat kernel $K(x,x';\sigma)$ can be thought of as a ``scaling limit'' of the probability of a random walk moving between the points $x$ and $x'$ in $t$ steps. By a scaling limit we mean a process in which the number of steps in the walk is taken to infinity while reducing the size of each step. 

The process of a random walk is obviously something that can be naturally defined on a graph, which allows the spectral dimension to be defined for each triangulation contributing to the path integral. In particular, for triangulations, we need to compute the return-probability $p(t)$ for the random walk to leave a given point and return to it in $t$ steps. In the scaling limit this quantity will become $K(x,x;\sigma)$ and we then can study its behaviour as $\sigma \rightarrow 0$ to find the spectral dimension.

Obviously the above procedure depends on giving a precise definition of the scaling limit. In fact we may extend the definition of the spectral dimension to discrete structures, without taking a scaling limit, by considering the $t \rightarrow \infty$ behaviour of $p(t)$. The reason for studying the $t \rightarrow \infty$ limit in the discrete case can be understood from considering the process of taking the scaling limit, in which the link length of the discrete space is reduced to zero. A walk whose length is of finite size measured in link lengths would produce a measurement of the spectral dimension contaminated by cutoff effects. In contrast, a short walk in the continuum space will still correspond to a walk whose length is much greater than the link length. We therefore expect the large $t$ behaviour of $p(t)$ to fall off with the same exponent as the small $\sigma$ behaviour of $K(x,x;\sigma)$. For discrete spaces we therefore {\emph{define}} the spectral dimension by,
\beq
\label{tdependence}
d_s=-2 \lim_{t\to\infty} \frac{\log(p(t))}{\log t},
\eeq
if the limit exists.

As we saw with the computation of the Hausdorff dimension, the matrix model approach has difficulty computing quantities that are not disc-functions or their generalisations. Instead, we had to augment the computation with a combinatorial approach. This problem is in fact worse for computing the spectral dimension; this is not surprising as the spectral dimension relates to a much more complicated phenomenon than the Hausdorff dimension, which simply measures volume. In DT there exists an approach to computing the spectral dimension, based on the continuum formulation \cite{Ambjorn:1997jf}, although we will not review this approach in detail here. The result of the computation in \cite{Ambjorn:1997jf} is that the spectral dimension in DT is two. Because the Hausdorff dimension differs from the spectral dimension, this confirms our expectation that the continuum limit in two dimensional DT is not describing a smooth manifold but a fractal space.

When it comes to CDT we are mostly stuck for the moment with merely exploring the spectral dimension numerically. The result of such numerical simulations is that in $d$-dimensional CDT (where $d=3,4$ has been considered) the spectral dimension is consistent with $d_s = d$ \cite{Benedetti:2009ge}. For $d_s = 2$ the situation is not so problematic and there has been progress, which we will review shortly, in computing the spectral dimension analytically. The result of such work suggests that unlike DT, the long distance behaviour of the resulting spacetime is that of a smooth $d$-dimensional manifold.

 

\subsection{Dimensional Reduction}
In addition to providing a measure of the dimension of the space, more information can be extracted from the analysis of a random walk by controlling the walk length. Since the size of the space explored by a random walk of $t$ steps goes as $\sqrt{t}$ we can probe shorter distances by reducing the length of the walk. We therefore can gain some insight into the Planck scale behaviour of gravity in CDT and other models. The effective spectral dimension as a function of the length probe has been investigated numerically in three and four dimensional CDT with the surprising result that although the dimension is three and four respectively on long distance scales, on very short scales the spectral dimension in both cases is measured to be two \cite{Ambjorn:2005aa}. This phenomenon is known as dimensional reduction.

The phenomenon of dimensional reduction of the spectral dimension is compelling, as similar behaviour has been observed in a variety of other approaches to quantum gravity in four dimensions such as the non-Gaussian fixed point possibility in the perturbation-is-fundamental program \cite{Lauscher} and also Horava-Lifshitz gravity \cite{Horava:2009aa,Horava:2009ab}. Other evidence has been discussed in \cite{benedetti,Carlip:2009aa}. This is very suggestive of dimensional reduction being a robust feature of quantum gravity. It will be our task in the next chapter to attempt to construct a toy model in which dimensional reduction of the spectral dimension occurs. 

\subsection{The Spectral Dimension in CDT}
The triangulations in DT and CDT are formulated mathematically as a set of graphs together with a measure on that set. We will refer to a set of graphs together with a measure as a random graph or an ensemble of graphs. To construct a toy model exhibiting dimensional reduction we will want to consider a random graph in which a scaling limit can be taken. Furthermore, we will want to keep the ensemble of graphs as closely related to the CDT graphs as possible. Before doing this it is worth pausing to draw inspiration from CDT. We will briefly review one attempt to derive the spectral dimension in two dimensional CDT. This will provide clues at to what sort of graphs should be used to construct our ensemble.

As previously discussed, matrix model and combinatorial methods are not appropriate for computing the spectral dimension. An alternative approach is to directly analyse the random walk. One method that proves useful in this approach is to encode the first-return-probability for a random walk in a generating function \cite{Durhuus:2005fq,Durhuus:2006vk,Durhuus:2009sm}. By deriving sufficiently tight bounds on this generating function we can obtain an exact value for the spectral dimension. This method has the advantage that, when it works, it can provide stronger statements than merely averages; it can show that a particular value for the spectral dimension holds occurs with probability one \cite{Durhuus:2006vk,Durhuus:2009sm}. It also has the advantage of being entirely mathematically rigourous. 

\begin{figure}[t]
\centering 
\parbox{7cm}{
\includegraphics[scale=0.5]{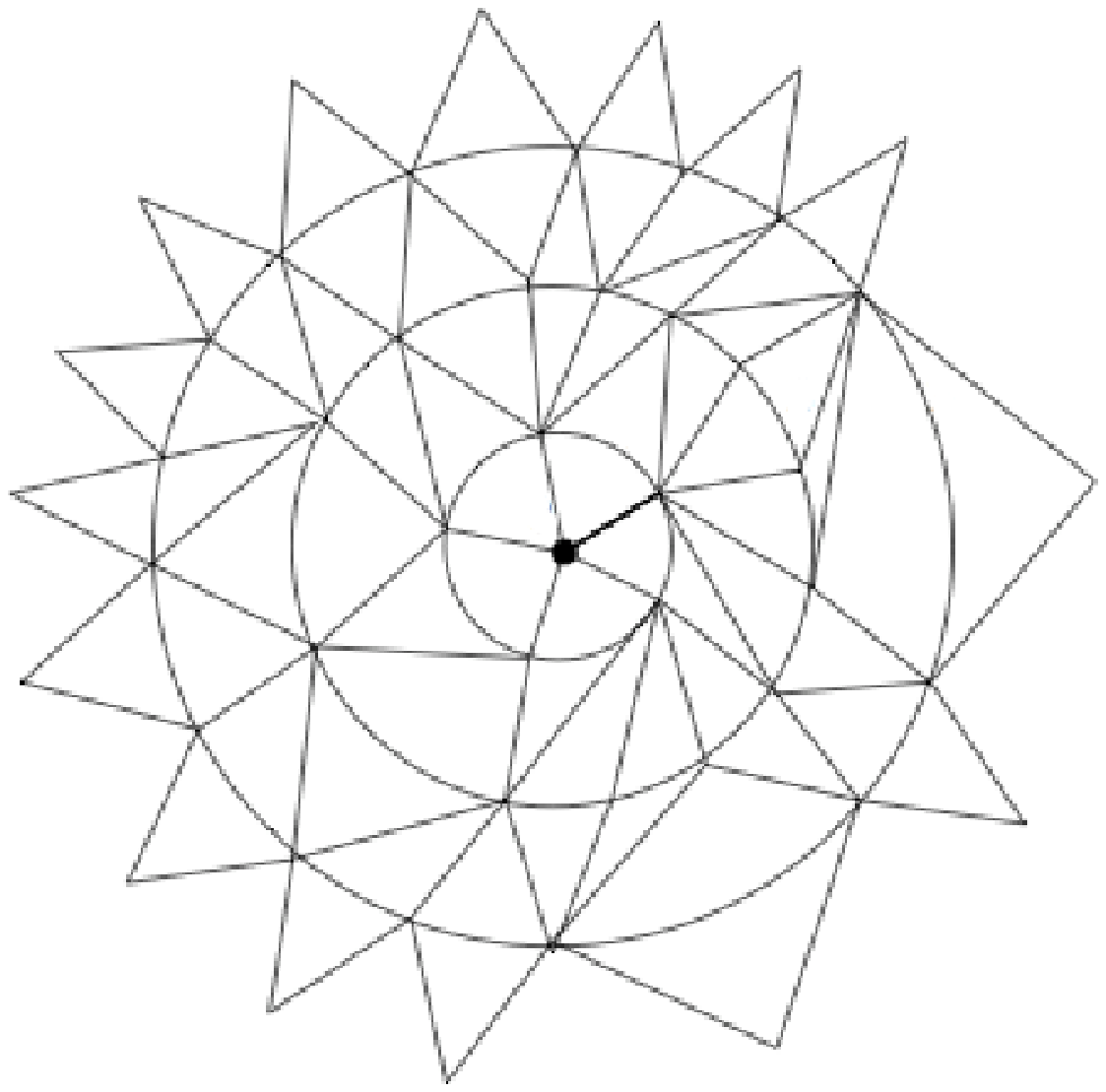}
\caption{A causal triangulation drawn in the plane. Note we have added an extra vertex with links to all vertices in the inner most spatial slice \cite{Durhuus:2009sm}.}
\label{CDTbijection1}}
\qquad 
\begin{minipage}{7cm}
\includegraphics[scale=0.5]{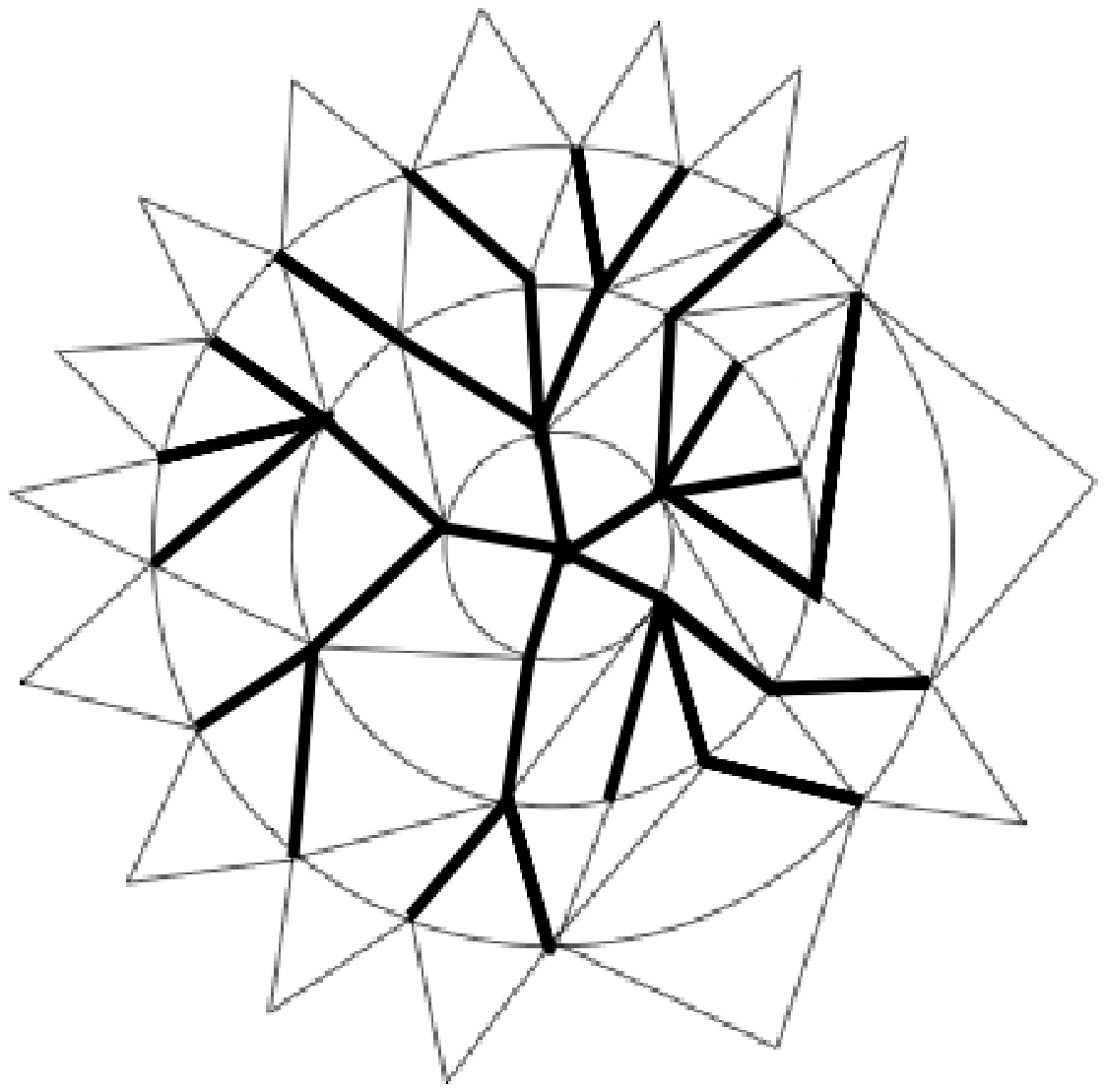}
\caption{Casual triangulations have a bijection with trees. In this case the tree associated to the triangulation is shown by thick lines. \cite{Durhuus:2009sm}.}
\label{CDTbijection2} 
\end{minipage} 
\end{figure}

To date it has proven too difficult to carry out the above strategy directly on the triangulation. Instead attempts have focused on considering ensembles of simpler graphs for two reasons; firstly, because they serve as toy examples in which to develop the necessary techniques and secondly because a number of simpler graph ensembles can be obtained from CDT while preserving some of its properties. In fact it can be shown that certain simpler graphs have spectral dimensions that bound the spectral dimension of CDT. In particular the graphs obtained by collapsing all vertices of a given height to a single vertex have been used to show that $d_s \leq 2$ almost surely for CDTs \cite{Durhuus:2009sm}. On the other hand one can relate a causal triangulation to a tree graph $T$ by the following procedure \cite{Durhuus:2009sm}. First we recognise that casual triangulations are planar graphs and so can be represented as lying in the plane with the spacelike slices forming concentric circles as shown in fig \ref{CDTbijection1}. We denote the vertices in the $k$th spacelike slice as $S_k$, we define $S_0$ as containing a single new vertex connected to all vertices in $S_1$. The algorithm then maps the triangulation into a tree $T$:
\begin{itemize}
\item[1.] All vertices in the triangulation are in $T$ in addition to a new root vertex that is only linked to $S_0$.
\item[2.] All links from $S_0$ to $S_1$ are in $T$.
\item[3.] All links from a vertex in $S_k$ to vertices in $S_{k+1}$, apart from the clockwise-most link, are in $T$.
\end{itemize}
This algorithm is shown in fig \ref{CDTbijection2}. In fact it can be shown that the map produced by this algorithm is a bijection and furthermore preserves the measure on the CDTs. The spectral dimension of the resulting random tree has long been known to be $4/3$ and this gives a lower bound on the spectral dimension of the CDTs. One tool to obtain results about such trees is another random graph known as a comb. It is this random graph that we will use in the next chapter to construct a toy example of dimensional reduction.



\chapter{A Toy Model of Dimensional Reduction\label{ChapComb}}
\newcommand\Qbar{\bar{Q}}

In the last chapter we saw that both DT and CDT exhibit fractal behaviour in some regimes of the theory. However DT appears to be far more pathological as fractal spacetimes dominate the continuum limit. Furthermore we saw that the situation in CDT is much improved, with both Hausdorff and spectral dimension agreeing and agreeing with the dimension, on long distance scales, of the underlying simplexes. At short distances however it exhibits a reduction in the spacetime dimension; a phenomenon that has been observed in other approaches to quantum gravity and has been given the name dimensional reduction. 

In this chapter we will construct a definition for a scale dependent spectral dimension and show that there are models which do indeed exhibit scale dependent spectral dimensions defined in this way. In particular we develop a simple model based on previous work on random combs \cite{Durhuus:2005fq}. These are a family of simple geometrical models which share some of the properties of the CDT model; instead of an ensemble of triangulations we have an ensemble of graphs consisting of an infinite spine with teeth of identically independently distributed length hanging off (we define these graphs precisely in Section 2). It was shown in \cite{Durhuus:2005fq} that the spectral dimension is determined by the probability distribution for the length of the teeth. In this chapter we show that it is possible to extend the work of \cite{Durhuus:2005fq} by taking a continuum limit thus ensuring that the cut-off scale is much shorter than all physical distance scales. We find that the spectral dimension is one if we take the physical distance explored by the random walk to zero and there exists a number of continuum limits in which the long distance spectral dimension differs from its short distance counterpart. As a by-product of this work we also extend some of the proofs given in \cite{Durhuus:2005fq} to a wider class of probability distributions.						

This chapter is organized as follows. In Section \ref{combdefs} we briefly review some known results for combs and their spectral dimension and then explain how in principle these can be extended to show different spectral dimensions at long and short distance scales. In Section \ref{SimpleExample} we introduce a simple model which we prove does in fact exhibit a spectral dimension that is different in the UV and IR. This model forms the basis of all later generalisations. In Section \ref{PowerComb} we generalise the results of Section \ref{SimpleExample} to combs in which teeth of any length may appear with a probability governed by a power law. In Section \ref{multiscale} we examine the possibility of intermediate scales in which the spectral dimension differs from both its UV and IR values. In Section \ref{genericcomb} we analyse the case of a comb in which the tooth lengths are controlled by an arbitrary probability distribution and show that continuum limits exist in which the short distance spectral dimension is one while the long distance spectral dimension can assume values in one-to-one correspondence with the positions of the real poles of the Dirichlet series generating function for the probability distribution. We then show how these techniques can be used to extend the results of \cite{Durhuus:2005fq}. In Section \ref{DiscussComb} we discuss our results and possible directions for future work. 

\section{Combs and Walks}
\label{combdefs}

In this section we review some basic facts about random combs and random walks. As much as possible we use the same notation and conventions as \cite{Durhuus:2005fq} and refer to that paper for proofs omitted here. 

\subsection{Definitions}

We use the definition of a comb given in \cite{Durhuus:2005fq}.  Consider the nonnegative integers regarded as a graph, which we denote $N_\infty$, so that $n$
has the neighbours $n\pm 1$ except for $0$ which only has $1$ as a neighbour.
Furthermore, let $N_\ell$ be the integers $0,1,\ldots ,\ell$ regarded as a graph so that
each integer $n\in N_\ell$ has two neighbours $n\pm 1$ except for $0$ and $\ell$
which only have one neighbour, $1$ and $\ell-1$, respectively. A comb $C$ is an infinite rooted tree-graph with a special subgraph $S$ called the spine which is isomorphic to $N_\infty$ with the root at $0$.  At each vertex of
$S$, except the root $0$,  there is attached one of the graphs $N_\ell$ or
$N_\infty$.  We adopt the convention that these linear graphs which are
glued to the spine are attached at their endpoint $0$.  The linear graphs
attached to the spine are called the teeth of the comb, see figure \ref{fig1}.
We will find it convenient to say that a vertex on the spine with no tooth
has a tooth of length $0$.  We will denote by $T_n$ the tooth attached to
the vertex $n$ on $S$, and by $C_k$ the comb obtained by removing the links 
$(0,1),\ldots ,(k-1,k)$, the teeth $T_1,\ldots ,T_k$ and relabelling the 
remaining vertices on the spine in the obvious way. The number of nearest neighbours of a vertex $v$ will be denoted $\sigma(v)$.

It is convenient to give names to some special combs which occur frequently. We denote by $C=*$  the full comb in which every vertex on the spine is attached to an infinite tooth, and by $C=\infty$  the empty comb in which the spine has no teeth (so an infinite tooth is itself an example of $C=\infty$).

\begin{figure}[t]
  \begin{center}
    \includegraphics[width=10cm]{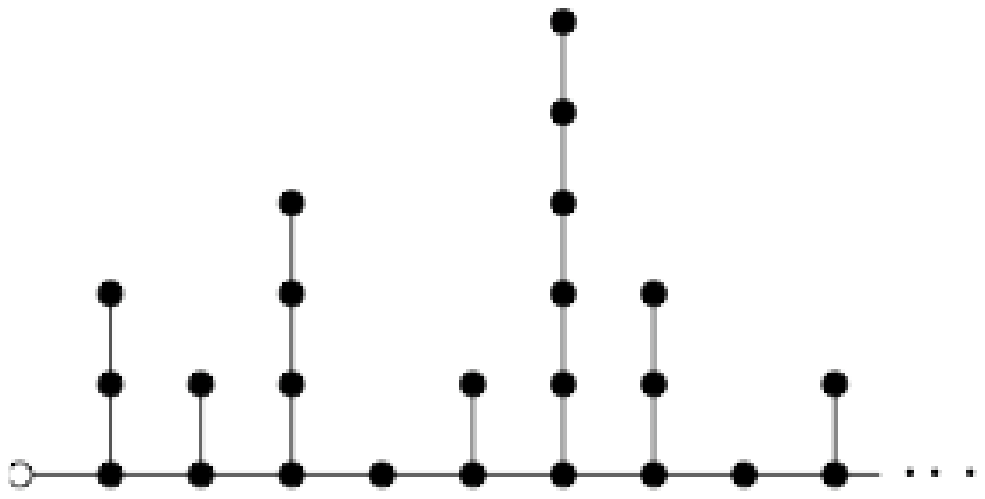}
    \caption{A comb.}
    \label{fig1}
  \end{center}
\end{figure}

Now let $\C{C}$ denote the collection of all combs and define a probability measure $\nu$ on $\C{C}$ by letting the length of the teeth be identically and independently distributed by the measure $\mu$.  We will refer to the set $\C{C}$ equipped with the probability measure $\nu$
as a random comb. Measurable subsets $\C{A}$ of $\C{C}$ are called {\it events} and $\nu (\C{A})$ is the probability of the event $\C{A}$.
The measure of the set of combs $\C{A}$ with teeth at $n_1, n_2,\ldots ,n_k$ having lengths $\ell_1,\ell_2,\ldots ,\ell_k$ is 
\beq
\nu (\C{A})=\prod_{j=1}^k \mu (\ell_j).
\eeq
For any $\nu$-integrable function $F$ defined on $\C{C}$ we define the expectation value \footnote{Since the space of combs is discrete, the integration should be replaced by a sum. However, we will continue to use an integral sign, with the understanding that it refers to a sum over discrete structures.}
\beq
\langle F(C)\rangle = \int F(C)\,d\nu.
\eeq
We will often use the shorthand $\bar F$ for $\langle F(C)\rangle$.

\subsection{Random Walks}
We consider simple random walks on the comb $C$ and count the time $t$ in integer steps. At each time step the walker moves from its present location  at vertex $v$  to one of the neighbours of $v$ chosen with equal probabilities $\sigma(v)^{-1}$.  Unless otherwise stated  the walker
always starts at the root at time $t=0$.  

The generating function for the probability $p_C(t)$ that the walker is at the        
root at time $t$, having left it  at $t=0$, is defined by
\beq
Q_C(x)=\sum_{t=0}^\infty (1-x)^{t/2}p_C(t)
\eeq
and we denote by $P_C(x)$ the corresponding generating function for the probability that the walker returns to the root for the \emph{first} time, excluding the trivial walk of length 0. Since walks returning to the root   can be decomposed  into walks returning for the 1st, 2nd etc time we have
\beq\label{QPrelation}
Q_C(x) = \frac{1}{1-P_C(x)}.
\eeq
It is convenient to consider contributions to $P_C(x)$ and $Q_C(x)$ from walks which are  restricted. 
Let $P_C^{(n)}(x)$ denote the contribution to $P_C(x)$ from walks whose maximal distance along the spine from the root is $n$ and  define 
\beq P^{(<n)}_C(x)=\sum_{k=0}^{n-1} P_C^{(k)}(x)\eeq
which is the contribution from all walks which do not  reach the point $n$ on the spine. Similarly we define
\beq P^{(>n-1)}_C(x)=\sum_{k=n}^{\infty} P_C^{(k)}(x).\eeq
Clearly $P_C(x)$ can be recovered from $P^{(<n)}_C(x)$ by setting $n\to\infty$. We define the corresponding restricted contributions to $Q_C(x)$ in the same way.
By decomposing walks contributing to  $P^{(<n)}_C(x)$ into a step to $1$, walks returning to $1$ without visiting the root, and finally a step back to the root it is straightforward to show that
\beq\label{Precurrence} P^{(<n)}_C(x)=\frac{1-x}{3-P_{T_1}(x)-P^{(<n-1)}_{C_1}(x)},\eeq
where we have adopted the convention that for the empty tooth, $T=\emptyset$,
\beq P_\emptyset(x)=1.\eeq
The relation \eqref{Precurrence} can be used to compute the generating function explicitly for any comb with a simple periodic structure and we list some standard results in \ref{StandardResults}.

There are a number of elementary lemmas which characterise the dependence of $P_C(x)$ on the length of the teeth and the spacing between them \cite{Durhuus:2005fq}. We state them here in a slightly generalized form which is useful for our subsequent manipulations.
\begin{lemma}
\label{MonoLem1}
The function $P^{(<n)}_C(x)$ is a monotonic increasing function of $P_{T_k}(x)$ and $P^{(<n-k)}_{C_{k}}(x)$ for any $n> k\ge 1$.
\end{lemma}
\begin{lemma}
\label{MonoLem2}
$P^{(<n)}_C(x)$ is a decreasing function of the length, $\ell_k$, of the tooth $T_k$ for any $n>k\ge 1$.
\end{lemma}
\begin{lemma}
\label{RearrangeLem1}
Let $C'$ be the comb obtained from $C$ by swapping the teeth $T_k$ and $T_{k+1}$, $k<n-1$. Then $P^{(<n)}_C(x)>P^{(<n)}_{C'}(x)$ if and only if $ P_{T_k}(x)> P_{T_{k+1}}(x)$.
\end{lemma}
The proofs use \eqref{Precurrence} and follow those given in \cite{Durhuus:2005fq} for the case $n=\infty$.

An important corollary, valid for any comb, of these lemmas is that
\bea \label{Qbounds} x^{-\quarter}\leq Q_C(x)\leq x^{-\half},\eea
which we will refer to as the trivial upper and lower bounds on $Q_C(x)$. The result follows from Lemma \ref{MonoLem2} with $n=\infty$, which gives
\bea
P_*(x) \leq P_C(x)\leq P_\infty(x),
\eea
 and the explicit expressions for  $P_*(x)$ and $P_\infty(x)$ given in \ref{StandardResults}.

\subsection{Two point functions}
Two point correlation functions on the comb correspond to the probability of a walk beginning at the root being at  a particular vertex on the spine at time $t$. In particular, let $p_C(t;n)$ denote the probability that a random walk that starts at the root at time zero is at the vertex $n$ on the spine at time $t$ having not visited the root in the intervening period. We will refer to the generating function for these probabilities as the two point function,  $G_C(x;n)$,  and define it by
\bea
G_C(x;n) = \sum^{\infty}_{t=1} (1-x)^{t/2} p_C(t;n).
\eea
$G_C(x;n)$ may be expressed as 
\bea
G_C(x;n)=\sigma(n)(1-x)^{-n/2}\prod_{k=0}^{n-1}P_{C_k}(x)
\eea
which may be used in conjunction with Lemma  \ref{MonoLem2} 
 to obtain the bounds,
\bea
\frac{G_*(x;n)}{3} \leq \frac{G_C(x;n)}{\sigma(n)}\leq \frac{G_\infty(x;n)}{2}.
\eea
Now let $r_C(t;n)$ denote the probability that a random walk that starts at the root at time zero is at the vertex $n$ on the spine for the first time at time $t$ having not visited the root in the intervening time. We define the modified two point function, $G^{0}_C(x;n)$, by, 
\bea
G^{0}_C(x;n) = \sum^{\infty}_{t=1} (1-x)^{t/2} r_C(t;n)
\eea
and note the following lemmas; 
\begin{lemma}
\label{P2bound}
The contribution $P_C^{(>N)}(x)$ to $P_C(x)$ from  walks whose maximal distance from the root is $N$ or greater satisfies 
\beq
P_C^{(>N-1)}(x) \leq 3x^{-1/2} G_C^0(x;N)^2.    
\eeq
\end{lemma}
The proof is given in section 2.4 of \cite{Durhuus:2005fq}.  
\begin{lemma}
\label{Mod2ptBounds}
The modified two point function satisfies 
\bea
G^{(0)}_{*}(x;n)\leq G^{(0)}_{C}(x;n) \leq G^{(0)}_{\infty}(x;n).
\eea
\end{lemma}
To prove this note that 
\bea G^{(0)}_{C}(x;n)=\frac{(1-x)^{-(n-2)/2}}{\sigma(n-1)}\prod_{k=0}^{n-2}P^{(<n-k)}_{C_k}(x)
\eea 
and use Lemma  \ref{MonoLem2} .

\subsection{Spectral dimension and the continuum limit}
\label{dsdiscussion}
As was discussed in the previous chapter, we may characterise the dimension of spaces which are not manifolds by use of the spectral dimension. For spaces obtained as the continuum limit of a graph or ensemble of graphs the probability that a random walk returns to its initial position generalises the heart kernel. However, given that it is easier to work with the generating function for the return probability, $Q(x)$, it is more convenient to define the spectral dimension in terms of the behaviour of $Q(x)$ implied by \eqref{tdependence},
\beq Q_C(x)\sim x^{-1+d_s/2},\label{QCdsdef},\eeq
where by $f(x)\sim g(x)$ we mean that
\beq c g(x)\le f(x)\le c' g(x),\quad 0<x<x_0,\eeq
where $c$, $c'$ and $x_0$ are positive constants. The property \eqref{QCdsdef} was adopted in \cite{Durhuus:2005fq} as the definition of spectral dimension, assuming it exists. The spectral dimension of an ensemble average is defined in the same way, simply replacing $p_C(t)$ and $Q_C(x)$ by their respective expectation values.

Our present goal is to extend the definition of the spectral dimension to give a mathematical meaning to the notion of a scale dependent spectral dimension as required to describe the phenomenon of dimensional reduction. In order for the spectral dimension to be scale dependent there must exist a length scale $L$ in the system besides the cutoff. Furthermore, in order for this length scale to survive in the continuum, it must be scaled as the continuum limit is taken. We will therefore introduce a second parameter into the tooth-length probability distribution $\mu(\ell)$, which defines a length scale for some structure of the comb. One would hope that in more realistic models such a length scale could be generated dynamically.

We now give a precise meaning to the term continuum limit. We assign the value $a$ to the distance between adjacent vertices in the graph and take the limit $a\to 0$ and $L\to\infty$ in such a way that the scaled combs have a finite characteristic distance scale; it is this limit to which we give the name continuum limit and quantities which exist in this limit we call continuum quantities. Walks much longer than $L$ will probe different structure from walks much shorter than $L$ but nonetheless both can be very long in units of the underlying cut-off scale $a$.

In the following sections we will denote dependence of a function on a number of variables $L_i$ \footnote{In the following, the parameters $L_i$ will becomes length scales related to structures in the comb.},  $i=1,\ldots,N$, by $L_i$ passed as one of the function arguments. Given a random comb ensemble specified by $\mu(\ell; L_i)$ and the corresponding $\bar{Q}(x; L_i)$  we define the continuum limit of $\bar{Q}$ by,  
\beq\label{Qlimit}
\tilde{Q}(\xi; \lambda_i) = \lim_{a\rightarrow 0} a^{\Delta_\mu} \bar{Q}(a \xi;a^{-\Delta_i}\lambda_i^{\Delta_i}),
\eeq
where the scaling dimensions $\Delta_\mu$ and $\Delta_i$ are chosen to ensure a non-trivial limit and the combinations $\xi\lambda_i$ are dimensionless. As we shall see, the choice for $\Delta_\mu$ and $\Delta_i$ is unique given mild assumptions. The function $\tilde Q$ can be used to define the spectral dimension at short and long distances.

The variable $\xi$ of the generating function plays the role of a fugacity for the walk length. Indeed it is the Laplace conjugate variable to walk length, this is completely analogous to the way in  which the variable for the resolvent in Chapter \ref{Backgroundnew} plays the role of the continuum boundary cosmological constant and is related to the boundary length by a Laplace transform. We therefore expect $\xi \rightarrow \infty$ to correspond to very short continuum walks and $\xi \rightarrow 0$ to correspond to long continuum walks. Let us for a moment suppose that the long and short distance behaviour of $\tilde{Q}(\xi; \lambda)$ is that of a power law and further that there exists only one length scale $\lambda$. If a diffusion experiment were performed in which the diffusion time was much greater than $\lambda \gg 1$ then an experimenter would see the power law behaviour associated with the $\xi \rightarrow \infty$ limit of $\tilde{Q}(\xi; \lambda)$. Conversely, if the walk length was much less then $\lambda$ then an experimenter would observe the power law behaviour of the $\xi \rightarrow 0$ limit of $\tilde{Q}(\xi; \lambda)$. If the exponent of the power law differs in these two limits then the experimenter would conclude the spectral dimension differed on short and long distance scales.

We now give an explicit definition for the long and short distance spectral dimension and show that the above long and short distance limits only contain contributions from long and short walks respectively. We will assume for simplicity that there is just one scale $L$ and that the spectral dimension in the sense of \eqref{tdependence} exists for $\avg{p_C(t)}$ which implies that there exists a constant  $t_0$ such that $\avg{p_C(t+1)} <\avg{p_C(t)}, t>t_0$.
Note that 
\bea\label{short} 
&&\fl{\sum_{t=0}^T \avg{p_C(t)} (1-x)^{t/2}= \Qbar(x;L)-\sum_{t=T+1}^\infty \avg{p_C(t)}(1-x)^{t/2}}\nn\\
&&\qquad\qquad \fl{=\Qbar(x;L)-(1-x)^{(T+1)/2}\sum_{t=0}^\infty \avg{p_C(t+T+1)}(1-x)^{t/2}}\nn\\
&&\qquad\qquad \fl{>\Qbar(x;L)-(1-x)^{(T+1)/2}\sum_{t=0}^\infty \avg{p_C(t)}(1-x)^{t/2} \nn \\
&&\qquad\qquad\quad  + (1-x)^{(T+1)/2}\sum^{t_0}_{t=0} (p_C(t)-p_C(t+T+1))(1-x)^{t/2}.}
\eea
Since the final sum is finite it will have an analytic dependence on $a$. Now choose   
\beq T=\left\lfloor a^{-1}\frac{1}{\xi \log(1+\frac{1}{\xi\lambda})}\right\rfloor-1\eeq
and set  $x=a\xi $ and $L= a^{-\Delta}\lambda^\Delta$ in \eqref{short} to get
\bea
&&\fl{\quad a^{\Delta_\mu} \Qbar(a\xi ;a^{-\Delta}\lambda^\Delta )\left(1-\exp\left(-\xi\lambda\right)\right)+ a^{\Delta_\mu} R \exp\left(-\xi\lambda\right) \nn \\
&&\qquad< a^{\Delta_\mu}\sum_{t=0}^T \avg{p_C(t)} (1-\xi a)^{t/2} \nn \\
&&\qquad<a^\Delta_\mu\Qbar(a\xi ;a^{-\Delta}\lambda^\Delta ),}\eea
where $R$ is a constant arising form finite sum in \eqref{short}. Provided that the limit in \eqref{Qlimit} exists we see that the behaviour of $\tilde Q(\xi;\lambda)$ as $\xi\to\infty$ characterizes the properties of walks of continuum time duration less than 
\beq \lim_{\xi\to\infty}\frac{1}{\xi \log(1+\frac{1}{\xi\lambda})}=\lambda,\eeq
  and we define the  spectral dimension $d^0_s$ at short distances  by
\beq d^0_s= 2\left(1+\lim_{\xi\to\infty} \frac{\log(\tilde Q(\xi;\lambda))}{\log \xi}\right),\label{dsShort}\eeq
provided this limit exists. 

We can define the spectral dimension at long distances in a similar way. First  note that 
by \eqref{Qbounds}
\bea\fl{  \sqrt{T}\ge\sum_{t=0}^\infty \avg{p_C(t)} \left(1-\Tinv\right)^{t/2}>\sum_{t=0}^T \avg{p_C(t)} \left(1-\Tinv\right)^{t/2} >\left(1-\Tinv\right)^T\sum_{t=0}^T \avg{p_C(t)}} \eea
so that 
\bea
\Qbar(x;L) -\sqrt{T}(1-\Tinv)^{-T}<\ \sum_{t=T+1}^\infty \avg{p_C(t)}(1-x)^{t/2}< \Qbar(x;L).\eea
This time letting $T=\lfloor a^{-1}\xi^{-1} \log(1+\xi\lambda)\rfloor-1$ we get
\bea\fl{ a^{\Delta_\mu}\Qbar(a\xi;a^{-\Delta}\lambda^\Delta) -e\sqrt{\xi^{-1} \log(1+\xi\lambda)}<  a^{\Delta_\mu}\sum_{t=T+1}^\infty \avg{p_C(t)}(1-\xi a)^{t/2}< a^{\Delta_\mu}\Qbar(a\xi; a^{-\Delta}\lambda^\Delta)}.\nn\\
\eea
If the above limit exists and furthermore $\tilde Q(\xi;\lambda)$ diverges as $\xi \rightarrow 0$ then we see that in this limit $\tilde Q(\xi;\lambda)$ only receives contributions from walks whose length exceeds $\lim_{\xi\to0} \xi^{-1} \log(1+\xi\lambda)=\lambda$. We define the spectral dimension $d_s^\infty$ at long distances to be
\beq d^\infty_s= 2\left(1+\lim_{\xi\to 0} \frac{\log(\tilde Q(\xi;\lambda))}{\log \xi}\right),\label{dsLong}\eeq
provided this limit exists.

At this point it is unclear whether there exists any examples for which the short and long distance spectral dimensions as in \eqref{dsShort} and \eqref{dsLong} exist and furthermore differ in value. We will now discuss some examples based on random combs in which we realise such examples. In all the examples given in this chapter it turns out that the exponent is $\Delta_\mu=\half$.


\section{A simple comb}
\label{SimpleExample}
We now introduce a random comb whose spectral dimension differs on long and short length scales and thus illustrates that the behaviour described in section \ref{dsdiscussion} can actually occur. 
This comb is defined by the measure,
\bea \label{easycomb}
\mu(\ell;L)&=&\begin{cases}1-\frac{1}{L},& \text{$\ell=0,$}\\ \frac{1}{L},& \text{$\ell=\infty,$}\\0,&\text{otherwise.}\end{cases}
\eea
This random comb  has infinite teeth and they occur with an average separation of $L$. Intuitively we would expect that if a random walker did not move further than a distance of order $L$ from its starting position it would not see the teeth and therefore would measure a spectral dimension of one. If however it were allowed to explore the entire comb it would see something roughly equivalent to a full comb and so feel a much larger spectral dimension. 
To prove this intuition correct we proceed by computing upper and lower bounds for $\bar{Q}$ which are uniform in $L$ and for $0<x<x_0$, where the constant $x_0$ is equal to one unless otherwise stated, and then take the continuum limit to obtain bounds for $\tilde{Q}$ . 


With complete generality we may obtain a lower bound on $\bar{Q}(x)$ by use of Jensen's inequality which takes the form,
\begin{lemma}
\label{GenLBLem}
Let $\bar{P}_T(x;L_i)$ be the mean first return probability generating function of the teeth of the comb defined by $\mu(\ell; L_i)$, then 
\bea
\label{Qeqn}
\bar{Q}(x;L_i) \geq (1+x-\bar{P}_T(x;L_i))^{-\half}.
\eea
\end{lemma}
The proof is given in \cite{Durhuus:2005fq}. For the  comb \eqref{easycomb} we have
\bea
\bar{P}_T(x;L) = 1 - \frac{1}{L}(1-P_{\infty}(x)) = 1 - \frac{\sqrt{x}}{L} 
\eea
 which implies
\bea
\bar{Q}(x;L) \geq \left(\frac{\sqrt{x}}{L}+x\right)^{-\half}.
\eea
Letting $x = a \xi$ and $L = a^{-\half}\lambda^{\half}$ gives

\bea
\label{SimpleLB}
\tilde{Q}(\xi; \lambda)= \lim_{a\to 0} a^\half \Qbar(a \xi; a^{-\half}\lambda^{\half})\geq
\xi^{-\half}\left(\upsilon^{-\half} + 1\right)^{-\half},
\eea
where we have introduced the dimensionless variable $\upsilon = \xi \lambda$.


To find an upper bound on  $\bar{Q}(x;L)$ we follow \cite{Durhuus:2005fq} and use Lemmas \ref{MonoLem1},  \ref{MonoLem2} and  \ref{RearrangeLem1} to compare a typical comb in the ensemble with the comb consisting of a finite number of infinite teeth at regular intervals.  
First we define the event
\bea
\C A(D,k) = \{C:D_i \leq D:i=0,...,k\}.
\eea
where $D_i$ is the distance between the $i$ and $i+1$ teeth and then
 write,
\bea
\label{Qintegral}
\bar{Q}(x;L) &=& \int_{\C{C}} Q_C(x;L) d\nu \nn\\
 &=& \int_{\C{C}/\C{A}(D,k)} Q_C(x;L) d\nu + \int_{\C{A}(D,k)} Q_C(x;L) d\nu.
\eea
Since the $D_i$ are independently distributed
\beq \nu(\C A(D,k))=(1-(1-1/L)^D)^k.\eeq

Consider a comb $C \in \C A(D,k)$; then by Lemmas \ref{MonoLem1},  \ref{MonoLem2} and  \ref{RearrangeLem1}, 
\beq P_C(x;L) \leq P_{C'}(x),\eeq
where $C'$ is the comb obtained by removing all teeth beyond the $k$ tooth and moving the remaining teeth so that the spacing between each is $D$. Now we can write 
\beq P_{C'}(x) = P^{(<Dk)}_{C'}(x) + P^{(>Dk-1)}_{C'}(x).\eeq Since the walks contributing to $P^{(<Dk)}_{C'}(x)$ do not go beyond the last tooth we have \beq P^{(<Dk)}_{C'}(x) \leq P_{*D}(x),\label{bd1}\eeq where $*D$ denotes the comb consisting of infinite teeth regularly spaced and separated by a distance $D$. Using \eqref{bd1}, Lemmas \ref{P2bound} and \ref{Mod2ptBounds} we have,
\beq\label{PCupper1}
P_C(x;L) \leq P_{{*D}}(x) +  3x^{-\half} G^{(0)}_\infty(x; Dk)^2
\eeq
uniformly in $\C A$. $P_{*D}(x)$ and $G^{(0)}_\infty(x; n)$ are given in Appendix A. 
Now set $D = \lfloor\tilde{D}\rfloor$ and $k=\lceil\tilde{k}\rceil$, where,
\beq 
\tilde{D} = 2 L|\log x L^2|, \qquad \tilde{k}= (xL^2)^{-1/2}.
\eeq
Since $G^{(0)}_\infty(x; n)$ is manifestly a monotonic decreasing function of $n$  and $ P_{*D}(x)  $ an increasing function of $D$,
\beq\label{bd100}
\fl{\bar{Q}(x;L) \leq x^{-1/2} (1-(1-(1-1/L)^{\tilde{D}-1})^{\tilde{k}+1}) + Q_{U}(x)(1-(1-1/L)^{\tilde{D}})^{\tilde{k}} }
\eeq
where we have used \eqref{Qbounds} and
\beq Q_U(x) = \left[1- P_{*\tilde{D}}(x) - 3x^{-\half} G^{(0)}_\infty(x; (\tilde{D}-1)\tilde{k})^2\right]^{-1}.
\eeq
Taking the continuum limit of \eqref{bd100} and using the results of \ref{StandardResults} then gives
\beq\label{SimpleUB} \tilde{Q}(\xi;\lambda) \leq \xi^{-1/2} F(\xi\lambda),\eeq
where
\bea F(v)=\begin{cases}1+o(v^{-1}), &\text{$v\to \infty$},\\v^\quarter\sqrt{\abs{\log v^2}}+o(v^\half), &\text{$v\to 0$}.\end{cases}\eea
It follows from \eqref{dsShort}, \eqref{dsLong}, \eqref{SimpleLB} and  \eqref{SimpleUB}
that
\bea d_s^0=1,\qquad
d_s^\infty=\threehalves.\eea

\section{Combs with Power Law Measures}
\label{PowerComb}
We now consider slightly more general combs in which the measure on the teeth is a power law of the form,
\bea \label{powerdist}
\mu(\ell;L)&=&\begin{cases}1-\frac{1}{L},&\text{$\ell=0,$}\\ \frac{1}{L}C_{\alpha} \ell ^{-\alpha},& \text{$\ell>0,$}\end{cases}
\eea
where $C_\alpha$ is a normalisation constant and as before $L$ plays the role of a distance scale.  We consider laws in the range $2>\alpha>1$ as it is known that for $\alpha\ge 2$ the comb has spectral dimension $d_s=1$ in the sense of \eqref{QCdsdef} \cite{Durhuus:2005fq} and therefore it is not possible to get a spectral dimension deviating from 1 on any  scale.

To compute a lower bound on the return probability generating function for the above distribution we apply Lemma \ref{GenLBLem} and reduce the problem to computing an upper bound on $1-\bar{P}_T(x)$.
The first return generating function $P_\ell(x)$ for a tooth of length $\ell$  is recorded in \eqref{Pell}; bounding  $\tanh(u)$ above by the function $f(u) = u$ for $u<1$ and $f(u) = 1$ for $u\geq 1$ gives\footnote{In this particular case we could in fact compute $1-\bar{P}_T(x)$ exactly by the Abel summation formula. However the bound we use is good enough to give the desired result with the advantage that the calculation can be done with elementary functions.}
\bea
\label{LB1}
1-\bar{P}_T(x;L_i) &\leq&  \sqrt{x} \left[ m_{\infty} (x)\sum^{[m_\infty^{-1}]}_{\ell=1}  \mu(\ell;L_i) \ell + \sum^{\infty}_{\ell=[m_\infty^{-1}]+1} \mu(\ell;L_i)\right] \nn \\
&\leq& \sqrt{x} -m_\infty(x)\sqrt{x} \int^{\frac{1}{m_\infty(x)}}_0\left( \sum^{[u]}_{\ell=0} \mu(\ell;L_i)\right) du.
\eea
To obtain the second inequality we have applied the Abel summation formula. We therefore have,
\begin{lemma}
\label{GenLBLem2}
For a random comb defined by the measure $\mu$,
\beq
1-\bar{P}_T(x;L_i) \leq \sqrt{x} -m_\infty(x)\sqrt{x} \int^{\frac{1}{m_\infty(x)}}_0 \chi(u;L_i) du
\eeq
where the cumulative probability function $\chi(u;L_i)$ is defined by $\chi(u;L_i)  = \sum^{[u]}_{\ell=0} \mu(\ell;L_i)$.
\end{lemma}
We will see shortly that all behaviour of the spectral dimension of the continuum comb is encoded in the asymptotic expansion of ${\chi}(u;L_i)$ as $u$ goes to infinity \footnote{In general it is not obvious that this asymptotic expansion exists due to the discontinuous nature of $\chi$. We will address this issue later when we consider generic measures.}.
In the present case ${\chi}(u;L)$ is trivially related to the partial sum of the Riemann $\zeta$-function whose leading asymptotic behaviour is well known and we find 
\bea
{\chi}(u;L) &=&  1 - \frac{C_\alpha}{L}\frac{u^{1-\alpha}}{\alpha-1} +\delta(u),
\eea
where 
\bea \abs{\delta(u)}<\frac{c}{L}u^{-\alpha}, \quad u\ge 2.\eea
It follows that for $x<x_0$, where $m_\infty(x_0)=\half$,
\bea 1-\bar{P}_T(x;L) \leq  m_\infty(x)\sqrt{x}\left(\frac{b_1}{L}m_\infty(x)^{\alpha-2}+\frac{b_2}{L}m_\infty(x)^{\alpha-1}+\frac{b_3}{L}\right),\eea
with $b_{1,2,3}$ being  constants depending only on $\alpha$ and $b_1>0$.
Choosing  $L= a^{-\Delta'} \lambda^{\Delta'}$ with $\Delta' = 1-\alpha/2$  yields a lower bound on the continuum return generating function,
\beq
\label{PowerLB}
\tilde{Q}(\xi,\lambda)
\geq \xi ^{-1/2}\left(1+b_1 (\xi\lambda)^{-(1-\alpha/2)}\right)^{-1/2}. \eeq


To obtain a comparable upper bound 
we need 
\begin{lemma}
\label{GenUBLem}
For any random comb and positive integers $H$, $D$ and $k$, the return probability generating function is bounded above by
\bea
\label{GeneralUB}
\bar{Q}(x;L_i) \leq x^{-1/2} (1-(1-(1-p)^{{D}})^{{k}}) + Q_{U}(x)(1-(1-p)^{{D}})^{{k}} ,
\eea
where
\bea p &=& \sum^{\infty}_{\ell = H+1} \mu(\ell;L_i),  \nn\\ 
\label{GeneralQU}
Q_U(x) &=& \left[1- P_{{H},*{D}}(x) - 3x^{-\half} G^{(0)}_\infty(x; {D}{k})^2\right]^{-1},
\eea
and $P_{H, *D}$ is the first return probability generating function for the comb with teeth of length $H+1$ equally spaced at intervals of $D$.
\end{lemma}
The proof is a slight modification of the upper bound argument used in Section \ref{SimpleExample}.
First define a long tooth to be one whose length is greater than $H$; then the probability that a tooth at a particular vertex is long is
\beq p = \sum^{\infty}_{\ell = H+1} \mu(\ell;L_i).  \eeq
Define the event
\bea
\label{Aevent}
\C A(D,k) = \{C:D_i \leq D:i=0,...,k\}
\eea
where now $D_i$ is the distance between the $i$ and $i+1$ long teeth so that
\bea
\bar{Q}(x;L_i) &=& \int_{\C{C}} Q_C(x;L_i) d\nu \nn\\
&=& \int_{\C{C}/\C{A}(D,k)} Q_C(x;L_i) d\nu + \int_{\C{A}(D,k)} Q_C(x;L_i) d\nu.
\eea
Since the $D_i$ are independently distributed
\beq\label{nuD} \nu(\C A(D,k))=(1-(1-p)^D)^k.\eeq
Now use Lemmas \ref{MonoLem2} and \ref{RearrangeLem1} in turn  to note that for 
\beq P_{C\in \C A(D,k)}(x,L) \leq P_{C'}(x,L)\eeq
where $C'$ is the comb in which all teeth but the first $k$ long teeth have been removed and the remaining long teeth have been arranged so that they have length $H+1$ and a constant inter-tooth distance  $D$. By the same arguments as we used in Section \ref{SimpleExample} to get \eqref{PCupper1} we obtain the bound 
\beq
\label{PCpexp}
P_{C\in \C A(D,k)}(x,L) \leq P_{{H+1}, *{D}} +3 x^{-1/2} G^{(0)}_\infty (x,Dk)^2.
\eeq
Lemma \ref{GenUBLem} then follows from \eqref{Qbounds},  \eqref{nuD} and \eqref{PCpexp}. We now specialise to the power law measure \eqref{powerdist} and set $H=\lfloor\tilde{H}\rfloor$, $D = \lfloor\tilde{D}\rfloor$ and $k=\lceil\tilde{k}\rceil$, where
\bea
\label{LambdaK}
\tilde{H} &=& x^{-1/2} \nn\\
\tilde{D} &=& (\Delta'+1) \frac{\alpha-1}{c_\alpha} x^{\Delta'-1/2} L |\log x L^{1/\Delta'}| \\
\tilde{k} &=& (x L^{1/\Delta'})^{-\Delta'}.\nn 
\eea

Using Lemma \ref{GenUBLem}, the scaling expressions for $P_{H, *D}$ and $G_\infty^{(0)}$ given in \eqref{Pelln}, and taking the continuum limit, gives, after a substantial amount of algebra,
\beq\label{PowerUB1} \tilde{Q}(\xi;\lambda) \leq \xi^{-1/2} F(\xi\lambda),\eeq
where
\bea F(v)=\begin{cases}1+O(v^{-1}), &\text{$v\to \infty$},\\c\, v^{1/2-\alpha/4}  \sqrt{\abs{\log v^2}}+O(v^{\Delta'}), &\text{$v\to 0$}.\end{cases}\eea
The main result of this section is  
\begin{theorem}
The comb with the power law measure \eqref{powerdist} for the tooth length has
 \bea d_s^0=1,\qquad
d_s^\infty=2-\frac{\alpha}{2}.\eea
\end{theorem}
The result follows immediately from \eqref{dsShort}, \eqref{dsLong}, \eqref{PowerLB} and  \eqref{PowerUB1}.

\section{Multiple Scales}
\label{multiscale}
Given the results for the power law distribution it is natural to investigate the behaviour for a  random comb that has a hierarchy of length scales. The easiest way to achieve such a comb is through a double power law distribution,
\bea \label{2powerdist}
\mu(\ell;L_i)&=&\begin{cases}1-L_1^{-1}-L_2^{-1},&\text{$\ell=0,$}\\ \frac{1}{L_1}C_{1} l ^{-\alpha_1}+\frac{1}{L_2} C_{2} l ^{-\alpha_2} ,& \text{$\ell>0.$}\end{cases}
\eea
We may assume without loss of generality that the length scales $L_i$ scale in the continuum limit to lengths $\lambda_i$ such that $\lambda_1< \lambda_{2}$ and that $1<\alpha_i<2$.  

Following the procedure of previous sections a lower bound on $\tilde Q(\xi;\lambda_i)$ is obtained  by using Lemma \ref{GenLBLem}  and noting that $\chi(x;L_i)$ for this comb is essentially the sum of the cumulative probability functions for each power law. This gives
\bea 
1-\bar{P}_T(x;L_i) \leq  m_\infty(x)\sqrt{x}\sum_{i=1}^2\left(\frac{b_{1i}}{L_i}m_\infty(x)^{\alpha_i-2}+\frac{b_{2i}}{L_i}m_\infty(x)^{\alpha_i-1}+\frac{b_{3i}}{L_i}\right).\nn \\\eea
Choosing $L_i$ to scale like $L_i=a^{-\Delta'_i} \lambda_i^{\Delta'_i}$, where $\Delta'_i = 1-\alpha_i/2$, gives a bound on the continuum return generating function of,
\bea
\label{DoubleLB}
{\xi^{-1/2}\left(c_0 +c_1 (\xi \lambda_1)^{-(1-\alpha_1/2)}+c_2( \xi\lambda_2)^{-(1-\alpha_2/2)} \right)^{-1/2} \leq \tilde{Q}(\xi)}.
\eea
An upper-bound on $\tilde{Q}(\xi)$ is obtained by application of Lemma \ref{GenUBLem} in which we set $H=\lfloor\tilde{H}\rfloor$, $D = \lfloor\tilde{D}\rfloor$ and $k=\lceil\tilde{k}\rceil$, where
\bea
\label{DoubleLambdaK}
\tilde{H} &=& x^{-1/2} \nn\\
\tilde{D} &=& \beta x^{-1/2} G(x L_1^{1/\Delta'_1},x L_2^{1/\Delta'_2})^{-1} |\log x L_1^{1/\Delta'_1}| \\
\tilde{k} &=& G(x L_1^{1/\Delta'_1},x L_2^{1/\Delta'_2}) \nn 
\eea
and for convenience we have introduced the function,
\beq
G(\upsilon_1,\upsilon_2) = \frac{C_1}{\alpha_1-1} \upsilon_1^{-\Delta'_1}+\frac{C_2}{\alpha_2-1} \upsilon_2^{-\Delta'_2} .
\eeq
Using  Lemma \ref{GenUBLem} and the scaling expressions in \ref{StandardResults} then gives
\bea
\label{DoubleUB}
&&\tilde{Q}(x;\lambda_i) \leq \xi^{-1/2}\Bigg[1-(1-\upsilon_1^{-s\beta})^G + \nn \\
&&\frac{ (1-\upsilon_1^{-s\beta})^{G} }{3\mathrm{cosech}^2(|\log\upsilon_1^\beta|)-\gamma+\sqrt{\gamma^2+1+2\gamma \coth(|\log\upsilon_1^\beta|/G)}} \Bigg]
\eea
where $\upsilon_i=\xi\lambda_i$,  $s=\sgn(\log\upsilon_1)$, $\gamma=\tanh(1)$ and  we have suppressed the arguments of $G(\upsilon_1,\upsilon_2)$ in order to maintain readability. 

We can now examine  \eqref{DoubleLB} and \eqref{DoubleUB} to see what they tell us about the behaviour of $\tilde{Q}(\xi;\lambda_i)$ on various length scales. 
\begin{itemize}
\item When $\xi\gg\lambda_1^{-1}$ both upper and lower bounds of $\tilde Q(\xi;\lambda_i)$ are dominated by the $\xi^{-\half}$ behaviour so taking the $\xi\to\infty$ limit leads to $d_s^0=1$ as in the previous sections.

\item If $\alpha_1 <  \alpha_2$ then when $\xi\ll\lambda_1^{-1}$ both upper and lower bounds of $\tilde Q(\xi;\lambda_i)$ are dominated by the $\xi^{-\alpha_1/4}$ behaviour so taking the $\xi\to 0$ limit leads to $d_s^\infty=2-\alpha_1/2 $. There is no regime in which $\alpha_2$ controls the behaviour.

\item If $\alpha_2 <  \alpha_1$ then when $\xi\ll\Lambda^{-1}$ where
\beq \Lambda^{-1}=\lambda_1^{(2-\alpha_1)/(\alpha_1-\alpha_2)}\lambda_2^{(2-\alpha_2)/(\alpha_2-\alpha_1)}\eeq
both upper and lower bounds of  $\tilde Q(\xi;\lambda_i)$ are dominated by the $\xi^{-\alpha_2/4}$ behaviour so taking the $\xi\to 0$ limit leads to $d_s^\infty=2-\alpha_2/2 $. However there is an intermediate regime $\Lambda^{-1}\ll\xi\ll\lambda_1^{-1}$ where the $\xi^{-\alpha_1/4}$ behaviour dominates and $\tilde Q(\xi;\lambda_i)$ lies in the envelope given by
\beq c_1 \xi^{-\alpha_1/4} <\tilde Q(\xi;\lambda_i)< c_2\xi^{-\alpha_1/4}\sqrt{\abs{\log\xi^\beta}}, \eeq
where the upper and lower bounds will have corrections suppressed by powers of $\xi \lambda_2$ and the upper bound will also have corrections of order $\xi^\beta$. Both $\lambda_2$ and $\beta$ may be chosen to make the corrections arbitrarily small in this scale range. The system therefore appears to have spectral dimension $\delta_S=2-\alpha_1/2$ in this regime. This is a fairly weak statement because $Q(\xi)$ could in principle exhibit a wide variety of behaviours between its upper and lower bounds; this region is just a part of the crossover regime from $d_S^0$ to $d_S^\infty$. However, as we are free to chose $\lambda_2$ to be as large as we like compared to $\lambda_1$, this regime can exist over a scale range of arbitrarily large size. We therefore can force the leading behaviour of $\tilde Q(\xi;\lambda_i)$ in this range to be as close to a power law with exponent $\delta_S=2-\alpha_1/2$ as we like. This is what might be observed, for example, in a numerical simulation; if the difference between the scales $\lambda_1$ and $\lambda_2$ is large then there will be a substantial range of walk lengths in which the data will indicate a spectral dimension of $\delta_S$. We will refer to a spectral dimension that appears in this weaker way as an {\emph{apparent spectral dimension}} and denote it by $\delta_S$ rather than $d_S$.

\end{itemize}

\section{Generic Distributions}
\label{genericcomb}
So far we have considered combs in which the distribution of tooth lengths has been governed by power laws or double power laws. In this section we  extend the results of the previous sections to the case were the form of the tooth length distribution is left arbitrary. The most general situation is that the measure on the combs is a continuous function of some parameters $w_i$. The continuum limit of such a comb is obtained in the usual way but with parameters $w_i$ scaling in a non-trivial way; $w_i = w_{c_i} + a^{d_i} \omega_i$. Given a random comb with such a measure we would like to know how many distinct continuum limits exist and for each compute how the spectral dimension depends on the length scale.

The approach we adopt here closely mimics the arguments of the preceding sections, indeed the main complication is technical. As we have seen the properties of the continuum comb are controlled by the asymptotic expansion of $\chi(u)$ as $u$ goes to infinity. The main difference in the generic case is that we may arrange matters so that the scaling dimensions of the coefficients in the asymptotic expansion are such that sub-leading terms appear in the continuum. For the generic case we obviously have no way of knowing the full asymptotic expansion of $\chi(u; w_i)$. However, we will see that for a large class of measures, the form of the asymptotic expansion is encoded in the asymptotic expansion of a particular generating function for $\mu(\ell; w_i)$.

Our first task is to introduce this generating function and relate it to the asymptotic expansion of $\chi(u; w_i)$. To this end we introduce the notion of a smoothed sum \cite{tao},
\bea
\fl{{\chi_\pm}(u;{w_i}) = \sum^\infty_{\ell=0} \mu(\ell;{w_i}) \eta_\pm(\ell/u)=\mu(0;w_i)+ \sum^\infty_{\ell=1} \mu(\ell;{w_i}) \eta_\pm(\ell/u)} \equiv\mu(0) +\chi_\pm^{(1)}(u;{w_i}), \nn \\
\eea
where $\eta_\pm$ is the smooth cut-off function introduced in \ref{AppendixBump} and $u$ controls where the cut-off occurs. Such smoothed sums are related to ${\chi}(u; w_i)$ by,
\bea
\label{chibound}
{\chi_-}(u;{w_i}) \leq \chi(u;{w_i}) \leq {\chi_+}(u;{w_i}). 
\eea
The reason for introducing the smoothed sums is that we may use powerful techniques from complex analysis to compute their asymptotic expansion (see eg  \cite{Flajolet}). The generating function to which the asymptotic expansion of the smoothed sums is related is  the Dirichlet series generating function of $\mu(\ell;{w_i})$,
\beq
\C{D}_\mu(s;{w_i}) = \sum^{\infty}_{\ell=1} \frac{\mu(\ell;{w_i})}{\ell^s}.
\eeq

\begin{figure}[t]
  \begin{center}
    \includegraphics[width=5cm]{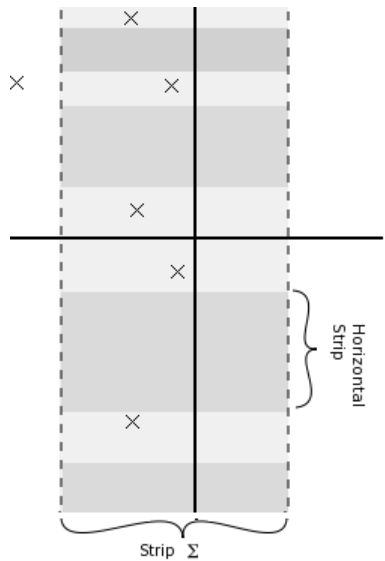}
    \caption{An illustration of the strip $\Sigma$ in which $\C{D}_\mu(s;{w_i})$ satisfies property (1) given given in the text. The poles of $\C{D}_\mu(s;{w_i})$ are denoted by crosses and the horizontal strips in which the growth condition holds are indicated by the dark
grey regions.}
    \label{fig2}
  \end{center}
\end{figure}

We now introduce a number of results and notations,
\begin{itemize}
\item[1.] For a strip in the complex plane $\Sigma(b,a) \equiv \{z: a< \mathrm{Re}[z] < b\}$ where $b>a$, we say $\C{D}_\mu(s;{w_i})$ has {\emph{slow growth}} in $\Sigma(b,a)$ if for all $s \in \Sigma$ we have $\C{D}_\mu(s;{w_i}) \sim O(s^{r})$ for some $r>0$ as $\mathrm{Im}[s] \rightarrow \pm \infty$. We say $\C{D}_\mu(s;{w_i})$ has {\emph {weak slow growth}} if the above property only holds for a countable number of horizontal regions across the strip. See figure \ref{fig2}.
\item[2.] Define $S_\pm$ to be the set containing the triples $(-\sigma+i \tau,-k,-r_\pm)$ such that the Laurent expansion of $\C{D}_\mu(s;{w_i}) \C{M}[\Psi_{\pm\epsilon}](s+1)/s$ about the point $-\sigma+i \tau$ contains the term $r_\pm/(s +\sigma-i \tau)^k$ where $k>0$ and $-\sigma+i \tau \neq0$. Here $\C{M}$ denotes the Mellin transformation and $\Psi_{\pm\epsilon}$ is introduced in \ref{AppendixBump}. We will often find it useful to refer to the elements of $S_\pm$ using an index $j$ and denote the $j$th element of $S_\pm$ by $(-\sigma_j+i \tau_j,-k_j,-r_{j,\pm})$. Note that since $\C{M}[\Psi_{\pm\epsilon}](s+1)$ is analytic in $\Sigma$ then the positions of the poles are determined only by $\C{D}_\mu(s;{w_i})$.

\item[3.] Define the indexing sets $J_R$ and $J_C$, such that if $j \in J_R$ then $\tau_j = 0$ whereas if $j \in J_C$ then $\tau_j \neq 0$ and for both $-\sigma_j+i \tau_j\in \Sigma$ i.e. they index the poles in $\Sigma$ which lie on the real line and off the real line respectively. 
\item[4.] For a Dirichlet series with positive coefficients the abscissa of absolute and conditional convergence coincide. Furthermore, since $\sum^\infty_{l=0} \mu(\ell;{w_i}) = 1$ the abscissa of convergence is less than zero.
\item[5.] Landau's theorem: For a Dirichlet series with positive coefficients there exists a pole at the abscissa of convergence. A corollary of this is that the coefficient of the most singular term in the Laurent expansion about the abscissa of convergence is positive and furthermore it is the right-most pole of the Dirichlet series in the complex plane.
\end{itemize}

If $\C{D}_\mu$ is of slow growth in $\Sigma(a,b)$ and no pole occurs to the right of $\Sigma(a,b)$, then by using the expression \eqref{genasymp} the asymptotic expansion of $\chi_\pm(u;{w_i})$ is, 
\bea
\label{ChiAsymp}
&&{\chi_\pm(u;{w_i}) = \mu(0;{w_i})+\frac{1}{2 \pi i}\oint_C u^{s} \C{D}_\mu(s;{w_i}) \frac{M[\Psi_{\pm\epsilon}](s+1)}{s} ds -R(u;{w_i})} \nn \\
&&= 1 - \sum_{j\in J_R}r_{j,\pm}({w_i})\frac{(\log u)^{k_j-1} }{(k_j-1)!}u^{-\sigma_j}-\\
&&\half \sum_{j\in J_C}\left(r_{j,\pm}({w_i}) u^{i \tau_j}+{r_{j,\pm}}^*({w_i}) u^{-i \tau_j} \right)\frac{(\log u)^{k_j-1}}{(k_j-1)!} u^{-\sigma_j}-R(u;{w_i}), \nn
\eea
where $C$ is the rectangular contour introduced in \ref{AppendixAsymp}\footnote{Since we have assumed there are no poles to the right of $\Sigma$ we may choose $c = b$, where $c$ is the constant appearing in the definition of the contour.}, $r_{j,\pm}$ are the coefficients of the Laurent expansions of $\C{D}_\mu(s;{w_i}) \C{M}[\Psi_{\pm\epsilon}](s+1)/s$, and $R(u)$ is a remainder function which satisfies $R(u) \sim u^{-N}$, where $N>\sigma_j$, as $u$ goes to infinity. Note that the difference between the coefficients $r_{j,+}$ and $r_{j,-}$ is of order $\epsilon$, the parameter introduced in \ref{AppendixBump}, which can be taken to be arbitrarily small and so $r_{j,+}$ and $r_{j,-}$ are for all purposes equal. 

We now are in a position to prove the main result of this section. Note that in the following we require the continuum comb is such that $\Delta_\mu= 1/2$ and so has spectral dimension one on the smallest scales. It seems very likely that this restriction could be lifted
but we will not  pursue such a generalisation here.

\begin{theorem}
\label{mainresult}
For a comb in which the teeth are distributed according to the measure $\mu$ then if $\C{D}_\mu$ has slow growth in the strip $\Sigma = \{z:-1 < Re[z]  \leq 0\}$ and has no poles on its line of convergence besides at the abscissa then continuum limits of the comb exist with $d_s^0=1$ and $d_s^\infty$ taking values in the set $\{\frac{3-\sigma_j}{2}:0<\sigma_j< 1, j\in J_R\}$.\end{theorem}

\begin{proof}
The proof proceeds in much the same way as the proofs in the previous sections; we first derive upper and lower bounds on $\tilde{Q}_C$ and then use these bounds to deduce the behaviour of $\tilde{Q}_C$ in different scale ranges.

We begin with the lower bound which by Jensen's inequality amounts to finding an upper bound on $1-\bar{P}_T$. Note that \eqref{Qeqn} implies that the scaling dimension of $1-\bar{P}_T$ must be greater than or equal to one in order for $\Delta_\mu = 1/2$ and only contributes to the continuum limit if it has scaling dimension one. From Lemma \ref{GenLBLem2} we have,
\bea
\label{LB2}
&&\fl{1-\bar{P}_T(x;{w_i}) \leq  m_\infty(x)\sqrt{x} \int^{\frac{1}{m_\infty(x)}}_0 (1-\chi_-(u;{w_i})) du} \nn \\
&&\qquad\qquad\fl{\leq \sqrt{x}m_\infty(x)C(K;{w_i}) + m_\infty(x)\sqrt{x} \int^{\frac{1}{m_\infty(x)}}_K (1-\chi_-(u;{w_i})) du},
\eea
where $C(K) = \int^K_0 (1-\chi_-(u;{w_i})) du$ is a constant independent of $x$. It is important to recall that since $0\leq \chi_\pm(u;{w_i}) \leq 1$ for all values of $u$ and $w_1,\ldots,w_M$, in particular when $w_1,\ldots,w_M$ assume their critical values, that the scaling dimensions of the coefficients $r_{j,\pm}$ appearing in \eqref{ChiAsymp} must be positive as otherwise $\chi_\pm(u;{w_i})$ would diverge if we were to set $w_1,\ldots,w_M$ to their critical values. Upon performing the integration in \eqref{LB2} we will find that a given $r_{j,\pm}$ now is the coefficient for a number of terms of increasing scaling dimension. Since we require the scaling dimension of $1-\bar{P}_T$ to be greater than or equal to one, only the term with smallest scaling dimension can appear in the continuum limit and we will drop any term that does not appear in the continuum limit.
 If we now substitute \eqref{ChiAsymp} into \eqref{LB2} we find
\bea
\label{PTLB}
&&\fl{1-\bar{P}_T(x;{w_i}) \leq  \sqrt{x}m_\infty(x)C(K;{w_i}) + \sqrt{x}\sum_{j\in J_R}\frac{r_{j,-}({w_i})}{(k_j-1)!}\frac{(-\log m_\infty)^{k_j-1} }{1-\sigma_j}m_\infty^{\sigma_j}+}\nn \\
&&\frac{\sqrt{x}}{2} \sum_{j\in J_C}\left(\frac{r_{j,-}({w_i}) m_\infty^{-i \tau_j}}{1-\sigma_j+i \tau_j}+\frac{{r_{j,-}({w_i})}^*m_\infty^{i \tau_j}}{1-\sigma_j-i\tau_j}  \right)\frac{(-\log m_\infty)^{k_j-1}}{(k_j-1)!}  m_\infty^{\sigma_j},
\eea
where the remainder term $R$ disappears since it cannot appear in the scaling limit. We now suppose we may choose the critical values and scaling dimensions of the parameters $w_i$ such that $r_{j,\pm}$ has a scaling form that can be written as,
\beq
\label{rscaling}
r_{j,\pm} = \left(\frac{a}{\lambda_{j,\pm}}\right)^{\theta_j} \left(-\frac{1}{2}\log a\right)^{{\hat\theta}_j},
\eeq
where $\theta_j$, ${\hat \theta}_j$ and $\lambda_{j,\pm}$ are constants and we have included the factor of $\frac{1}{2}$ for later convenience. From \eqref{PTLB} one can see that since we require $\Delta_\mu = 1/2$ the only terms which appear in the continuum limit are those for which $\theta_j = (1-\sigma_j)/{2}$ and $\hat{\theta}_j = -(k_j-1)$ and so it is useful to define a restricted indexing set
\bea
\tilde{J}_R = \{j \in J_R : \theta_j = (1-\sigma_j)/{2}\quad \mathrm{and} \quad \hat{\theta}_j = -(k_j-1)\} 
\eea
and an equivalent one $\tilde{J}_C$ for $J_C$. The continuum limit is thus,
\bea
&&\fl{1-\bar{P}_T(x;w_i) \leq  a C \xi+\sum_{j\in \tilde{J}_R}\frac{a}{(k_j-1)!(1-\sigma_j)}\xi \left(\xi\lambda_{j,-}\right)^{-\frac{1-\sigma_j}{2}}+} \\
&&\qquad \qquad \fl{\sum_{j\in \tilde{J}_C}\frac{a}{(k_j-1)!}\Bigg(\mathrm{Re}[\phi_j]\cos(\frac{\tau_j}{2} \log (a\xi))+\mathrm{Im}[\phi_j]\sin(\frac{\tau_j}{2} \log (a\xi)) \Bigg)\xi \left(\xi|\lambda_{j,-}|\right)^{-\frac{1-\sigma_j}{2}}}, \nn
\eea
where $\phi_j = (|\lambda_{j,-}|/\lambda_{j,-})^{(1-\sigma_j)/2}(1-\sigma_j+i \tau_j)^{-1}$. The lower bound on $\tilde{Q}_C$ is then obtained from Lemma \ref{GenLBLem}.

Before deriving an upper bound on $\tilde{Q}_C$, we must analyse the consequences of any of the oscillatory terms in the second sum appearing in the continuum limit (i.e. if $\tilde{J}_C$ is non-empty). In the case of the double power law measure we saw that one could obtain intermediate behaviour in which an effective spectral dimension was measured that differed from the UV and IR spectral dimensions. A similar phenomenon occurs in the generic case when the various length scales of the continuum limit are well separated. If we scale the coefficients of the oscillatory terms such that these terms appear in the continuum limit then we are lead to an inconsistency by the following argument, 
\begin{itemize}
\item[1.] Let the term associated with an oscillatory term have index $j_0$. Consider choosing the scaling form of $x$ and $r_{j,-}$ such that $\xi |\lambda_{j_0,-}| \ll 1$ but $\xi |\lambda_{j,-}| \gg 1$ if $|\lambda_{j,-}| \gg |\lambda_{j_0,-}|$.
\item[2.] If we were to take the scaling limit in such a scenario then the term that dominates the size of $1-\bar{P}_T$ is the one associated with the length scale $|\lambda_{j_0,-}|$. This term will be oscillating around a mean of zero and hence must be going negative infinitely often as we approach the continuum.
\item[3.] Since $\bar{P}_T$ is a probability generating function it cannot be negative, hence showing we have an unphysical limit.
\end{itemize}
The cause of this behaviour is that we declared by fiat that the parameters $w_i$ must be scaled 
to give \eqref{rscaling}; however the parameters must also satisfy the constraints $\mu(\ell;{w_i}) > 0$ and $\sum^\infty_{l=0}  \mu(\ell;{w_i}) = 1$ for all values of the parameters and we have not ensured these constraints are compatible with \eqref{rscaling}. It is sufficient for our purposes to understand that the oscillatory terms prevent the continuum limit being taken for certain walk lengths so they are certainly unphysical; we therefore must scale them so that they disappear in the continuum. We therefore have as a lower bound on $Q_C$,
\bea
\tilde{Q}(\xi;\omega_i) \geq \xi^{-1/2}\left(1+\delta_{\Delta{C(K)},0}C+\sum_{j\in \tilde{J}_R}\frac{1}{(k_j-1)!(1-\sigma_j)}\upsilon_{j,-}^{-\frac{1-\sigma_j}{2}}\right)^{-1/2}.\nn\\\label{GeneLB}
\eea

We now consider the upper bound. Without loss of generality we may arrange that the indexing set $J_R$ has the property that if $j_1 < j_2$ then $\lambda_{j_1,\pm} < \lambda_{j_2,\pm}$. Furthermore define
\bea
&&\fl{H(x) = [\tilde{H}(x)] = [x^{-1/2}]}, \nn \\
&&\fl{D(x) = [\tilde{D}(x)] = [\frac{2\beta}{1-\sigma_I} (1-\chi(\tilde{H}(x)))^{-1} |\log( r_{I,+} x^\frac{\sigma_I-1}{2} (-\frac{1}{2}\log x)^{k_I-1})|]}, \\
&&\fl{k(x) = [\tilde{k}(x)] = [x^{-1/2}(1-\chi(\tilde{H}(x)))]}, \nn
\eea
where $I = \inf \tilde{J}_R$, i.e. $\lambda_{I,\pm}$ is the smallest length scale with a spectral dimension differing from one. Some modified versions of the above quantities will also be needed,
\bea
&&\fl{D_\pm(x) = [\tilde{D}_\pm(x)] = [\frac{2\beta}{1-\sigma_I} (1-\chi_\pm(\tilde{H}(x)))^{-1} |\log( r_{I,+} x^\frac{\sigma_I-1}{2} (-\frac{1}{2}\log x)^{k_I-1})|]},\\
&&\fl{k_\pm(x) = [\tilde{k}_\pm(x)] = [x^{-1/2}(1-\chi_\pm(\tilde{H}(x)))]}. \nn
\eea
It is clear that $\tilde{D}_-(x) \leq \tilde{D}(x) \leq \tilde{D}_+(x)$ and $\tilde{k}_+(x) \leq \tilde{k}(x) \leq \tilde{k}_-(x)$ which together with Lemma \ref{GenUBLem} allows us to conclude that the continuum limit is,
\bea
\label{GeneUB}
&&\fl{\tilde{Q}(x;\omega_i) \leq \xi^{-1/2}\Bigg[1-e^{G_- \log(1-\upsilon_I^{-s\beta})} +} \nn \\
&&\qquad\qquad\fl{\frac{e^{(G_+-1) \log(1-\upsilon_I^{-s\beta})}}{3\mathrm{cosech}^2(|\log\upsilon_I^\beta|(1-1/G_+))-\gamma+\sqrt{\gamma^2+1+2\gamma \coth(|\log\upsilon_I^\beta|/G_+)}} \Bigg]}
\eea
where we have defined $G_\pm(\xi,\lambda_1,\ldots) = \sum_{j\in \tilde{J}_R} \upsilon_{j,\pm}^{-\frac{1-\sigma_j}{2}}/(k_j-1)!$ and arranged that no oscillatory terms appear.

We are now in a position to prove Theorem \ref{mainresult} by analysing the behaviour of \eqref{GeneLB} and  \eqref{GeneUB} for various walk lengths. We first note that both the upper and lower bounds are controlled by the relative sizes of the quantities $\upsilon_j^{-1+\sigma_j/2}$. In particular the largest of these quantities, $\Upsilon\equiv\mathrm{sup}\{\upsilon_j^{-1+\sigma_j/2}: j\in\tilde{J}_R\}$, will determine the behaviour in a particular scale range; if we suppose $\Upsilon = \upsilon_j^{-1+\sigma_j/2}$ for some $j$ then the leading contribution to both the upper and lower bound will be proportional to $\xi^{-(1+\sigma_{j})/4}$. Using this we find the following behaviour,
\begin{itemize}
\item On very short scales corresponding to $\xi\gg\lambda_j^{-1}$ for all $j$ then the lower bound of $\tilde Q(\xi)$ is dominated by the $\xi^{-\half}$ behaviour, which together with the trivial upper bound means that taking the $\xi\to\infty$ limit leads to $d_s^0=1$ as in the previous sections.
\item On very long scales corresponding to $\xi\ll\lambda_j^{-1}$ for all $j$ then there will exist a length scale $\Lambda$ such that for $0< \xi < \Lambda^{-1}$, $\Upsilon = \upsilon_{\hat{j}}^{-1+\sigma_{\hat{j}}/2}$ where $\sigma_{\hat{j}} \leq \sigma_{j}$ for all $j$. Explicitly $\Lambda$ will be given by, 
\beq \Lambda=\sup\{\lambda_j^{-(2-\sigma_j)/(\sigma_j-\sigma_{\hat{j}})}\lambda_{\hat{j}}^{(2-\sigma_{\hat{j}})/(\sigma_{j}-\sigma_{\hat{j}})}:j\in \tilde{J}_R\}.\eeq
If we take the limit $\xi \rightarrow 0$ we obtain $d_s^\infty=(3-\sigma_{\hat{j}})/2$.
\end{itemize}

We see that the spectral dimension increases monotonically on successive length scales and the spectral dimension has measured values given by $d_s = (3-\sigma_{\hat{j}})/2$ therefore proving Theorem \ref{mainresult}.
\end{proof}

It is interesting to note that a constraint on the crossover behaviour exists for generic measures much as it did for the case of the double power law measure. In particular on intermediate scales there exists $J \in \tilde{J}_R$ such that $\xi\ll\lambda_j^{-1}$ for $j\leq J$ and $\xi\gg\lambda_j^{-1}$ for $j > J$. Hence  $\upsilon_j^{-1+\sigma_j/2} \ll 1$ if $j>J$ and so $\Upsilon=\mathrm{sup}\{\upsilon_j^{-1+\sigma_j/2}: j\in\tilde{J}_R, j\leq J\}$. This means there will exist a length scale $\Lambda_J$ such that for ${\lambda_{J+1}}^{-1}< \xi < {\Lambda_J}^{-1}$, $\Upsilon = \upsilon_{\hat{j}(J)}^{-1+\sigma_{\hat{j}(J)}/2}$ where we have defined $\hat{j}(J)$ implicitly by $\sigma_{\hat{j}(J)} \leq \sigma_{j}$ for all $j\leq J$. Of course this does not ensure that ${\lambda_{J+1}}^{-1} < {\Lambda_J}^{-1}$, indeed if the length scales $\lambda_j$ are not sufficiently separated this may not be true and we would not have a scale range in which this term dominated. The expression for ${\Lambda_J}$ is,
\beq \Lambda_J=\sup\{\lambda_j^{-(2-\sigma_j)/(\sigma_j-\sigma_{\hat{j}(J)})}\lambda_{\hat{j}(J)}^{(2-\sigma_{\hat{j}(J)})/(\sigma_{j}-\sigma_{\hat{j}(J)})} :j\in \tilde{J}_R, j \leq J\}\eeq
and so we may choose $\lambda_{J+1}$ independent of $\Lambda_J$ thereby allowing the scale range over which this behaviour exists to be arbitrarily large. This would result in an apparent spectral dimension of $\delta_s=(3-\sigma_{\hat{j}(J)})/2$.

Finally, an interesting application of the techniques used to prove Theorem \ref{mainresult} is that they allow one to analyse a wider class of combs than the class for which the results of \cite{Durhuus:2005fq} are valid. In particular it was proven in \cite{Durhuus:2005fq} that for a random comb the spectral dimension is $d_s = (3-\gamma_0)/2$ where $\gamma_0=\sup\{\gamma\geq 0: I_\gamma <\infty \}$ and $I_\gamma =\sum_{\ell =0}^\infty \mu_\ell\, \ell ^{\gamma}$. This was proved subject to the assumption that there exists $d>0$ such that, 
\beq
\label{RCassmp} \sum^{\infty}_{l=[x^{-1/2}]} \mu(\ell) \sim x^d
\eeq
as $x$ goes to zero.

Given the results in this section we see that $-\gamma_0$ may be interpreted as the abscissa of convergence for the Dirichlet series generating function. Furthermore, it is clear from our results that there are distributions where the assumption \eqref{RCassmp} does not hold and that we may use the techniques we have developed to analyse these cases. Recalling that for a random comb we may compute the spectral dimension using the relation \eqref{QCdsdef} then we must perform a similar analysis to that done for the continuum comb but now only scaling $x$ to zero.

Due to Landau's theorem there always exists a pole at the abscissa. Suppose it is of order $k$ and consider $1-\chi_\pm(u)$, 
\bea
&&\fl{1-\chi_\pm(u) = \left[r_{0,\pm} + \sum^\infty_{j=1} \mathrm{Re}[r_{j,\pm}] \cos(-{\tau_j} \log u) + \mathrm{Im}[r_{j,\pm}] \sin(-\tau_j \log u)\right] \frac{(\log u)^{k-1}}{(k-1)!} u^{-\gamma_0}} \nn \\
&&\qquad\qquad\fl{\equiv \Omega(u^{-1}; \{r_{j,\pm}\}) \frac{(\log u)^{k-1}}{(k-1)!} u^{-\gamma_0}} 
\eea
where $j$ runs over the poles on the line of convergence and $\Omega(y; \{r_{j,\pm}\})$ is implicitly defined in the second line. We have not included any poles of order less than $k$ since these will not be leading order as $y$ goes to zero and we have assumed the there are no poles of order greater than $k$ as then the above quantity would go negative due to the oscillating terms. Applying the same argument as we did for the continuum comb we obtain the lower bound,
\beq
\fl{\bar{Q}(x) \geq \left[ \frac{c'(k-1)!}{\Omega(x^{1/2}; r_{j,\pm}(1-\gamma_0 \pm i \tau_j)^{-1})}\right]^{1/2}\left(\half|\log x|\right)^{(1-k)/2} x^{-\frac{\gamma_0+1}{4}}.}
\eeq
where $c'$ is a constant. To obtain an upper bound we apply a very similar argument as used for the continuum comb; the only difference being that we choose,
\bea
H &=& [\tilde{H}] = \beta x^{-1/2} |\log x|, \nn \\
D &=& [\tilde{D}] = \beta (1-\chi(\tilde{H}))^{-1} |\log x|, \nn \\
k &=& [\tilde{k}] = x^{-1/2} (1-\chi(\tilde{H})), \\
D_\pm &=& [\tilde{D}_\pm] = \beta (1-\chi_\pm(\tilde{H}))^{-1} |\log x|, \nn \\
k_\pm &=& [\tilde{k}_\pm] = x^{-1/2} (1-\chi_\pm(\tilde{H})) \nn
\eea
and compute the behaviour as $x$ goes to zero of $G^{(0)}_\infty(x,(\tilde{D}-1)(\tilde{k}-1))$, $P_{\tilde{H}, *\tilde{D}_+}$ and the size of the set $\C A$ in \eqref{Aevent}. The result is,
\beq
\fl{\bar{Q}(x) \leq  \left[\frac{c''(k-1)!}{\Omega(\frac{x^{1/2}}{|\log x^\beta|}; r_{j,+})} \right]^{1/2} \left(|\log\left[ \frac{x^{1/2}}{|\log x^\beta|}\right]|\right)^{(1-k)/2} |\log x^\beta|^{\gamma_0/2}x^{-\frac{\gamma_0+1}{4}}.}
\eeq
where $c''$ is a constant. We see that we reproduce the result of \cite{Durhuus:2005fq} when $k=1$ and there are no poles on the line of convergence. If $k\neq 1$, the spectral dimension is the same as the $k=1$ case as only logarithmic corrections are introduced. If we allow poles to appear on the line of convergence then $\Omega(x; r_{j,\pm}) $ will have an oscillating $x$ dependence which, if the bounds above are tight enough, would also imply that the functional form of $\bar{Q}$ is not a power law and so the spectral dimension does not exist. Whether the above bounds are tight enough to make this conclusion we leave for future work. 

\section{Discussion}
\label{DiscussComb}
In this chapter we demonstrated that there exist models in which a CDT-like scale dependent spectral dimension could be shown to exist analytically. That we could do this was important, as all evidence for a scale dependent spectral dimension in CDT has been numerical in nature and therefore does not provide any understanding of the mechanism causing the reduction in dimensionality. We hope the work begun in this chapter may be extended to shed some light on this question.

We have given in Theorem \ref{mainresult} a reasonably complete classification of what behaviours one can find on a continuum comb and indeed it is likely that the cases not covered by the theorem do not have a well defined spectral dimension. The behaviour for the combs covered by Theorem \ref{mainresult} is fairly rich as there exists a cross-over region between the $UV$ and $IR$ behaviours in which a hierarchy of apparent spectral dimensions exist. One might expect that if one were able to extend the results of this chapter to more CDT like models one could find examples in which the spectral dimension again depends on the scale. Such an expectation is not unreasonable since the proof of the bijection between two-dimensional CDT and trees in \cite{Durhuus:2009sm} shows that any ensemble of trees is in bijection with a CDT-like theory. It is only when the random tree has the uniform measure that we obtain precisely CDT on the other side of the bijection. One could therefore imagine a variation of CDT, constructed from an ensemble of random trees with a measure that differs from the uniform measure and with a dependence on a length scale. This length scale could then be scaled while taking a continuum limit, as we have done in the work here. Such a model should display a spectral dimension that differs in the UV and IR.

One may query how closely related a random comb is to an actual model of quantum gravity. For the case of both CDT and random trees one may show that the evolution of the spatial volume with time is generated by a Hamiltonian. A crucial difference when one considers random combs is that the growth of a comb is not a Markovian process \footnote{In fact the growth process of a comb is Markovian if the tooth length distribution is of exponential form, however we have seen this does not lead to any interesting scale dependent behaviour of the spectral dimension.}; the probability of having a certain number of vertices at height $h$ from the root may not be computed from the state of the system at height $h-1$, meaning a Hamiltonian formulation of the evolution of a comb-like universe is not possible. An extension of the work described in this chapter, using the techniques of \cite{Durhuus:2006vk}, to the case of random trees would be very interesting as this would constitute an example of CDT-like dimensional reduction for a model of quantum gravity which admits a Hamiltonian formulation. As was briefly mentioned earlier, a random graph that is even more closely related to the random graphs of CDT was introduced in \cite{Durhuus:2009sm} by identifying all verticies belong to the same spatial slice. This leads to a line like graph with multiple links between each node. The question of whether the above phenomenon can be extended to this case would also be interesting to pursue.




\chapter{DT with Matter, String Theory and Branes\label{ChapString}}

In the previous chapters we considered the emergent dimension of spacetime in models of pure quantum gravity in addition to the disc function and the string susceptibility. This exhausts the list of observables that can be studied in such theories. We expect the theory to become richer upon the addition of matter and it is this topic we will address in this chapter. We will also discuss the relation of the metric-is-fundamental approach, including matter, to string theory. In fact the consideration of the string interpretation of these models will lead to a large class of two dimensional quantum gravity models known as $(p,q)$ minimal gravity or $(p,q)$ minimal string theory.

\section{Adding Matter using Matrix Models}
Of the approaches to discrete two dimensional quantum gravity considered in the preceding chapters the one which most easily admits a generalisation to models that describe the interaction of matter with gravity is the matrix model approach \footnote{This does not mean the matrix model approach allows any matter to be coupled. In general it is restricted to CFTs arising from spin systems such as the Ising model. Furthermore, transfer matrix techniques have not been pursued to the same extent as matrix models in quantum gravity and therefore could yet yield a description of matter interacting with gravity.}. An immediate consequence of this is that since the standard method used to compute amplitudes in CDT is the combinatorial method, it has so far not been possible to analytically investigate the consequences of adding matter to CDT. There have been a number of attempts to formulate CDT as a matrix model \cite{Benedetti:2008hc,Ambjorn:2008gk,Ambjorn:2008jf}, however these have so far been too complicated to solve or resisted the generalisations necessary to include matter. Hence, for the rest of this thesis we will work within the DT approach.
 
We begin by considering a naive method by which one can construct a matrix model that describes matter interacting with two dimensional gravity. In particular we know that the Ising model on a fixed lattice possesses a critical point at which the theory is described by a continuum quantum field theory. Therefore if we attach a spin to each triangular face in a triangulation and take a scaling limit in which both the lattice and the Ising model on it become critical, we expect we might recover the critical Ising model interacting with two dimensional gravity. For such a model, we expect the partition function for the discrete theory to be of the form,
\beq
Z^\mathrm{DT}_h = \sum_{T\in\mathrm{Tri}_h} e^{-\sum_{<i,j>} J \sigma_i\sigma_j-\Lambda A(T)},
\eeq
where $J$ is the coupling between spins and $\sigma_i$ is the spin on the $i$th triangle in the triangulation, which can take on values of $\pm 1$. For purposes of illustration and to make contact with graph theory terminology it is usual to refer to the spins as ``colours'' so that an assignment of spins to a triangulation becomes a ``colouring'' of the graph. For a triangulation $T$, let $n_{+-}(T)$ denote the number of edges shared by triangles differing in colour and let $n_p(T)$ be the total number of edges. The action may then be written as,
\beq
J\left(n_{+-}(T) - \left[n_p(T) - n_{+-}(T)\right]\right) - \Lambda |T| = 2J n_{+-}(T) - (\Lambda + 3J/2)|T|,
\eeq
which implies that the partition function, assuming the couplings are small, may be written as,
\beq
\label{ColouredDT}
Z^\mathrm{DT}_h = \sum^\infty_{n=0} \sum^\infty_{m = 0} \C{N}_h(n;m) \tilde{g}^{n} k^m,
\eeq
where $k = \exp{(-2J)}$, $\tilde{g} = \exp(-\Lambda - 3J/2)$ and $\C{N}_h(n;m)$ is the number of triangulation containing $n$ triangles and $m$ edges separating triangles of differing colours. It is this partition function we must realise using a matrix model.

\begin{figure}[t]
\centering
\subfigure[The Feynman rules generated by the matrix model \eqref{2MMZ} \cite{EynardReview}.]{
\includegraphics[scale=0.5]{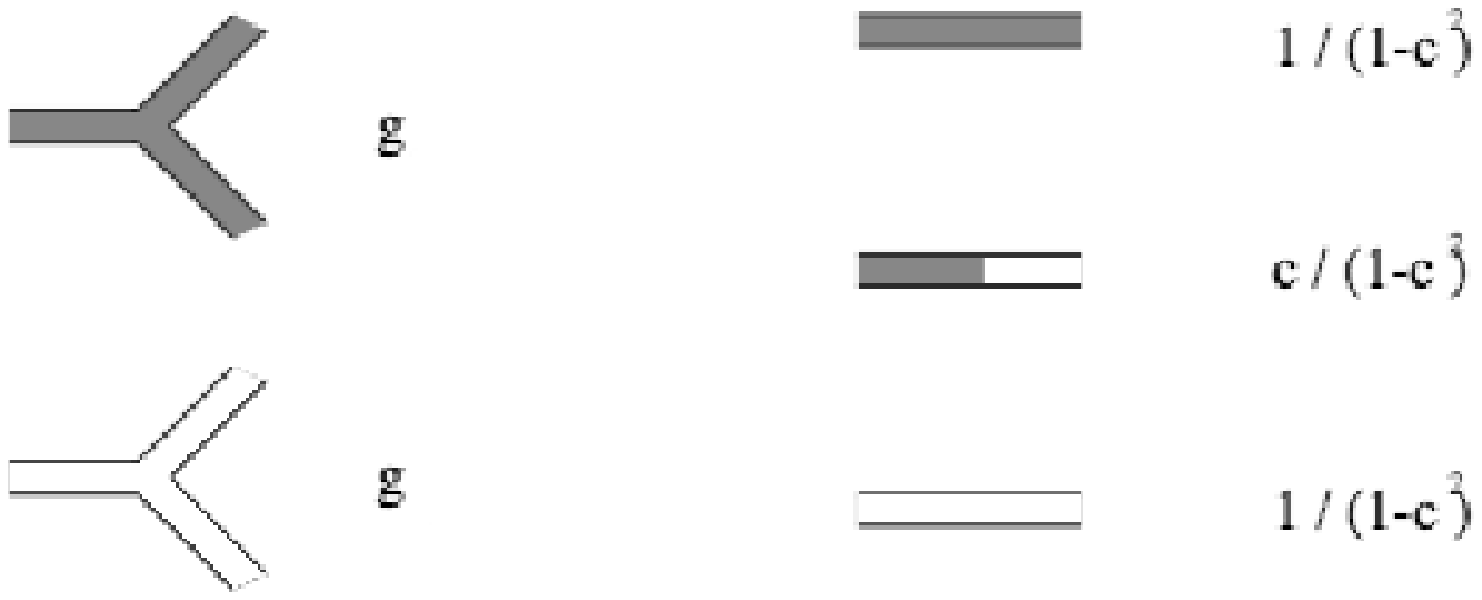}
\label{IsingFeynmanRules}
}
\qquad
\subfigure[An example of a Feynman diagram generated by the Feynman rules in fig \ref{IsingFeynmanRules} \cite{EynardReview}.]{
\includegraphics[scale=0.35]{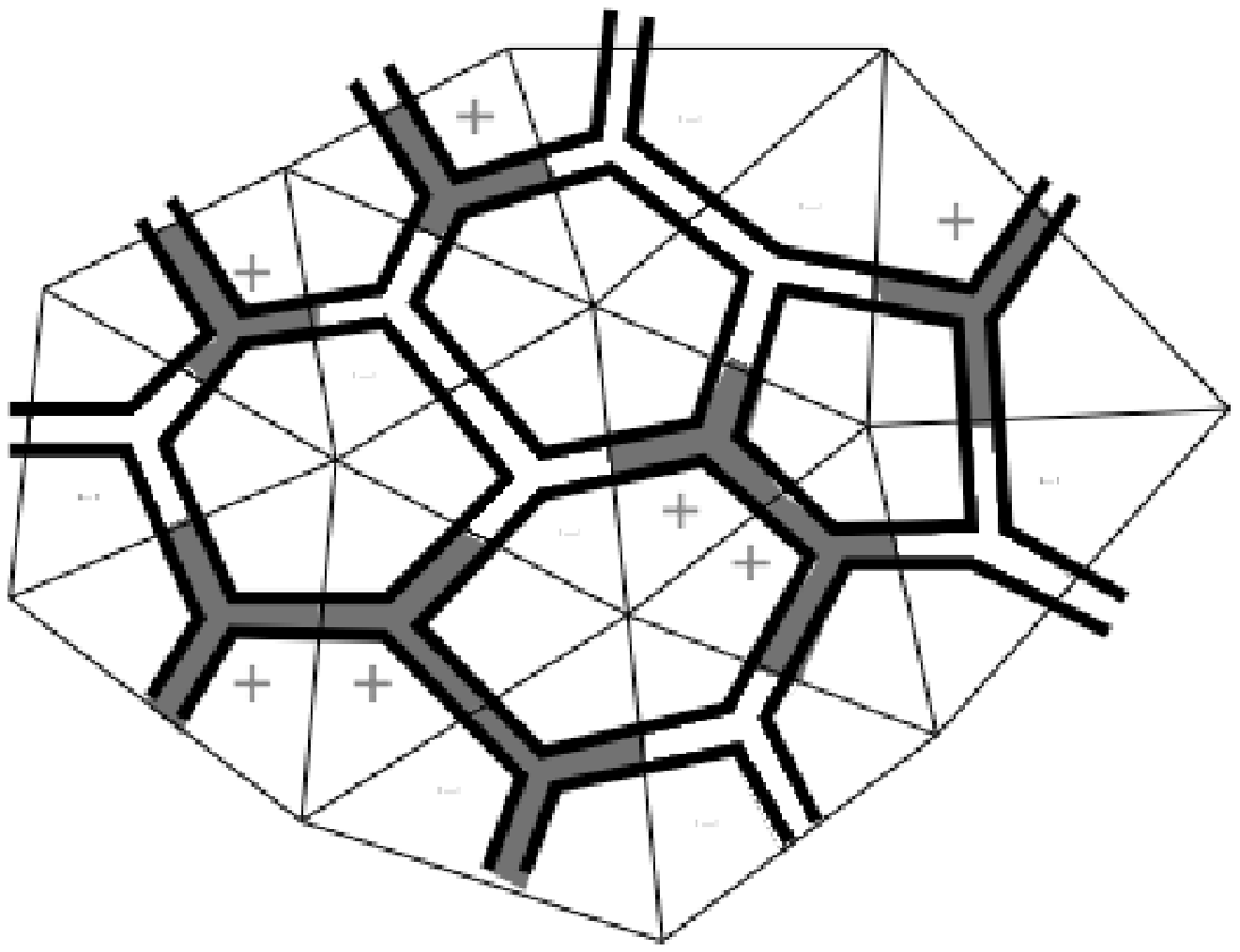}
\label{IsingFeynmanGraph}
}
\end{figure}
The key observation is that we can generate labelled graphs (i.e. containing colours) using a matrix model by considering models which contain more than one matrix. A matrix model containing more than one matrix can produce Feynman diagrams containing more than one type of ``particle'' species, which can then be considered as a label. As an example, consider the matrix model defined by,
\bea
\label{2MMIsing}
\Z=\int dM_1\,dM_2\exp\left(-N\Tr \left(-M_1 M_2+\half(M_1^2+M_2^2)+\frac{g}{3}(M_1^3+M_2^3)\right)\right).
\eea
It is straightforward to see that this model will produce Feynman diagrams containing propagators for the $M_1$ and $M_2$ matrices, which will be denoted by grey and white double lines respectively. Furthermore, due to the first term in the action, the model also produces a mixed propagator that interpolates between grey and white. Finally, the cubic interaction terms will produce grey and white 3-vertex interactions. The various propagators and vertices are shown in fig \ref{IsingFeynmanRules} together with their Feynman rules. An example of a Feynman diagram produced by this model is shown in fig \ref{IsingFeynmanGraph}.

To simplify the algebra let the contribution of a colour changing propagator be $\gamma$ and that of a colour preserving propagator be $\delta$. The value of a Feynman diagram composed of $n$ triangles with $n_{+-}$ colour changing propagators and $n_p$ propagators in total is,
\beq
\frac{1}{\Omega} g^n \gamma^{n_{+-}} \delta^{n_p - n_{+-}} = \frac{1}{\Omega} (g\delta^{3/2})^n \left(\frac{\gamma}{\delta}\right)^{n_{+-}},
\eeq
where $\Omega$ is the symmetry factor of the diagram. We can now compute the free energy for the matrix model by summing over all connected graphs and for each graph summing over $n_{+-}$. If we let $g\delta^{3/2} = \tilde{g}$ and $\frac{\gamma}{\delta} = k$ then we obtain precisely \eqref{ColouredDT}.

\subsection{The Continuum Limit and Multicritical Points}
The matrix model \eqref{2MMIsing} belongs to a general class of two matrix models of the form,
\bea
\Z=\int dM_1\,dM_2\exp\left(-N\Tr \left(- cM_1 M_2+V_1(M_1)+V_2(M_2)\right)\right).
\eea
These models have been solved \cite{Eynard:2002kg} using loop equations yielding an equation for the resolvent $W_{M_1}^{(0)}(x) = \avg{\Tr(x-M_1)^{-1}}$;
\beq
\label{Mastereqn}
(V'_1(x)-y)(V'_2(y)-x)-P(x,y)+1 = 0,
\eeq
where $y = V_1(x) - W_{M_1}^{(0)}(x)$ and $P(x,y)$ is a polynomial in $x$ and $y$ with unknown coefficients. The equation is sometimes known as the master equation. The equation for $W_{M_2}^{(0)}(y)$ can be obtained by letting $x = V_2(y) - W_{M_2}^{(0)}(y)$ in the master equation. If we specialise to the Ising matrix model considered above we see that the loop equation resulting from \eqref{Mastereqn} is a 3rd order polynomial in $W_{M_1}^{(0)}(x)$. In order to compute the unknown constants appearing in $P$ we enforce the condition that the master equation must be a genus zero curve. We may then compute the unknown constants by requiring that the resulting solution for the resolvent has the correct behaviour as $x \rightarrow \infty$, i.e. there is a sheet of the resolvent in which $W_M^{(0)}(x) \sim 1/x$. In order to compute the critical point of the matrix model we again consider the quantity $w_1(g,c)$ and look for the value of $g$ at which $w_1$ is non-analytic. In this case however we see that since we have a second variable, $c$, the matrix model will possess a critical line in $g,c$ space rather than having a single critical point. In the presence of a critical line there is the possibility that there is a higher order critical point lying on the line. This is in fact the case for the Ising matrix model. 

We may qualitatively understand these critical lines and points by first letting $c = 0$; in this case the Ising spins do not interact and we have two copies of the one matrix model considered previously. We therefore expect that any critical point obtained by tuning $g$ so that we hit the critical line, corresponds to pure gravity. At the critical point however, $c$ has a very particular value and therefore we expect it is at this point that we obtain a theory of critical Ising matter interacting with gravity. Such higher order critical points are known as multi-critical points. 

One may wonder what the two matrix model with general potentials corresponds to. In this case there will be a number of free parameters contained in each potential and therefore we have the possibility of a large number of critical hypersurfaces of decreasing dimension. By tuning the coupling so that the model lies on some critical hypersurface we can obtain a variety of continuum models \cite{Ginsparg:1993is,EynardReview}. We shall see that in this way we may obtain a family of different field theories known as $(p,q)$ minimal models interacting with gravity. The Ising model corresponds in this classification to the $(3,4)$ model. 

\section{Conformal Field Theory and Minimal Models}
\label{CFTandMMs}
One manner in which one could specify a particular quantum field theory is to give the symmetries of the theory. One can then specify the dynamical objects in the theory, by which we mean we specify which representation of the symmetry group they form. These two pieces of information largely fix the dynamics of the theory. In the usual formulation of quantum field theory these properties are introduced using the action and Feynman path integral, which can then be used to construct a perturbative series that computes correlation functions. However, in two dimensions the group of conformal transformations is infinite dimensional. If we have a QFT invariant under conformal transformations, this allows correlation functions to be computed without the use of an action or path integral. We will now review this construction following \cite{YB}.

It is often customary when discussing two dimensional field theories to use complex coordinates, defined by $z = x + i y$ and $\bar{z} = x - iy$. A field $\phi(z,\bar{z})$ is defined to be {\emph{primary}} if under a conformal transformation, $z \rightarrow w= w(z)$, $\bar{z} \rightarrow \bar{w}= \bar{w}(\bar{z})$, it transforms as,
\beq
\phi'(w,\bar{w}) = \left(\frac{dw}{dz}\right)^{-h} \left(\frac{d\bar{w}}{d\bar{z}}\right)^{-\bar{h}} \phi(z,\bar{z}).
\eeq
Fields which do not transform in this way are called secondary, of which an important example is the stress-energy tensor. The importance of the stress-energy tensor is that its Fourier coefficients generate conformal transformations. In terms of complex coordinates the stress energy tensor has two non-zero components that lie on the diagonal, we denote them $T(z)$ and $\bar{T}(\bar{z})$. The Fourier expansion of the stress energy tensor may be written as $T(z) = \sum_{n \in \BB{Z}} z^{-n-2}L_n$ and $\bar{T}(\bar{z}) = \sum_{n \in \BB{Z}} \bar{z}^{n-2} \bar{L}_n$. We will refer to $L_n$ and $\bar{L}_n$ as the generators of the conformal transformations. The algebra for the generators of conformal transformations is,
\bea
\left[L_n,L_m\right] &=& (n-m)L_{n+m} + \delta_{n+m,0} \frac{c}{12}n(n^2-1), \nn \\
\left[\bar{L}_n,\bar{L}_m\right] &=& (n-m)\bar{L}_{n+m} + \delta_{n+m,0} \frac{c}{12}n(n^2-1), \nn \\
\left[L_n,\bar{L}_m\right] &=& 0,
\eea
and is known as the Virasoro algebra, the real number $c$ is known as the central charge.

Since we work in a complexified Euclidean space the choice of how to slice the space into surfaces of constant time is arbitrary. This is equivalent to the freedom one has in splitting up a statistical mechanical system into subsystems in order to define a transfer matrix. It is useful in conformal field theory to take the surfaces of constant time to be concentric circles centred on the origin. The Hamiltonian for the system then generates the transformation $z \rightarrow \lambda z$, for $\lambda \in \BB{R}$. The origin is taken to be at $t=-\infty$ and $t=\infty$ is at complex infinity. This setup is known as radial quantisation. Because the states at $t=-\infty$ are localised at the origin there is a one-to-one correspondence in conformal field theories between local operators and states known as the state-operator correspondence. Hence for every primary field there exists a state known as a primary state. The primary states are the highest weight states when constructing the highest weight representations of the Virasoro algebra. The states obtained from the primary state by use of the lowering ladder operators, $L_n$ and $\bar{L_n}$ where $n<0$, are referred to as descendant states and they in turn give rise to fields known as descendant fields. A primary field together with its descendants is known as a conformal family.

A final important ingredient is the operator product algebra of the fields. This algebra arises by considering the correlation function of a product of local operators and studying its behaviour as two of the operators are brought towards the same point. In general the correlation function will diverge as the operators are brought together, however we can characterise this divergence by use of a Laurent like series in which the coefficients are other local operators. This series is known as the operator product expansion (OPE). Crucially which operators appear as coefficients in the OPE is independent of the other operators appearing in the correlation function, i.e.,
\beq
\avg{\ldots \C{O}_1(z_1) \C{O}_2(z_2) \ldots} = \sum^{\infty}_{n=-\infty} \frac{\avg{\ldots \C{O}_n(z_1) \ldots}}{(z_1-z_2)^n}.
\eeq
Since the OPE is independent of the other operators appearing in the correlation function, this relation is usually written as,
\beq
\C{O}_1(z_1) \C{O}_2(z_2) = \sum^{\infty}_{n=-\infty} \C{O}_n(z_1){(z_1-z_2)^n}.
\eeq
The operator product algebra is defined as the collection of all OPEs of the fields in the theory.
It turns out that the operator product algebra is entirely fixed by the value of the central charge together with the 3-point correlation functions of the theory.

In general the OPE of two operators will include other local operators in the expansion besides the two original operators and their descendants. Generically this will mean that in order to have an operator product algebra that closes the theory must contain an infinite number of primary fields. However, for special values for the central charge there exist finite sets of primary fields for which the operator algebra closes. These special conformal field theories are known as minimal models. In particular, given two positive coprime integers $p,q$ such that $p <  q$ then the conformal field theory with central charge
\beq
c = 1- 6\frac{(q-p)^2}{pq}
\eeq
and whose primary fields have weights given by
\beq
h_{r,s} = \frac{(qr-ps)^2 - (q-p)^2}{4pq},
\eeq
where $r$ and $s$ are integers satisfying, $1 \leq r < p$ and $1 \leq s < q$, is known as the $(p,q)$ minimal model.

\section{Boundary States of Minimal Models}
\label{MMBCFT}
To study minimal models in the presence of boundaries it is most useful to consider the case of the minimal model on a finite length cylinder. The boundaries of the cylinder carry boundary conditions we call $\alpha$ and $\beta$. Using real coordinates $x$ and $t$ we will analyse the case when the boundaries of the cylinder are curves of constant $t$. A necessary condition that any theory must satisfy at a boundary is that there is no flow of energy or momentum across the boundary. In terms of the components this requirement corresponds to $T_{tx} = 0$ at the boundary. If we complexify the coordinates by changing to the coordinates $z = t+ix$ and $\bar{z} = t-ix$, then the condition $T_{tx}=0$ becomes $T(z) - \bar{T}(\bar{z}) = 0$. Since the boundary is a curve of constant time, the system at the boundary should correspond to some state in the Hilbert space. We will refer to these boundary states as $\ket{\alpha}$ and $\ket{\beta}$. As was mentioned previously, calculations in conformal field theory are usually done in a coordinate system in which the slices on constant time form concentric circles centred on the origin. We can transform to such a coordinate system via the map $w = e^{\frac{2\pi}{L} z}$ resulting in the cylinder being mapped to an annulus. The condition on the boundary states can then be written as,
\beq
(L_n - \bar{L}_{-n})\ket{\alpha} = 0.
\eeq
The solutions to this equation are known as the Ishibashi states and have the explicit form,
\beq
\ket{j} = \sum_n \ket{j;n}\otimes \ket{\bar{j};\bar{n}}
\eeq
where the sum is over all states in the irreducible representation of the Virasoro algebra associated to the primary field $j$. Note that there is one Ishibashi state for each primary field in the theory. We therefore appear to have a space, with dimension equal to the number of primary fields in the theory, of boundary states. However, it turns out that there exist only a finite number of consistent boundary states, which are known as Cardy states and are in one-to-one correspondence with the primary fields in the theory.

\section{String Theory}
The central hypothesis of string theory is that fundamental particles may be modelled as quantum strings \footnote{This is unlikely to be the correct underlying principle but it is how the subject developed historically.}. The space in which the string is embedded is known as the target space. The simplest action for the string is known as the Polyakov action and takes the form,
\beq
S_P[X,g,h] \propto \int_{M_h} d^2\sigma \sqrt{g} \partial_a X^\mu(\sigma) \partial^a X_\mu(\sigma),
\eeq
where $X$ is the map embedding the string in the target space and $g_{ab}$ is the metric on the worldsheet, $M_h$, which is a two-dimensional Riemannian manifold of genus $h$. We have taken both the worldsheet and target space to be Euclidean as then the corresponding quantum theory is well defined. Later we will show that the quantum theory is equivalent to a theory on a flat worldsheet and so may be Wick rotated back to a Lorentzian theory.

The Polyakov action corresponds to a string whose excitations are purely bosonic and is therefore known as the bosonic string. 
The classical bosonic string has a number of symmetries that are manifest in the Polyakov action;

{\bf{Target space Poincare Symmetry:}} The action is invariant under the transformation $X^\mu \rightarrow {\Lambda^\mu}_\nu X^{\nu} + a^\mu$, where $\Lambda^\mu$ is a Lorentz transformation.

{\bf{Worldsheet Diffeomorphisms:}} The action is invariant under worldsheet diffeomorphisms, $\sigma^a \rightarrow f^a(\sigma^b)$, where $f$ is a smooth function. Under such transformations the embedding fields $X$ behave as scalars.

{\bf{Weyl Symmetry:}} Finally the action is also invariant under Weyl rescalings of the metric, $g_{ab}(\sigma) \rightarrow e^{2\phi(x)}g_{ab}(\sigma)$, where $\phi(x)$ is a smooth function of $x$.

To quantise the string we integrate over the embedding fields $X$ and the worldsheet metric $g$. We must only count each physically equivalent configuration once and it is therefore important to decide which symmetries of the action we consider gauge symmetries. For now we will consider Weyl transformations and worldsheet diffeomorphism to be gauge symmetries. This has the consequence that we must divide by the volume of these symmetry groups in the path integral. The expression for the partition function is then, 
\beq
Z_h =  \int \frac{[dg]}{\mathrm{Diff} \times \mathrm{Weyl}} \int [dX]e^{-S_P[X,g,h]}.
\eeq
In order to perform this integral we must define the measures $[dX]$ and $[dg]$. One might already note that the partition function for the string is that of the metric-is-fundamental approach but with the addition of a set of scalar fields and a lack of cosmological constant. 

For now we will keep the measures in the integral as formal objects and instead work on gauge-fixing the integral. The diffeomorphism gauge symmetry allows $g_{ab}$ to be brought to the form $e^{2\phi} \tilde{g}_{ab}(\tau_i)$ where $\tilde{g}$ depends only on a finite number of parameters $\tau_i$, known as the moduli. Furthermore, due to the extra Weyl symmetry in this theory the metric may be brought to the form $\tilde{g}_{ab}(\tau_i)$. Inserting the Faddeev-Popov determinant necessary to achieve this change of variable we obtain,
\beq
\label{gaugefixedZ}
Z_h =  \int d\tau \int \frac{[d\phi][d\zeta]}{\mathrm{Diff} \times \mathrm{Weyl}} \int [db][dc] \int [dX]e^{-S_P[X,\tilde{g},h]-S_\mathrm{ghost}[b,c,\tilde{g},h]},
\eeq
where $b$ and $c$ are ghost fields. In order for the integrals over $d\phi$ and $d\xi$ to cancel the volume of the gauge group we need the integral,
\beq
\label{Xandghosts}
Z'_h[\tilde{g},\phi, \xi] = \int [db][dc] \int [dX]e^{-S_P[X,\tilde{g},h]-S_\mathrm{ghost}[b,c,\tilde{g},h]},
\eeq
to be independent of $\phi$ and $\xi$. The action itself is invariant under Weyl transformations and diffeomorphisms as these are classical symmetries. Hence the only potential problems lie in the definition of the measures $[dX]$, $[db]$ and $[dc]$. If these measures have a dependence on $\phi$ or $\xi$ then we will have a Weyl or diffeomorphism anomaly, respectively. 
In fact, it happens that the measures are invariant under diffeomorphisms but the case of Weyl transformations is more complicated and it turns out that there is indeed a Weyl anomaly present. The Weyl anomaly can be most easily seen by recalling that a field theory is invariant under conformal and hence Weyl transformations if the trace of the stress energy tensor vanishes i.e. $\avg{T^\mu_\mu} = g_{\mu \nu} \avg{T^{\mu \nu}} = 0$.

If one computes the trace of stress energy tensor for the theory defined by \eqref{Xandghosts} one obtains $\avg{T^a_a} = -\frac{D-26}{12} R[\tilde{g}]$, where $D$ is the number of scalar embedding fields present and $R$ is the Riemann tensor for the worldsheet. At this point it is easy to see how the critical dimension of the bosonic string arises; if the number of dimensions $D$ is twenty-six then the Weyl anomaly is not present and the integral over $\phi(x)$ trivially cancels the remaining gauge group volume. This also has a number of important consequences. Firstly, since the geometry of the worldsheet is entirely a gauge degree of freedom we need not worry about which particular class of worldsheet geometries, for example Euclidean of Lorentzian, we sum over. Furthermore, since the worldsheet metric can be made flat everywhere, the computation of worldsheet correlation functions, which are important in the calculation of S-matrix elements, are reduced to correlation functions in a flat space conformal field theory. 

\begin{figure}[t]
  \begin{center}
    \includegraphics[width=10cm]{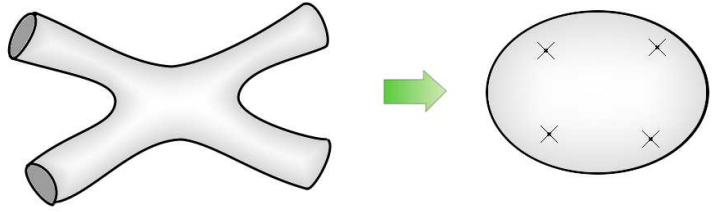}
    \caption{The worldsheet corresponding to the interaction of two strings may be transformed via a conformal transformation to a worldsheet of equivalent topology but with all asymptotic string states transformed into local operators.}
    \label{stringperturb}
  \end{center}
\end{figure}

Having defined the partition function for the string, an ansatz for how to compute S-matrix elements is required. If one considers a diagram for the interaction of two strings one can note that, by utilising the conformal symmetry of the worldsheet theory, the strings propagating from infinity may be encoded in local worldsheet operators, which we denote by $V$. Hence the string interaction diagrams may be reduced to the computation of correlation functions in the worldsheet theory. This is shown in fig \ref{stringperturb}. Furthermore, the genus of the worldsheet determines the power of the string coupling appearing in the amplitude. The ansatz is therefore that S-matrix elements for string states may be computed perturbatively by computing
\beq
\avg{V_1 \ldots V_N} = \sum^\infty_{h=0} g_s^{2-2h} \avg{V_1 \ldots V_N}_h,
\eeq
where $\avg{V_1 \ldots V_N}_h$ is the $N$-point correlation function on a surface of genus $h$ and $g_s$ is the string coupling constant. Clearly this construction fails when the string coupling becomes of order one. In such cases one would need a non-perturbative definition of string theory in order to proceed. Unfortunately, such a formulation is currently incomplete.

Although, a full definition of string theory is still lacking there has been much progress in understanding some of the crucial ingredients of such a theory. In particular it was found that Dirichlet boundary conditions of open strings correspond to the interaction of a string with a dynamical object which has become known as a D-brane. An important question in any string background then is what is the spectrum of D-branes, as this will have consequences for the non-perturbative physics.

Another aspect in which the above formulation is incomplete is that the target space was taken to be flat Minkowski space of dimension $D$. Obviously there are other spacetimes of great importance that arise as solutions to general relativity; namely the black hole spacetimes such as the Schwarschild metric and time dependent solutions, relevant to cosmology, such as the FRW metric. In particular it has proved difficult to gain any insight into the Big Bang singularity due to the difficulty of computing in a time dependent background.

Finally, a point related to the problem of backgrounds is that in the most straightforward interpretation given in this section, the number of dimensions in the target space is determined by the Weyl anomaly. However, the Weyl anomaly is not as stringent as this; it only requires that the overall central charge of the gauge fixed theory is zero. The ghost sector will always arise during the gauge fixing process and comes with central charge of $-26$, we therefore must start with a worldsheet theory with a central charge of $26$. There are many ways one could arrange this besides requiring $26$ non-interacting scalar fields to live on the worldsheet and indeed, one must modify the string model in this section as there are clearly not 26 non-compact dimensions! The most studied way to resolve this problem is to postulate that some of the target space dimensions are compact; in the case of the superstring these compact dimensions are usually taken to be a Calabi-Yau manifold. However there is nothing preventing us from cancelling the anomaly instead by introducing a conformal field theory on the world sheet with no geometric interpretation at all; this would be interpreted as degrees of freedom internal to the string. This will form the topic of the next section.

\subsection{Liouville Theory and Non-Critical String Theory}
Usually the dimension of the target space is defined to be equal to the number of worldsheet scalar fields and we can therefore formulate string theory in a lower number of dimensions by incorporating a CFT with a non-geometric interpretation into the worldsheet theory. There is a well studied CFT which has a tuneable central charge satisfying the constraint $c\geq 25$; it is known as Liouville theory. The action for Liouville theory takes the form,
\beq
\label{LiouvilleAction}
S_\mathrm{Liouville}[\phi,h] = \frac{1}{4\pi} \int d^2\sigma \sqrt{g}\left[ \partial_a \phi\partial^a \phi + Q\phi R[g] + 4\pi \mu e^{2 b \phi}\right],
\eeq
where $Q = b+1/b$ and $b$ and $\mu$ are free parameters. The field $\phi(\sigma)$ transforms as $\phi(\sigma) \rightarrow \phi(\sigma) - \omega(\sigma) /(2b)$ under the Weyl transform $g \rightarrow \exp(2\omega(\sigma)) g$. More generally under a conformal transformation $z = z(w)$, we have $\phi(z) = \phi'(w) - \frac{Q}{2} \log |\frac{dz}{dw}|^2$. Finally, the metric appearing in this expression is defined to be Euclidean. This theory may be quantised, while preserving the classical conformal symmetry it possesses, to produce a theory with central charge $c = 1+ 6Q^2$ and primary fields of the form $V_{\alpha}=e^{2\alpha\phi}$.

If we include the Liouville theory in the worldsheet theory then we are left with a central charge of $c\leq 1$ and so we still require some other sector in order to completely cancel the Weyl anomaly. In the case when $c=1$ a natural choice is to include a single scalar field. This is known as the $c=1$ string and has been studied quite extensively, see  \cite{Nakayama:2004vk} and references therein. However, if we choose $b$ in \eqref{LiouvilleAction} such that $c<1$ then instead of scalar fields, the natural candidate for the matter sector is a $(p,q)$ minimal model. The resulting string theory is known as $(p,q)$ minimal string theory. It is these string theories that will be studied in the remainder of this thesis.

One might wonder, since the $(p,q)$ minimal string has no embedding fields, if we have lost all geometric interpretation of the target space. However, there is a geometric way of interpreting the $(p,q)$ strings. The $(p,q)$ minimal models often arise as the continuum theory describing a statistical mechanical lattice model \cite{EynardReview,Ginsparg:1993is}, such as the Ising or Potts model, at its critical point. Such lattice models assign a degree of freedom to each lattice point which can take values in a finite set $G$ of points. If we interpret this degree of freedom as an embedding field then these model describe the embedding of the lattice into $G$. Hence we can think of the $(p,q)$ minimal strings as being strings embedded in a discrete target space.

We may go further and show that the Liouville sector of the worldsheet theory also has a geometric interpretation. Let us turn our attention back to \eqref{gaugefixedZ} in which we have treated both the Weyl transformations and worldsheet diffeomorphisms as gauge transformations. Because of this it is important that there is no anomaly in either of these symmetries. However, it is possible to construct a different theory if we drop the interpretation of the Weyl transformations as a gauge symmetry. In this case there is nothing preventing the appearance of a cosmological constant term, resulting in,
\beq
Z_h =  \int d\tau \int \frac{[d\phi][d\zeta]}{\mathrm{Diff}} \int [db][dc] \int [dX]e^{-S_P[X,\tilde{g}]-S_\mathrm{ghost}[b,c,\tilde{g},h]-\int_{M_h} d^2\sigma \sqrt{g} \mu}.
\eeq
In analysing this theory it will be useful to give a concrete definition of the measures in the integral, which we do by following \cite{Mansfield:1990tu,Polyakov:1981rd,D'Hoker:1988ta}. To define a measure on a space of functions we must introduce a metric in that space. For the two dimensional spacetime metric $g$ this metric is unique if we require it to respect diffeomorphism invariance and contain no dependence on derivatives of $g$; a condition known as ultralocality. The metric then takes the form,
\beq
(\delta g,\delta g') = \int_{\C M} d^2\sigma \sqrt{g} \left(u \delta g_{a b} \delta g'_{c d} g^{a b}g^{c d} + \delta g_{a b} \delta g'_{c d} g^{a c}g^{b d} \right),
\eeq
where $u$ is an arbitrary constant. Since this theory lacks a Weyl symmetry the diffeomorphism gauge symmetry only allows $g_{ab}$ to be brought to the form $e^{2\phi(\sigma)} \tilde{g}_{ab}(\tau_i)$. Inserting the Faddeev-Popov determinant for this transformation we obtain,
\beq
\label{gaugefixedZnoWeyl}
Z_h =  \int d\tau \int \frac{[d\phi][d\xi]}{\mathrm{Diff}} \int [db][dc][dX] e^{-S_\mathrm{ghost}[b,c,\tilde{g},h]-\int_{\C M} d^2 \sigma \sqrt{g}(\sigma) \mu}.
\eeq
We can integrate over diffeomorphisms $\xi$ as both the integrand and the measures are invariant under them. However, as we have seen, the measure for the ghosts and embedding fields still contains a dependence on $\phi$. Let us consider how this works explicitly. The metric on the space of embedding functions $X$ is unique if we require it to respect diffeomorphism invariance and be ultralocal. The metric then takes the form,
\beq
(\delta X,\delta X') = \int_{M_h} d^2\sigma \sqrt{g} \delta X(\sigma) \delta X'(\sigma),
\eeq
with a similar expression for the ghosts. As can be seen, both of these metrics have a dependence on $\phi$. Using this definition for the measures it is possible to show that the $\phi$ dependence of the measures is given by,
\bea
\int [dX]_{e^\phi \tilde{g}} &=& \int [dX]_{\tilde{g}} e^{\frac{D}{48 \pi} S_{L}[\phi]}, \nn \\
\int [dbdc]_{e^\phi \tilde{g}} &=& \int [dbdc]_{\tilde{g}}e^{\frac{-26}{48 \pi} S_{L}[\phi]}
\eea
where $S_L$ is the Liouville action with $Q=1$. We therefore obtain,
\beq
Z_h =  \int d\tau \int [d\phi] \int [db][dc] \int [dX]e^{-S_P[X,\tilde{g}]-S_\mathrm{ghost}[b,c,\tilde{g},h]-\frac{D-26}{48 \pi}S_L[\phi,\tilde{g}]},
\eeq
where the original cosmological constant has been absorbed into the cosmological constant appearing in the Liouville action. Naively, we appear to have solved the problem of how to integrate over $\phi$, however this conclusion is premature. The measure for $\phi$ is inherited from the measure on $g_{ab}$ and therefore depends on $\phi$ itself. It has the form, 
\beq
(\delta \phi, \delta \phi') = \int_{M_h} d^2\sigma e^{2\phi} \sqrt{\tilde{g}} \delta \phi \delta \phi'.
\eeq
The dependence on $\phi$ makes it difficult to split the measure into a product of Fourier coefficients which can easily be integrated over. In order to complete the calculation of $Z$ we must relate this $\phi$ dependent measure to one that is independent of $\phi$. It was conjectured in \cite{Distler:1988jt} that the effect of changing from this metric to a flat one, given by,
\beq
(\delta \phi, \delta \phi') = \int_{M_h} d^2\sigma \sqrt{\tilde{g}} \delta \phi \delta \phi',
\eeq
would only renormalise the constants appearing in the Liouville action. This conjecture is borne out by comparison to the matrix model approaches. We therefore see that although classical gravity is entirely topological, the measure in the path integral generates dynamics for the metric. Furthermore, the Liouville action we added to the worldsheet at the beginning of this section can be understood as describing the degrees of freedom internal to the worldsheet that correspond to the scale factor of the metric. The consequence of this on the statistical lattice model picture of the minimal models is that the Liouville field describes the configuration of the lattice itself. We therefore see how the program of DT naturally emerges from the path integral.

\subsection{Wave-functions and Correlation functions in Liouville Theory}
\label{superspaceAndCorrelationN}
We will now discuss an approximation, known as the mini-superspace \footnote{The term ``superspace'' in this context refers to the space of geometries of the worldsheet rather than a supersymmetric space.} approximation that has proved useful in the study of Liouville theory. This approximation is obtained by neglecting all oscillations of the string, leaving the zero modes as the only dynamical quantity. Beginning with the Liouville action we may obtain the Hamiltonian and by setting the transverse oscillations to zero we obtain the mini-superspace Hamiltonian,
\beq
H=\frac{1}{2}p^2+ 2 \pi  \mu e^{2b\phi_0},
\eeq
where $\phi_0$ is the zero mode of $\phi$ and $p$ is its conjugate momentum. It will turn out to be useful to study this Hamiltonian using the Schrodinger representation in which states in the Hilbert space becomes functionals of the field variable. The energy eigenstates then satisfy,
\beq
\label{miniH}
\left(-\frac{1}{2}\left(\frac{\partial}{\partial \phi_0} \right)^2+2 \pi \mu e^{2b\phi_0}\right) \psi[\phi] = E^2 \psi[\phi].
\eeq
An interpretation of the wave function for the zero mode may be found by consideration of the state-operator correspondence. Although discussed previously in Section \ref{CFTandMMs}, the state-operator correspondence can be given a more explicit realisation using path integrals \cite{Polchinski1}. Working in the radial quantisation picture we can associate any local operator to a wave-functional by the following map,
\beq
\Psi[\phi_b] = \int^{\phi_b} [d \phi] e^{-S[\phi]} \C{O}(\phi(0)),
\eeq
where the path integral is taken over the interior of the unit disc and the field $\phi(x,t)$ assumes the value $\phi_b(x,1)$ at the boundary. We therefore expect the mini-superspace wave functions to be approximations to disc functions in which the Liouville field has some specified value on the boundary. In order to compare the mini-superspace wave-functions to later quantities it is useful to note that in the mini-superspace approximation the value of $\phi_b$ is completely fixed if we specify the boundary length of the disc. Indeed, the two are related by $l = 2 \pi e^{b \phi_b}$ and so we may rewrite \eqref{miniH} to give,
\beq
\left(-\frac{1}{2}\left(l \frac{\partial}{\partial l} \right)^2+\frac{\mu}{2 \pi b^2} l^2\right) \psi[\phi] = \frac{E^2}{b^2} \psi[\phi],
\eeq
where we have now dropped the subscript $b$ on the field $\phi$. The solutions to this equation are the probability distribution for the boundary length of the disc and they have the form,
\beq
\label{miniLwave}
\psi(l) = A K_{\sqrt{2}E i /b}\left[\sqrt{\frac{\mu}{\pi b^2}}l\right],
\eeq
where $A$ is a constant. Furthermore, given the state operator correspondence we expect the parameter $E$ to label which local state has been inserted in the disc. It is this quantity we will compare with exact results from Liouville and matrix model calculations.

We now turn to the problem of computing correlation functions in Liouville theory. As was mentioned in section \ref{CFTandMMs}, the only input necessary to fix the operator algebra of a conformal field theory is the three point function coefficients. We therefore need only compute the three point correlation function. It is unnecessary for our purposes to review the computation of the three point function in detail, however it will be useful to consider the general approach. Historically, the Liouville three point function was first computed in \cite{Dorn:1994xn,Zamolodchikov:1995aa} using the method of \cite{Goulian:1990qr}, however a more efficient method was introduced later in \cite{Teschner:1995yf}. The method of \cite{Goulian:1990qr} is based on the observation that for particular combinations of operators and worldsheet topology the correlation function may be expressed in terms of free field correlators. Explicitly, consider the following $n$-point correlation function,
\beq
\avg{\prod^n_{i=0} e^{2\alpha_i \phi}} = \int [d\phi] \prod^n_{i=0} e^{2\alpha_i \phi} e^{-S_L[\phi]}.
\eeq
We proceed by splitting the integral into an integral over the zero mode $\phi_0$ and the non-zero modes $\bar{\phi}$,
\beq
\avg{\prod^n_{i=0} e^{2\alpha_i \phi}} = \int [d\bar{\phi}] \prod^n_{i=0} e^{2\alpha_i \bar{\phi}} e^{-S^{(0)}_L[\bar{\phi}]} \int^\infty_{-\infty} d\phi_0 e^{-\left(Q\chi - 2\sum^n_{i=0}\alpha_i\right)\phi_0} e^{-e^{2b\phi_0} \mu \int d^2 x \sqrt{g} e^{2b\bar{\phi}}},
\eeq
where $S^{(0)}_L[\bar{\phi}]$ is the free field action obtained by setting $\mu = 0$ in $S_L[\phi]$ and $\chi$ is the Euler characteristic of the worldsheet. We will now focus our attention on the zero mode integration,
\beq
\label{zeromodeInt}
\frac{1}{2b} \int^\infty_{0} dt \left[t^{-s-1} e^{-t \mu \int d^x \sqrt{g} e^{\bar{\phi}}}\right],
\eeq
where $s \equiv (Q\chi - 2\sum^n_{i=0}\alpha_i)/(2b)$ is known as the KPZ exponent and we have introduced $ t= \exp(2b\phi_0)$. It is clear that this integral is divergent for $s>0$ and therefore must be regularised in those cases.  
The method employed in \cite{Goulian:1990qr} is to compute the zero mode integral for $\mathrm{Re}[s] < 0$  and then analytically continue, giving the result,
\beq
\label{LigoulianZ}
\avg{\prod^n_{i=0} e^{2\alpha_i \phi}}=\frac{\Gamma(-s)}{2b} \mu^s \avg{\left(\int d^2 x \sqrt{g} e^{2b\bar{\phi}}\right)^s  \prod^n_{i=0} e^{2\alpha_i \bar{\phi}}}_0,
\eeq
where the average is taken using the free action $S^{(0)}_L[\bar{\phi}]$. We now make the observation that, if we arrange the operators and topology of the worldsheet to be such that $s \in \BB{Z}^+$, then the remaining correlation function in \eqref{LigoulianZ} may be computed as it is a free field theory correlator. One problem we must address is the prefactor of $\Gamma[-s]$ which is infinite at precisely the points for which the correlator becomes a free field correlator. In \cite{Zamolodchikov:1995aa} this property was interpreted as implying that the three point coefficients have a pole precisely when the relation $s \in \BB{Z}^+$ is satisfied. The result \eqref{LigoulianZ} therefore gives the residue at these poles. Knowing the residue of the three point function at these poles allowed \cite{Zamolodchikov:1995aa} and \cite{Dorn:1994xn} to guess the correct analytic continuation of these results to all $s$.

The above regularisation of the integral in the case $s>0$ was achieved via analytic continuation. A more brutal regularisation would be to simply cut off the $\phi_0$ integral. Unusually, the divergence of the zero mode has both a UV and IR interpretation. With respect to the background metric $\tilde{g}$ the divergence of the zero mode integral looks like an infrared divergence which should be solved by placing the system in a finite sized box. Alternatively, with respect to the original metric $g =e^{2 p \phi} \tilde{g}$, the divergence of the integral due to contributions from $\phi \rightarrow -\infty$ corresponds to a divergence due to the contributions from worldsheets in which all distances shrink to zero, which is a UV problem. We therefore may think of placing a cutoff at $\phi = -\frac{1}{2b} \log \Lambda$ as being either a short or long distance cutoff. If we implement the cutoff in \eqref{zeromodeInt}, we obtain,
\beq
\frac{1}{2b} \int^\infty_{\Lambda^{-1}} dt \left[t^{-s-1} e^{-t \mu \int d^x \sqrt{g} e^{\bar{\phi}}}\right] =  \frac{\Gamma(-s)}{2b} \mu^s \left(\int d^2 x \sqrt{g} e^{2b\bar{\phi}}\right)^s + P(\mu,\Lambda),
\eeq
where $P(\mu,\Lambda)$ is a cutoff dependent polynomial in $\mu$. The property that the cutoff dependent part of the above integral is polynomial in $\mu$ prompted \cite{Seiberg:1990eb} to classify contributions to correlation functions that are analytic in $\mu$ to be non-universal. Another way to think of this is that since these contribution arise from zero area worldsheets, then if they were to arise in the matrix model formulation of the theory such contributions would feel the effects of the lattice and hence be non-universal. Due to the exponential interaction term in the Liouville action the theory does not probe far into the region $\phi > -\frac{1}{2b}\log \mu$. Hence, in the presence of the cutoff, the volume of the Liouville direction is $V_\phi  =  -\frac{1}{2b}\log \left(\frac{\Lambda}{\mu}\right)$.

Recall that in the presence of boundaries we could compute amplitudes with either the boundary cosmological constant or the boundary length fixed, with these amplitudes related to one another via the Laplace transform \eqref{Wfixedl}. One use of the fixed boundary length amplitudes is that one expects worldsheets of zero extent to not contribute since a surface can not disappear if the boundary length is non-zero \footnote{We make a distinction between vanishing worldsheets and one with zero area. Branched polymers, as generated by the Gaussian matrix model, represent an example of worldsheets with zero area but non-zero boundary length and therefore have universal meaning. See \cite{Moore:1991ir} for further discussion.}.

Contributions to the amplitude from degenerate worldsheets could only reappear precisely when the boundary lengths go to zero. We therefore have two ways to determine which terms are non-universal; they are analytic in $\mu$ and also they should only have support when all boundary lengths are zero.

\subsection{Minimal String Theory}
Minimal string theories are obtained by interpreting Liouville theory coupled to a $(p,q)$ minimal model as the worldsheet theory for a string. The requirement that the total central charge of the system is zero, together with the expression for the central charge for the minimal models, Liouville theory and ghosts leads to $b^2=\frac{p}{q}$. The physical operators of the $(p,q)$ minimal string were first given by \cite{Lian:1992aj} as $\C{U}_n \equiv \C{O}_n V_{\alpha(n)}$ where,
\beq
\label{LZstates}
\alpha(n) =b \frac{p+q-n}{2p} \quad n \neq 0 \mod p \quad n \neq 0 \mod q \quad n \in \BB{Z},
\eeq
and $\C{O}_n$ is an operator containing contributions from the ghost and matter sector in addition to derivatives of the Liouville field. An important physical operator is obtained by gravitationally dressing the primary fields $\Oh_{r,s}$ of the minimal model; such operators are referred to as tachyons and take the form,
\bea
\C{T}_{r,s}&=&c\bar{c}\Oh_{r,s}V_{\alpha_{r,s}}, \\ 
\alpha_{r,s}&=& \frac{p+q-|qr-ps|}{2\sqrt{pq}},
\eea
where $c$ and $\bar{c}$ are ghost fields. The expression for $\alpha_{r,s}$ arises from the requirement that the conformal dimension of the total tachyon operator is one as this allows for the integral of the operator over the worldsheet to be conformally invariant.

\subsection{FZZT Branes and the Seiberg-Shih Relation}
Let us return briefly to consider the standard critical string. It is well known that string theory contains other degrees of freedom which do not appear in the perturbative string expansion. The most well known of these are D-branes. Such objects correspond to extended membranes thereby generalising the idea of a string. The D-branes are particularly interesting as they do admit a description in terms of perturbative string theory; although they do not appear in the closed string expansion they can be described by open strings in which the boundaries carry Dirichlet boundary conditions. If $26-n$ of the string embedding coordinates carry Dirichlet boundary conditions then this means the end of the string is confined to move in a $n$ dimensional hyperplane, thereby giving rise to the membrane like structure of the D-brane, which in this case would be called an $n$-brane. Since D-branes are described by worldsheets with boundaries the worldsheet theory is known as a boundary conformal field theory.
 
Given this motivation it is natural to enquire whether the $(p,q)$ minimal strings contain any branes. The answer to such a question would involve classifying the possible conformally invariant boundary conditions in the worldsheet CFT. A review of the result of such a program for the $(p,q)$ minimal model not coupled to gravity was given in Section \ref{MMBCFT}. We now consider a similar question for Liouville theory and hence the $(p,q)$ minimal string. The Liouville theory on a manifold with boundary is defined by the Lagrangian,
\beq
S_L=\frac{1}{4\pi}\int_M d^2x\sqrt{g}\left(g^{ab}\partial_a\phi\partial_b\phi+Q\phi R[g]+4\pi\mu e^{2b\phi}\right) + \int_{\partial M} \mu_B e^{b\phi} dx,
\eeq
where $\mu_B$ is known as the boundary cosmological constant. The full conformal bootstrap of Liouville theory on the disc was performed in \cite{Fateev:2000ik}\cite{Ponsot:2001ng}, resulting in the identification of a consistent boundary condition for the Liouville theory; the FZZT brane. As part of the conformal bootstrap, the bulk one-point function on a disc, i.e. the bulk one-point function with FZZT boundary, was calculated:
\beq
\label{FZZTstate} \Psi_{\sigma_1}(P) \equiv \langle v_P |\sigma\rangle_{FZZT}=(\pi \mu \gamma(b^2))^ {-iP/b}\Gamma(1+2ibP)\Gamma(1+2iP/b)\frac{\cos(2\pi \sigma P)}{iP}
\eeq 
where $\sigma$ is defined implicitly by, $\mu \cosh^2\pi b \sigma = {\mu_B}^2\sin(\pi b^2)$ and $\ket{v_P}$ is the primary state associated to $V_{Q/2+iP}$.

In the minimal string the full FZZT brane is $\ket{\sigma;k,l}=|\sigma\rangle_{FZZT} \otimes |k,l\rangle$, where $|k,l\rangle$ is the $(k,l)$ Cardy boundary state of the minimal model. We will find it convenient to follow \cite{Seiberg:2003nm} in rescaling the bulk and boundary cosmological constant so that $\pi \mu \gamma(b^2) \rightarrow \mu$ and $\mu_B \sqrt{\pi\gamma(b^2)\sin(\pi b^2)} \rightarrow \mu_B$. This modifies the relation between $\mu$ and $\sigma$ to $\mu_B = \sqrt{\mu} \cosh(\pi b \sigma)$. 

It was noted in \cite{Dorn:1994xn} \cite{Zamolodchikov:1995aa} that the results of the conformal bootstrap on the sphere are invariant under the duality transformation $b\rightarrow \frac{1}{b}$, $\mu \rightarrow \tilde{\mu}=\mu^{\frac{1}{b^2}}$. This self duality is present in the bootstrap on the disc if the boundary cosmological constant transforms as, $\mu_B \rightarrow \tilde{\mu}_B$, where $\frac{\tilde{\mu}_B}{\sqrt{\tilde{\mu}}}=\cosh \frac{\pi \sigma}{b}$. The transformed boundary condition is referred to as the dual FZZT brane and provides a physically equivalent description of the FZZT brane. Given that all computed amplitudes exhibit this duality it is in our opinion important for this symmetry to be maintained in any modifications of amplitudes, such as those we consider later.

The amplitude in the minimal string for the insertion of a tachyon in the bulk of a disc is obtained using \eqref{FZZTstate} \cite{Seiberg:2003nm},
\bea\label{tachyonDisk}
\braket{\C{T}_{r,s}}{\sigma;k,l}=&&A_D\left[\sin\left(\frac{\pi t}{p}\right)\sin\left(\frac{\pi t}{q}\right)\right]^{\frac{1}{2}} \Gamma(2b \alpha -b^2) \Gamma(2b^{-1} \alpha -b^{-2}-1) \times \nn\\
&&U_{k-1}\left(\cos\frac{\pi t}{p}\right)U_{l-1}\left(\cos\frac{\pi t}{q}\right)\mu^\frac{t}{2p}\cosh\left(\frac{\pi t \sigma}{\sqrt{pq}}\right), 
\eea
where $t=qr-ps$, $2\alpha=Q-\frac{t}{\sqrt{pq}}$, $A_D$ is a constant independent of $k$, $r$, $s$ and $l$ and $U_n(x)$ is the $n$th Chebyshev polynomial of the second kind. This expression for the disc amplitudes motivated Seiberg and Shih to conjecture in \cite{Seiberg:2003nm} that the construction just outlined over counts the number of FZZT branes and that, up to BRST null states, the following identification holds,
\beq \label{SSequiv}
\ket{\sigma;k,l}=\sum^{k-1}_{n=-(k-1),2} \sum^{l-1}_{m=-(l-1),2}  \ket{\sigma+i\frac{mp+nq}{\sqrt{pq}};1,1}.
\eeq 
This relation, which we will refer to as the SS (for Seiberg-Shih) relation, was originally obtained by inspection of \eqref{tachyonDisk} but  was later derived for discs using the ground ring in \cite{Basu:2005sda}. Essentially it states that there is only a single FZZT brane $\ket{\sigma;1,1}$ and that all the others are related to it by complex shifts in the boundary cosmological constant which has somehow absorbed all the information about the matter boundary condition.

It is worth pausing at this point to compare the above exact result for the disc function amplitude in the presence of FZZT boundary conditions, with the result obtained via the mini-superspace approximation. To compare the above result \eqref{tachyonDisk} to \eqref{miniLwave}, we must take the inverse Laplace transform. This may be achieved by use of the identity,
\beq
\label{InverseLaplaceK}
\frac{\pi b \cos(2 \pi P \sigma)}{2 P \sinh(\frac{2\pi P}{b})} = \int^\infty_0 \frac{dl}{l} e^{-M l \cosh(\pi b \sigma)} K_{\frac{2iP}{b}}(M l).
\eeq
Setting $M = \sqrt{\mu}$, we find the length distribution for the $1$ point tachyon disc amplitude to be,
\beq
\C{L}^{-1}\braket{\C{T}_{r,s}}{\sigma;k,l} \propto K_{t/p}(\sqrt{\mu} l),
\eeq
hence identifying $E/b = n/p$, where $n$ is the integer defining the operator in \eqref{LZstates}. We see that the exact calculation of the Liouville disc function reproduces the mini-superspace result up to rescalings of $l$ and the normalisation constant. This has in fact been observed to be a generic feature of the mini-superspace approximation in this setting; it will give the correct functional form for any amplitude. This implies that for general disc amplitudes we have,
\beq
\label{LZdisc}
\braket{\C{U}_{n}}{\sigma;k,l} \propto \mu^\frac{n}{2p} \cosh\left(\frac{\pi n \sigma}{\sqrt{pq}}\right).
\eeq
This is consistent with the application of \eqref{FZZTstate} with ${Q/2+iP} = \alpha(n)$.

\subsection{Cylinder Amplitudes and the Seiberg-Shih Relation}
The minimal string cylinder amplitude was first computed in \cite{Martinec:2003ka} for the case of $\ket{\sigma;a,1}$ boundary conditions. The cylinder amplitude with arbitrary boundary conditions was given in \cite{Gesser:2010fi} and we have reproduced the derivation in Appendix \ref{Cylindercalc}. A compact expression for the general cylinder amplitude is,
\beq
\label{Zintegralc}
\C{Z}(k,l;\sigma_1|r,s;\sigma_2)= \frac{2\pi^2 }{\sqrt{2}pq} \int^{\infty}_0 dP \frac{P}{P^2+\epsilon^2} \frac{\cos(2\pi\sigma_1 P)\cos(2\pi\sigma_2 P)\sinh(\frac{2\pi P}{\sqrt{pq}})}{\sinh(2\pi b P)\sinh(\frac{2\pi P}{b})} F_{k,l,r,s}(\frac{2\pi i P}{\sqrt{pq}}),\\
\eeq
where $\epsilon$ is an IR cutoff and,
\bea 
F_{k,l,r,s}(z)=\sum^{\lambda_p(k,r)}_{\eta=|k-r|+1,2} \sum^{\lambda_q(l,s)}_{\rho=|l-s|+1,2} \sum^{p-1}_{a=1} \sum^{q-1}_{b=-(q-1)} &&\Bigg[\sin(\frac{\pi t}{p})\sin(\frac{\pi t}{q}) \times \\ 
&&U_{\eta-1}(\cos\frac{\pi t}{p})U_{\rho-1}(\cos\frac{\pi t}{q}) \frac{1}{\cos{\frac{\pi t}{pq}}-\cos{z}}\Bigg]\nn,
\eea
where $\lambda_p(k,r)=\mathrm{min}(k+r-1,2p-1-k-r)$ and $t=qa+pb$. It is the purpose of this section to investigate whether the SS relation is satisfied when applied to cylinder amplitudes. For this purpose it is useful to define the deviation from the SS relations for a particular amplitude $\C{A}(k,l;\sigma|X)$ as
\beq 
\label{SS1}
\Delta \C{A}(k,l;\sigma|X)=\C{A}(k,l;\sigma|X)-\sum^{k-1}_{n=-(k-1),2} \sum^{l-1}_{m=-(l-1),2}  \C{A}(1,1;\sigma+i\frac{mp+nq}{\sqrt{pq}}|X)
\eeq
where the presence of other boundaries and operator insertions is denoted by $X$. Furthermore, it will be useful to have at hand some explicit formulae for amplitudes resulting from the application of \eqref{Zintegralc}. For the $(3,4)$ minimal string the full set of cylinder amplitudes are \footnote{These expressions are correct up to an unimportant common normalisation constant and additions of numerical constants.},
\bea
\label{34amps}
&&\C{Z}(1,1;\sigma_1|1,1;\sigma_2)=\C{Z}(2,1;\sigma_1|2,1;\sigma_2)=\ln(\frac{z_1-z_2}{x_1-x_2})+\frac{1}{2\sqrt{3}\epsilon}, \\
&&\C{Z}(2,1;\sigma_1|1,1;\sigma_2)=-\ln(z_1+z_2)+\frac{1}{4\sqrt{3}\epsilon}, \nn\\
&&\C{Z}(1,1;\sigma_1|1,2;\sigma_2)=\C{Z}(2,1;\sigma_1|1,2;\sigma_2)=-\ln(-1 + 2 z_1^2 + 
  2 \sqrt{2} z_1 z_2 + 2 z_2^2)+\frac{1}{2\sqrt{3}\epsilon},\nn\\
&&\C{Z}(1,2;\sigma_1|1,2;\sigma_2)=\ln(\frac{z_1-z_2}{(z_1+z_2)(x_1-x_2)})+\frac{3}{4\sqrt{3}\epsilon}.\nn
\eea
where we have introduced the following notation,
\beq
z \equiv \cosh\left(\frac{\pi \sigma}{\sqrt{pq}}\right), \qquad x \equiv \frac{\mu_B}{\sqrt{\mu}} = \cosh\left(\frac{\pi p \sigma}{\sqrt{pq}}\right), \qquad \tilde{x} \equiv \frac{\tilde{\mu}_B}{\sqrt{\tilde{\mu}}} = \cosh\left(\frac{\pi q \sigma}{\sqrt{pq}}\right).
\eeq
The first check of the SS relation was performed in \cite{Kutasov:2004fg} in which it was applied to cylinder amplitudes in which all boundaries carried boundary states of the form $\ket{\sigma;a,1}$. It was found that the SS relations are not satisfied identically in this case, however it was argued that the deviation was non-universal and therefore should be discarded. The evidence for non-universality was two-fold, corresponding to the two properties of non-universal terms introduced in Section \ref{superspaceAndCorrelationN}. Firstly it was noted that the deviation possessed the following properties;

\begin{property}{1}
The deviation could be made independent of the bulk cosmological constant, and hence analytic in $\mu$, by setting the IR cutoff $\epsilon$ in \eqref{Zintegralc} to the volume of the Liouville direction.  \label{prop1}
\end{property}
As an example consider the second of the $(3,4)$ minimal string cylinder amplitudes above,
\beq
\label{exampleDev}
\Delta\C{Z}(2,1;\sigma_1|1,1;\sigma_2) = -\frac{\sqrt{3}}{4\epsilon}-\log(x_1+x_2).
\eeq
Choosing $\epsilon$ equal to the volume of the Liouville direction, $\frac{1}{\epsilon} = \frac{1}{b}\log\frac{\Lambda}{\mu}$ where $\Lambda$ is a constant and recalling that $x = \mu_B/\sqrt{\mu}$ we find,
\bea
\Delta\C{Z}(2,1;\sigma_1|1,1;\sigma_2) = -\frac{1}{2}\log\Lambda-\log({\mu_B}_1+{\mu_B}_2).
\eea
Since this is independent of $\mu$ it was interpreted as non-universal and therefore set to zero. For a demonstration that all cylinder amplitudes with boundary conditions given by the states $\ket{\sigma;a,1}$ have this property the reader is referred to the appendix of \cite{Kutasov:2004fg}
\\

\begin{property}{2}
The inverse Laplace transform, with respect to all boundary cosmological constants, of the deviation minus any regularisation dependent parts is zero almost everywhere i.e. it is supported only at points. \label{prop2}
\end{property}

This property was uncovered by considering \eqref{Zintegralc}. Using \eqref{InverseLaplaceK} in \eqref{Zintegralc} we find,
\beq
\C{L}^{-1}[\C{Z}](k,l;L_1|r,s;L_2) \propto \int^{\infty}_0 dP \psi_P(L_1)\psi_P(L_2) \frac{\sinh(\frac{2 \pi P}{\sqrt{pq}}) }{\sinh(2\pi P b)} F_{k,l,r,s}(\frac{2\pi i P}{\sqrt{pq}}),
\eeq
where $\psi_P(L) = \sqrt{P \sinh(2 \pi P/b)} K_{2iP/b}(L)$. By noting that the integrand is symmetric, the integral may be computed by closing the integral along the entire real axis by a semicircular contour. Since the Seiberg-Shih deviation is merely a linear combination of amplitudes its double inverse Laplace transform with respect to each boundary cosmological constant also exists and is given by
\beq
\label{invLZint}
\C{L}^{-1}[\Delta\C{Z}](k,l;L_1|r,s;L_2) \propto \int^{\infty}_0 dP \psi_P(L_1)\psi_P(L_2) \frac{\sinh(\frac{2 \pi P}{\sqrt{pq}}) }{\sinh(2\pi P b)} \Delta F_{k,l,r,s}(\frac{2\pi i P}{\sqrt{pq}})
\eeq
where
\beq
\Delta F_{k,l,r,s}(z) \equiv F_{k,l,r,s}(z) - U_{k-1}(\cos qz) U_{l-1}(\cos pz) F_{1,1,r,s}(z)
\eeq
and we have assumed that the Seiberg-Shih transformation has been applied to the $(k,l)$ boundary. An important property of $\Delta F_{k,l,r,s}(z)$ is that it is an entire function and therefore the only contributions to the integral comes from the poles at $P = \frac{in}{2b}$ where $n \in \BB{Z}$. In the case that $(k,l) = (a,1)$ and $(l,s) = (a',1)$ the residues of these poles are zero and the above expression is zero.

Although contrary to the original motivation of P2, we have allowed for the possibility of the amplitude being non-zero at points away from zero. We merely require the weaker condition of having point-like support.

We now want to review the results of \cite{Gesser:2010fi} in which all boundary conditions on the cylinder amplitudes were considered and the consequences of this on the SS relations. Before proceeding it is important to note that the sum in the Seiberg-Shih relation does not respect the symmetry of the Kac table and so for a given amplitude there are two possible Seiberg-Shih relations that might possess P1 or P2. However, there are cases where both possibilities lead to deviations that do not possess either property P1 or P2. 

For an example of a deviation that does not possess property P1 consider,
\beq
\Delta\C{Z}(1,2;\sigma_1|1,1;\sigma_2)= \ln\left[\frac{-1 + 2{x_1}^2 + 2\sqrt{2} x_1 x_2 + 2{x_2}^2}{2(\tilde{x}_1 + \tilde{x}_2)}\right],
\eeq
or, since $\C{Z}(2,2;\sigma_1|1,1;\sigma_2)=\C{Z}(1,2;\sigma_1|1,1;\sigma_2)$,
\beq
\Delta\C{Z}(2,2;\sigma_1|1,1;\sigma_2)= \ln \left[ (-1 + 2{x_1}^2 - 2\sqrt{2} x_1 x_2 + 2{x_2}^2)\right],
\eeq
where we have suppressed the $\epsilon$ dependent term. It is clear that since the argument of both expressions is not a homogeneous polynomial in $x$ the dependence on $\mu$ cannot be factored out and placed in a separate term. Hence the IR cutoff cannot be chosen to cancel all dependency on the bulk cosmological constant.

An example of an amplitude not possessing P2 need only have the property that the residues are non-zero \eqref{invLZint}, leading to a function with global support. As an example consider,
\bea
\label{P2examplea}
\Delta\C{Z}(1,2;\sigma_1|1,2;\sigma_2)= \ln\left[\frac{\tilde{x}_1 - \tilde{x}_2}{2(x_1 - x_2)}\right]
\eea
or, since $\C{Z}(2,2;\sigma_1|2,2;\sigma_2)=\C{Z}(1,2;\sigma_1|1,2;\sigma_2)$,
\bea
\label{P2exampleb}
\Delta\C{Z}(2,2;\sigma_1|2,2;\sigma_2)= \ln\left[\frac{(x_1 - x_2)(x_1^2+x_2^2-1)}{2(\tilde{x}_1 - \tilde{x}_2)}\right].
\eea
By applying \eqref{invLZint} it is now easy to show P1 and P2 are not satisfied for \eqref{P2examplea} and \eqref{P2exampleb}; in both cases there exist poles when $P = \frac{in}{2b}$ and therefore when the integral contour is closed around them it will result in a sum of terms containing Bessel functions of the form $K_\frac{4n}{3}(L)$ which have global support in $L$.

Having shown that the cylinder amplitudes possess deviations from the SS relations which do not fit the non-universal classification advanced in \cite{Kutasov:2004fg} we now wish to motivate a possible extension of the criteria by which non-universality is judged in the hope of rescuing the SS relations. The observation, is that, the terms that disappear in \eqref{Zintegralc} under an inverse Laplace transform are dual under the Liouville duality to the terms that produce the troublesome poles in the above test of P2 for arbitrary amplitudes. We therefore propose that non-universal deviations should possess,

\begin{property}{3}
The deviation may be written in terms of contributions for which either the inverse Laplace transform with respect to all boundary cosmological constants $\mu_B$ or all their duals, $\tilde{\mu}_B$, has point-like support. \label{prop3}
\end{property}

We devote the remainder of this section to a more careful motivation of P3. Consider the integral representation of the amplitude \eqref{Zintegralc}; it may be computed by extending the region of integration to the entire real line and then splitting the integrand up into terms for which the contour may be closed in either the upper or lower half plane. Assuming that $\sigma_1 > \sigma_2$, then upon substituting in \eqref{Fresidues} and \eqref{Gresidues} this results in the expression,
\bea
\label{vExplicitZ}
&&\C{Z}(k,l;\sigma_1|r,s;\sigma_2) \propto F_{k,l,r,s}(0)(\frac{1}{2\epsilon}-\pi \sigma_1)\frac{1}{\sqrt{pq}} + \\
&&\sum_{\substack{n=1\\ n\neq 0modp\\ n\neq 0modq }}\frac{4pq}{n} U_{k-1}(\cos\frac{\pi n}{p})U_{r-1}(\cos\frac{\pi n}{p}) U_{l-1}(\cos\frac{\pi n}{q})U_{s-1}(\cos\frac{\pi n}{q}) e^{-\frac{\pi n \sigma_1}{\sqrt{pq}}} \cosh \frac{\pi n \sigma_2}{\sqrt{pq}} + \nn \\
&&\sum^{\infty}_{\substack{n=1\\ n\neq 0\mod q}}  \frac{2}{n} (-1)^n F_{k,l,r,s}(\frac{\pi n}{q})\frac{\sin(\pi n /q)}{\sin(\pi p n /q)} e^{-\frac{\pi p n \sigma_1}{\sqrt{pq}}} \cosh(\pi b n \sigma_2) +\nn\\
&&\sum^{\infty}_{\substack{n=1\\ n\neq 0\mod p}} \frac{2}{n} (-1)^n  F_{k,l,r,s}(\frac{\pi n}{p})\frac{\sin(\pi n /p)}{\sin(\pi q n /p)}e^{-\frac{\pi q n \sigma_1}{\sqrt{pq}}} \cosh (\pi n \sigma_2/b) + \nn \\
&&\sum^{\infty}_{n=1}  \frac{2}{npq} (-1)^{n(p+q+1)} F_{k,l,r,s}(\pi n) e^{-\frac{\pi p q n \sigma_1}{\sqrt{pq}}} \cosh \frac{\pi p q n \sigma_2}{\sqrt{pq}}. \nn
\eea
The above expression may be understood as representing a cylinder amplitude as a sum over disc amplitudes with various local or boundary operators inserted. If we compute the deviation using this expression the first sum cancels while the others remain with the replacement $F_{k,l,r,s} \rightarrow \Delta F_{k,l,r,s}$. The sum in which each term is proportional to $\cosh(\pi b n \sigma_2)$ arises from the poles in \eqref{Zintegralc} due to the factor of $\sinh(\frac{2\pi P}{b})$ in the denominator; it is precisely these terms which vanish when we take the double inverse Laplace transform with respect to the boundary cosmological constant. These terms have the interpretation of some form of descendant dual boundary length operator inserted on the boundary of a disc. On the other hand, the sum in which each term is proportional to $\cosh(\pi n \sigma_2/b)$ arises from the poles in \eqref{Zintegralc} due to the factor of $\sinh(2\pi b P)$ in the denominator; it is these terms which, under a double inverse Laplace transform with respect to the boundary cosmological constant, become the terms proportional to Bessel functions of the form $K_\frac{4n}{3}(L_2)$. Similarly, if we were instead to take the double inverse Laplace transform with respect to the dual boundary cosmological constant, the terms proportional to $\cosh(\pi n \sigma_2/b)$ would vanish and the terms proportional to $\cosh(\pi n b \sigma_2)$ would become terms proportional to $K_\frac{3n}{4}(L_2)$. We therefore see that the deviation computed from \eqref{vExplicitZ} can be understood as a sum of terms which vanish under a double inverse Laplace transform with respect to the boundary cosmological constant or its dual.

\chapter{FZZT Branes in the $(3,4)$ Minimal String at Higher Genus\label{ChapBrane}}


\section{Testing the Seiberg-Shih relations using a Matrix Model}
\label{MatrixSec}
In the last chapter we argued that the deviations for all cylinder amplitudes have property P3. Obviously it would be interesting to know whether this is true for deviations of any amplitude. Unfortunately amplitudes more complicated than the disc and cylinder have never been computed using the Liouville approach. An alternative is to use the matrix model formulation of minimal string theory in which the computation of amplitudes with arbitrary numbers of boundaries and handles is straightforward. 


The disadvantage of using the matrix model is that it gives the amplitude in a fully integrated and summed form and so it is hard to understand the structure of the amplitude in terms of continuum concepts such as states circulating in a loop. Furthermore, generally it is obscure how the graph labels map to the conformal field theory degrees of freedom and even more importantly, the matrix model appears to not contain all the boundary states present in the minimal string. Indeed, this was used as evidence in \cite{Seiberg:2003nm} in support of the Seiberg-Shih relation. However, there are special cases where the relation is manifest; we saw that the $(3,4)$ minimal string admits a formulation as a matrix model \eqref{2MMIsing} in which the matrices are directly related to the spin degrees of freedom. In this model it is straight forward to construct all conformal boundary states. This is the model we will study in this chapter. In recent work of \cite{Ishiki:2010wb,Bourgine:2010ja} a number of other boundary states of the $(p,q)$ minimal string have been constructed in the matrix model formulation and it would be interesting to extend our results to those cases.

In this chapter we will find it useful to change the definition of the coupling constants from those considered in \eqref{2MMIsing}. Instead we will use the matrix model defined by,
\bea
\label{2MMZ}
\Z=\int dM_1\,dM_2\exp\left(-\frac{N}{g}\,\Tr \left(- cM_1 M_2+\half(M_1^2+M_2^2)+\third(M_1^3+M_2^3)\right)\right)\label{ZXY}
\eea
The coupling between neighbouring spins is controlled by the parameter $c$ whereas $g$ controls the cost associated with adding more vertices to the graph. It is $g$ and $c$ we use to tune the matrix model to its critical point which for our matrix model we do by setting,
\bea
\label{gscale}
g = g_c(1-a^2 \eta \mu), \qquad c = c_c,
\eea
where $\eta$ is some constant which can be determined by comparing with the Liouville theory and the critical point is achieved by letting $a \rightarrow 0$\footnote{The constants $g_c$ and $c_c$ have the well known values $g_c=10c_c^3$ and $c_c = \frac{1}{27}({2\sqrt{7}-1})$.}. We will often refer to this limit as the scaling or continuum limit.

The conformally invariant boundary conditions of minimal CFTs are in one-to-one correspondence with the primary fields of the theory. For a CFT that describes the critical point of some discrete model the boundary conditions can sometimes be understood as universality classes of boundary conditions in the discrete model. As we approach the critical point of a discrete model all boundary conditions in a given class will flow to the same boundary condition in the CFT. In particular for the $(3,4)$ minimal string its three conformal boundary conditions are the continuum limit of the discrete configurations of boundary spins consisting of all spin $+$, all spin $-$, or free spins. It is these boundary conditions we want to implement in the matrix model.

The usual way of inserting a boundary in the matrix model is to compute the resolvent $W_{M_1}$ as it is a generating function for triangulations of discs. We will adopt the following notation for more general quantities,
\bea
\label{genresolve}
\lefteqn{W^{Q_1}{}_{F_1}{}^{Q_2}{}_{F_2\ldots;}{}^{Q_3}{}_{G_1}{}^{Q_4}{}_{G_2\ldots;\ldots}(x_1,x_2,\ldots;y_1,y_2,\ldots;\ldots)=}\nn\\
&&\avg{\frac{1}{N}\Tr\Big(\frac{Q_1(M_1,M_2)}{x_1-F_1(M_1,M_2)}\frac{Q_2(M_1,M_2)}{x_2-F_2(M_1,M_2)}\dots\right)\ldots \\ 
&&\frac{1}{N}\Tr\left((\frac{Q_3(M_1,M_2)}{y_1-G_1(M_1,M_2)}\frac{Q_4(M_1,M_2)}{y_2-G_2(M_1,M_2)}\dots\Big)\ldots}\nn
\eea
By tuning the $x_i, y_i, \ldots$ as we take the matrix model to its critical point we can extract continuum quantities corresponding to amplitudes with macroscopic boundaries. By way of example, for $W_{M_1}(x)$ we use \eqref{gscale} and set,
\bea
\label{xscale}
x = x_c(1-a^{d} \kappa \mu_B)
\eea
where $d$ is chosen to produce a non-trivial limit and $\kappa$ is chosen to agree with Liouville theory. As we let $a \rightarrow 0$, $W_{M_1}$ will have an expansion in powers of $a$ \footnote{When taking the scaling limit in the remainder of the chapter we will choose $\eta = 1$ and $\kappa = \sqrt{g_c/c_c}$. It should be kept in mind though that by changing $\eta$ and $\kappa$ we can renormalise $\mu_B$ and $\mu$.},
\bea
W_{M_1}(x_c(1-a^{d} \kappa \mu_B)) = \sum_i a^{d_i} W_i(\mu_B,\mu) + a^{d_W} \tW_{M_1}(\mu_B,\mu) + h.o.t \quad.
\eea
In this expression $W_i(\mu_B,\mu)$ is analytic in both arguments whereas $\tW_{M_1}$ is defined as being the first term non-analytic in $\mu_B$ and $\mu$. It is $\tW_{M_1}$ that corresponds to the continuum quantity and so in this case is the partition function of the continuum theory defined on a surface with a finite sized boundary with the boundary condition determined by which universality class the discrete quantities $\avg{\Tr M^n}$ belong to. This is slightly complicated by the fact that the resolvent actually corresponds to a boundary with a marked point and therefore must be integrated with respect to the boundary cosmological constant before comparing to Liouville theory. This motivates introducing the integrated quantity,
\bea
\omega({\mu_B}_1,...,{\mu_B}_n,\mu) \equiv \int \prod_i d{\mu_B}_i \tW({\mu_B}_1,...,{\mu_B}_n,\mu)
\eea
where we have allowed for more general amplitudes which have more than one boundary. The generalisation to amplitudes of the form \eqref{genresolve} is obvious.

For the matrix model \eqref{ZXY} it is easy to construct resolvents which generate boundary conditions that flow to the three different boundary conditions in the CFT:
\bea
\label{boundaryResolve}
W_{M_1}(x)&=&\avg{\frac{1}{N}\Tr\frac{1}{x-M_1}}, \nn \\
W_{M_2}(x)&=&\avg{\frac{1}{N}\Tr\frac{1}{x-M_2}},  \\
W_{M_1+M_2}(x)&=&\avg{\frac{1}{N}\Tr\frac{1}{x-(M_1+M_2)}}. \nn
\eea
The first two resolvents $W_{M_i}$ generate surfaces on which only $M_i$ vertices appear on the boundary. Since the $M_i$ vertices map directly to the spin degrees of freedom these flow to the fixed spin boundary conditions in the continuum limit. The final resolvent generates boundaries in which both types of vertices appear with equal weighting. This will therefore flow to the free spin boundary condition. We can parameterise these resolvents in the following way,
\bea
\label{genXresolvent}
W_X(x) = \avg{\frac{1}{N}\Tr\frac{1}{x-X}} \qquad X = \frac{(1 - \alpha)}{2}M_1-\frac{(1 + \alpha)}{2}M_2-(1 + c) \frac{\alpha}{2}.
\eea
This generates graphs with weighting corresponding to the Ising model on a random lattice but with a boundary magnetic field, to which the parameter $\alpha$ is related. Clearly, we can choose $\alpha=\pm 1$ to reproduce the resolvents $W_{M_i}$. By taking $\alpha$ to infinity we may also obtain the resolvent $W_{M_1+M_2}$. One method to compute the resolvent $W_X$ is to make the following change of variable in \eqref{ZXY},
\bea
\label{SXdef}
S = M_1+M_2+1+c, \qquad
X = \frac{(1 - \alpha)}{2}M_1-\frac{(1 + \alpha)}{2}M_2-(1 + c) \frac{\alpha}{2}.
\eea
The partition function in these new variables takes the form,
\bea
\label{ZSX}
\mathcal{Z}=\int [dS dX] \exp{\left[-\frac{N}{g}\Tr\left(X^2 S+\alpha XS^2+V(S)\right)\right]}
\eea
where,
\beq 
\label{VSX}
V(S)=\frac{1}{12}(1+3\alpha^2)S^3-\frac{c}{2}S^2+\quarter(3c-1)(1+c)S.
\eeq
Generic values of $\alpha$ (i.e. the Ising model on a random lattice in the presence of a boundary magnetic field) were investigated using matrix model techniques in \cite{Carroll:1995qd}, in which the disc amplitudes were calculated. The method employed in \cite{Carroll:1995qd} to solve the matrix model is combinatorial and so is not easily generalised to compute more complicated amplitudes and so we will not employ it here. However, it is clear that if we are only interested in amplitudes in which all the boundary conditions are of the free spin form we may set $\alpha = 0$ in \eqref{ZSX} to obtain the $O(1)$ model whose solution is well known. Similarly if we are only interested in amplitudes in which all the boundary conditions are of the fixed spin form we may simply use \eqref{2MMZ}. This is in fact the case for most of the results we require for the remainder of this section and so they can be obtained form the literature  \cite{Eynard:2002kg}\cite{Eynard:1992cn, Eynard:1995nv}. For more complicated amplitudes involving both fixed and free spin boundary conditions it is necessary to solve the matrix model \eqref{ZSX}; this may be done using loop equations and we give details of this in the next section.


\subsection{Loop equations for SX Matrix Model}
\label{MMloops}
The resolvents for the matrix model \eqref{ZSX} can be obtained by generating an appropriate set of loop equations. We begin by obtaining the resolvents $W_X$ and $W_S$. Consider the following change of variables in the matrix model,
\bea
X \rightarrow X+\epsilon\frac{1}{z-S} 
\eea
this gives the loop equation,
\bea
\label{LE1d}
W^X{}_S(z)=\frac{\alpha}{2} - \frac{\alpha z}{2} W_S(z)
\eea
Now consider,
\bea
\label{LE1a}
X &\rightarrow X+\epsilon \left( \frac{1}{x-X} \frac{1}{z-S}+\frac{1}{z-S} \frac{1}{x-X} \right)\\
\label{LE1b}
S &\rightarrow S+\epsilon \left( \frac{1}{x-X} \frac{1}{z-S}+\frac{1}{z-S} \frac{1}{x-X} \right)\\
\label{LE1c}
X &\rightarrow X+\epsilon \left( \frac{1}{z-S}\frac{1}{x-X}\frac{1}{-z-S}+h.c \right)
\eea
Transformation \eqref{LE1a} together with \eqref{LE1b} and \eqref{LE1d} gives,
\bea
\label{LE1}
-\frac{1}{g} \Omega^P_X(z,x)=\Omega_S(z,x) \Omega_{SX}(z,x)+\frac{1}{N^2}\left(W_{SX;S}(z,x;z)-\alpha W_{SX;X}(z,x;x) \right)
\eea 
where we have introduced the following functions,
\bea \nonumber
&&\Omega_S(z,x) = W_S(z)-\alpha W_X(x) - \frac{1}{g} (x^2-\alpha^2 z^2 + V'(z))\\ 
\label{omegadef}
&&\Omega_{SX} = W_{SX}(z,x)+\frac{1}{g}(x-\frac{\alpha z}{2})\\ \nonumber 
&&\Omega^P_X =W^P_X(z,x)+\frac{\alpha}{2}+\frac{1}{g}(x-\frac{\alpha z}{2})(x^2-\alpha^2 z^2 + V'(z))
\eea
where $P=Q-\alpha^2 S+\alpha x -\frac{3\alpha^2 z}{2}$ and $Q = (V'(z)-V'(S))/(z-S)$.

Finally, the third transformation \eqref{LE1b} gives,
\bea
\label{LE2}
\Omega_{SX}(z,x) \Omega_{SX}(-z,x)=&&\frac{1}{g} \left(W_S(z)+W_S(-z)+\alpha W_X(x)+\frac{1}{g}(x^2-\frac{\alpha^2 z^2}{4}) \right)\\ \nn
-&&\frac{1}{N^2}W_{SX;SX}(z,x;-z,x)
\eea
Eliminating $\Omega_{SX}(z,x)$ between this and (\ref{LE1}) gives,
\bea
\nonumber
&&\frac{1}{g}\Omega^P_{X}(z,x) \Omega^P_{X}(-z,x)=\Omega_S(z,x)\Omega_S(-z,x) \left(W_S(z)+W_S(-z)+\alpha W_X(x)+\frac{1}{g}(x^2-\frac{\alpha^2 z^2}{4}) \right)\\ 
&&-\frac{1}{N^2} \Big(g \Omega_S(z,x)\Omega_S(-z,x)W_{SX;SX}(z,x;-z,x) + \Omega^P_{X}(z,x) W_{SX;S}(-z,x;-z) \label{LEwithN} \\ \nonumber
&&-\alpha \Omega^P_{X}(z,x) W_{SX;X}(-z,x;x) +\Omega^P_{X}(-z,x)\Big(W_{SX;S}(z,x;z)-\alpha W_{SX;X}(z,x;x)\Big) \Big)
\eea
By substituting in the large $N$ expansion for all resolvents appearing in this expression we can obtain a recursive definition for the genus $h$ resolvents. The genus zero resolvents satisfy:
\bea
\label{mastereqn}
\frac{1}{g}{\Omega^P_{X}}^{(0)}(z,x) {\Omega^P_{X}}^{(0)}(-z,x)=&& \\ \nn
{\Omega_S(z,x)}^{(0)}{\Omega_S(-z,x)}^{(0)}&&\left({W_S}^{(0)}(z)+{W_S}^{(0)}(-z)+\alpha {W_X}^{(0)}(x)+\frac{1}{g}(x^2-\frac{\alpha^2 z^2}{4}) \right)
\eea

This is an important equation and all subsequent loop equations will be related to this through derivative-like operations such as the loop insertion operator. 
It is important to note that the LHS is polynomial in $z$ and so if we expand the RHS in $x$ about any point, the resulting Laurent expansion must have coefficients that are equal to a polynomial function of $z$. A convenient point to choose for this expansion is infinity as the definition of the resolvent then  coincides with the Laurent expansion, whose coefficients are the yet to be determined quantities, $\avg{ \Tr{X^n}}$. Thus we produce a number of equations containing ${W_S}^{(0)}(z)$, 
\bea
\label{eqngen}
\oint_{C_{\infty}} \frac{dx}{2\pi i x^{n+1}} {\Omega_S(z,x)}^{(0)}&&{\Omega_S(-z,x)}^{(0)} \times \\ \nn &&\left({W_S}^{(0)}(z)+{W_S}^{(0)}(-z)+\alpha {W_X}^{(0)}(x)+\frac{1}{g}(x^2-\frac{\alpha^2 z^2}{4})\right) = p_n(z)
\eea
where $C_{\infty}$ is a contour around infinity, $p_n(z)$ is a polynomial in $z$ with a finite number of unknown coefficients, $n$ is an integer and the integral merely picks out the $n$th coefficient of the expansion. From the large $x$ behaviour of $W_X(x) \sim \frac{1}{x}$ it is easy to show that for $n>6$ the LHS is zero and for $n\geq 4$ the LHS yields an expression with no dependence on ${W_S}^{(0)}$. The non-trivial cases occur for $n=0$ and $n=2$ which generate enough equations to solve for ${W_S}^{(0)}(z)$. The large $z$ expansion of equations obtained for $n<0$ give relations among the quantities $\avg{\Tr{S^n}}$. The explicit equation obtained for ${W_S}^{(0)}(z)$ is,
\bea
\label{WSeqn}
{W_S}^{(0)}(z)^3&+& \frac{1}{g}\Big(\frac{3 z^2 \alpha^2}{4} + V'(-z) - 2 V'(z)\Big) {W_S}^{(0)}(z)^2\\
&+& \frac{1}{4 g^2}\Big(V'(-z) \Big(3 z^2 \alpha^2 - 4V'(z)\Big)-3 z^2 \alpha^2 V'(z) + 4 V'(z)^2 + P_2(z)\Big){W_S}^{(0)}(z) \nn\\
&+&\frac{1}{4 g^3}(P_0(z) + (z^2 \alpha^2 - V'(z)) P_2(z)) = 0\nn
\eea
where,
\bea
\label{Pdef}
P_i(z)=-p_i(z)+(\textrm{Coefficient of }({W_S}^{(0)}(z))^0\textrm{ in \ref{eqngen} for } n=i)
\eea
The unknown constants in $p_n(z)$ are not all independent and they may be found in terms of the constants $\langle \Tr{S^n} \rangle$ by expanding the RHS of \eqref{eqngen} about $z=\infty$ and equating $p_n$ to the polynomial part of the Laurent expansion. The requirement that the singular part of the expansion vanishes again gives relations between the quantities $\avg{\Tr{S^n}}$. 

To find ${W_X}^{(0)}$, we again consider \eqref{mastereqn}, however we now know ${W_S}^{(0)}$ in terms of a finite number of unknown constants. On the LHS we have the function ${\Omega^P_{X}}^{(0)}(z,x)$ where $P$ is a polynomial in $z$ and $S$ as defined in \eqref{omegadef}, whose highest power for both $z$ and $S$ is $d-1$ and so it contains $d$ unknown functions of the form $W^{S^n}_X(x)$ for $0 \leq n <d$.

We can generate a system of equations for these resolvents by again expanding both sides of the loop equation about infinity but this time in terms of $z$. This results in,
\bea
\label{WXgen}
&&\oint_{C_{\infty}} \frac{dz}{z^{n+1}}\bigg[\frac{1}{g}{\Omega^P_{X}}^{(0)}(z,x) {\Omega^P_{X}}^{(0)}(-z,x)-\\ \nonumber
&&{\Omega_S(z,x)}^{(0)}{\Omega_S(-z,x)}^{(0)} \left({W_S}^{(0)}(z)+{W_S}^{(0)}(-z)+\alpha {W_X}^{(0)}(x)+\frac{1}{g}(x^2-\frac{\alpha^2 z^2}{4}) \right)\bigg]=0
\eea

Again we may vary $n$ to generate different equations, however we now generate non-trivial equations for all $0 \leq n <d$, which is enough to solve for all unknown resolvent functions appearing in ${\Omega^P_{X}}^{(0)}(z,x)$. However, unlike ${{W_S}^{(0)}}$ for which we obtained an equation for general $V(x)$, we were unable to do this for the equations for $W^{S^n}_X(x)$. We also note that for the equation giving ${W_X}^{(0)}$, the highest power of ${W_X}^{(0)}$ appearing depends on the order of $V(x)$.

The approach thus outlined gives an algorithmic solution to this matrix model with an arbitrary potential and results in algebraic equations for all resolvents.

We now apply the above procedure to the case when the potential is \eqref{VSX}. There are two unknown constants in $P_{0}(z)$ and $P_{2}(z)$ corresponding to $\langle \Tr{S} \rangle$ and $\langle \Tr{S^2} \rangle$. Their values are determined by the requirement that \eqref{WSeqn} is a genus zero curve - which, in this case, is equivalent to the one-cut assumption. 

For ${W_{X}}^{(0)}$ we can generate two equations from \eqref{WXgen}, for $n=0$ and $n=1$, containing ${W_{X}}^{(0)}$ and ${W^S_{X}}^{(0)}$. The equation for ${W_{X}}^{(0)}$ resulting from the elimination of ${W^S_{X}}^{(0)}$ has a number of interesting properties. Generically it is quartic in ${W_{X}}^{(0)}$, however it reduces to a cubic for $\alpha = \pm 1$ and it reduces to a quadratic of $({W_{X}}^{(0)}+ \textrm{polynomial in } x)^2$ for $\alpha=0$. Because of this the scaling limit for generic $\alpha$ does not hold for these particular values of $\alpha$ and must be taken separately.

As was discussed before, the resolvents corresponding to the spin $+$ and $-$ Cardy states are ${W_{M_1}}^{(0)}$ and ${W_{M_2}}^{(0)}$ respectively, which may be computed directly from \eqref{ZXY} and are clearly identical. However a faster method, since we already have ${W_{X}}^{(0)}$, is to set $\alpha=-1$ implying $X=M_1+\half(1+c)$. 

It is easy to generate loop equations for a general 2-loop amplitudes of the form $W_{S;H}$ and $W_{X;H}$ where $H$ is a string of $S$ and $X$ matrices of length $\mathcal{N}_H$. We can write $H$ as, $H=\prod^{\mathcal{N}_H}_{n=1}\chi_n$, where $\chi_n$ is a matrix defined by,
\begin{equation}
\chi_n = \left\{
\begin{array}{c l}
X, & n\in\C{I}^H_X \\
S, & n\in\C{I}^H_S
\end{array}
\right .
\end{equation}
and $\mathcal{I}^H_X$ and $\mathcal{I}^H_S$ are disjoint indexing sets. For such a product, the notation $H_{(i,j)}$ is defined as ,
\beq
H_{(i,j)}=\prod^j_{n=i}\chi_i
\eeq
\noindent The loop equations for $W_{S;H}$, can be generated by the following changes of variables.
\bea
\label{LE4d}
X &&\rightarrow X+\epsilon\frac{1}{z-S} \Tr{H(S,X)} \\
\label{LE4a}
X &&\rightarrow X+\epsilon \left( \frac{1}{z-S}\frac{1}{x-X}+\frac{1}{x-X}\frac{1}{z-S} \right) \Tr{H(S,X)} \\
\label{LE4b}
S &&\rightarrow S+\epsilon \left( \frac{1}{z-S} \frac{1}{x-X}+\frac{1}{x-X} \frac{1}{z-S} \right) \Tr{H(S,X)} \\
\label{LE4c}
X &&\rightarrow X+\epsilon \left( \frac{1}{z_1-S}\frac{1}{x-X}\frac{1}{z_2-S}+h.c \right) \Tr{H(S,X)}
\eea
The loop equations generated by $X\rightarrow X+\epsilon F(X,S)\Tr{H}$, may by obtained from the loop equations generated by $X\rightarrow X+\epsilon F(X,S)$ by the following procedure,
\bea
\label{looprules}
W_{{G_1}(X,S)} &&\rightarrow W_{{G_1}(X,S);H} \nn \\ 
W_{{G_1}(X,S)} W_{{G_2}(X,S)} &&\rightarrow W_{{G_1}(X,S)} W_{{G_2}(X,S);H}+W_{{G_2}(X,S)} W_{{{G_1}(X,S)};H} \nn \\ 
\frac{1}{N^2}(...)&& \rightarrow \frac{1}{N^2}(...) + \sum_{i \in \C{I}^H_X} W_{H_{(1,i)}F(X,S) H_{(i,\C{N}_H)}}(z)
\eea
where ${G_i}(X,S)$ is any function of $X$ and $S$. An equivalent set of rules applies for loops equations generated by $S\rightarrow S+\epsilon F(X,S)\Tr{H}$. The resulting loop equations can then be solved by using the same method as used to solve the one-loop amplitude equations. For example, the two loop amplitude $W_{S;S}(z;u)$ may be obtained by choosing $H=(u-S)^{-1}$ in \eqref{LE4d}, \eqref{LE4a}, \eqref{LE4b} and \eqref{LE4c}. Applying the rules in \eqref{looprules} we get,
\bea
-\frac{1}{g}\Omega^P_{X;S}(z,x;u)&=&\Omega_{SX;S}(z,x;u)\Omega_{S}(z,x)+ \nn \\ 
&&\Omega_{SX}(z,x)\Omega_{S;S}(z,x;u)+\partial_u \left[\frac{W_{SX}(u)-W_{SX}(z)}{u-z}\right]
\eea
and
\bea
\Omega_{SX}(z,x)\Omega_{SX;S}(-z,x;u)&+&\Omega_{SX;S}(z,x,u)\Omega_{SX}(-z,x) = \nn \\
&&\frac{1}{g} \left(W_{S;S}(z;u)+W_{S;S}(-z;u)+\alpha W_{X;S}(x;u)\right)
\eea
By eliminating $\Omega_{SX;S}(z,x,u)$ we obtain an equation relating $W_{S;S}(z;u)$, $W_{S;S}(-z;u)$ and $W_{X;S}(x;u)$. If we arrange the resulting equation so that the LHS is polynomial in $z$, we may expand the RHS in $x$ about $x=\infty$ to generate a number of equations for $W_{S;S}(z;u)$ and $W_{S;S}(-z;u)$. This is exactly the same procedure we used to compute the disc amplitudes in which the resulting expression contained a number of unknown constants, which for the equal potential case were $\avg{\Tr{S}}$ and $\avg{\Tr{S^2}}$. However here we get an expression for $W_{S;S}(z;u)$ in terms of unknown functions of the form $W^{S}{}_{;S}(u)$ and $W^{S^2}{}_{;S}(u)$. These unknown functions may be found by computing the expansion of $W_{S;S}$ at $z=\infty$ and, by the symmetry of $W_{S;S}$, equating it to the expansion at $u=\infty$. The expressions for $W^{S}{}_{;S}(z)$ and $W^{S^2}{}_{;S}(z)$ then depend on a number of unknown constants of the form, $\avg{\Tr{S^m}\Tr{S^n}}_c$, not all of which are independent, as can be seen by considering the large $z$ expansion of the $W^{S}{}_{;S}(z)$ and $W^{S^2}{}_{;S}(z)$. In the end the independent unknown quantities correspond to $\avg{\Tr{S}\Tr{S}}_c$ and $\avg{\Tr{S}\Tr{S^2}}_c$ which can be fixed by requiring that $W^{S}{}_{;S}(z)$ has no singularities besides those at the branch points of $W_S$. The loop equations then give $W_{X;S}$ in terms of $W_{S;S}$. 

In order to compute the $1/N$ corrections it is clear from \eqref{LEwithN} that $W_{SX;S}$, $W_{SX;X}$ and $W_{SX;SX}$ need to be calculated. The calculation of $W_{SX;X}$ and $W_{SX;SX}$, follows a similar procedure however the resulting loop equations contain new amplitudes of the form $W_{SX..SX}$. The calculation of these new quantities does not pose a great difficulty; appropriate changes of variables give loop equations very similar to the ones presented previously. We include the necessary change of variables in Appendix \ref{MoreLoopEqns}.

Once these quantities have been calculated and substituted in \eqref{LEwithN}, we may again expand both sides around $x=\infty$ and then $z=\infty$ to find an expression for ${W_S(z)}^{(1)}$ and ${W_X(x)}^{(1)}$. Again this procedure introduces a new unknown constant corresponding to the $\frac{1}{N}$ correction to $\avg{\Tr{S}}$ which may be determined by requiring that ${W_S(z)}^{(1)}$ does not possess any poles besides those at the branch points of ${W_S(z)}^{(0)}$.

\subsection{Disc, Cylinders and Disc-with-Handles.}
To take the scaling limit of $W_X$ at $\alpha = -1$ we may use the results of \cite{Eynard:2002kg} or the results in the last section. The critical value of $x$ is $x=-c$ and the solution to the loop equation gives, in the scaling limit,
\beq \tW_{M_1}(\mu_B,\mu)=-\frac{1}{2^{{\frac{5}{3}}}5^\frac{1}{3}c}\left(\left(\mu_B+\sqrt{\mu_B^2-\mu}\right)^\frac{4}{3}+ \left(\mu_B-\sqrt{\mu_B^2-\mu}\right)^\frac{4}{3}\right)\label{WX}.\eeq
For ${W_S}^{(0)}$ we may use \cite{Eynard:1992cn, Eynard:1995nv} or solve \eqref{WSeqn}. The critical value of $x$ is at $x=0$ and computing the scaling limit of ${W_S}^{(0)}$ we get,
\beq
\tW_{M_1+M_2}(\mu_B,\mu)=-\frac{1}{2^{{\frac{5}{3}}}5^\frac{1}{3}c}\left(\left(\frac{\mu_B}{\sqrt{2}}+\sqrt{\frac{\mu_B^2}{2}-\mu}\right)^\frac{4}{3}+ \left(\frac{\mu_B}{\sqrt{2}}-\sqrt{\frac{\mu_B^2}{2}-\mu}\right)^\frac{4}{3}\right).\label{WS}
\eeq
Here we run into one final technicality. For a resolvent that produces graphs with some configuration of spins on the boundary $W_{F(M_1,M_2)}$, there exist an entire family of resolvents which produce the same configuration of boundary spins but with different overall weighting
\beq
W_{\frac{1}{\lambda} F(M_1,M_2)}(x) = \lambda W_{F(M_1,M_2)}(\lambda x),
\eeq
where $\lambda$ is a constant. Changing $\lambda$ has the affect of renormalising $\mu_B$, however we want the same normalisation for all continuum quantities and so this fixes $\lambda$ to be $\sqrt{2}$ for $W_{M_1+M_2}$. Therefore using \eqref{WS} and \eqref{WX}, we have that 
\beq
\omega_{(M_1+M_2)/\sqrt{2}}\,(\mu_B,\mu)=\sqrt{2}\omega_{M_1}(\mu_B,\mu).
\eeq
This is precisely the relation between these two disc amplitudes found in the continuum calculation.

One may wonder what the scaling limit of ${W_{X}}^{(0)}$ is for generic values of $\alpha$; given that there are only a finite number of Cardy states the expectation is that there should be no new non-trivial scaling limits. For generic $\alpha$ the critical values of $x$ occur at $\frac{2c}{ \alpha-1}$ and $\frac{2c}{ \alpha+1}$. Taking the scaling limit, for both critical values of $x$, gives
\beq
\omega_{-X/\alpha}(\mu_B,\mu)=\omega_{M_1}(\mu_B,\mu)\label{WXa}
\eeq
which confirms our expectation. We have therefore demonstrated that we can build all the boundary states of the $(3,4)$ minimal string. This will be confirmed further by reproducing the cylinder amplitudes shortly. The nature of the scaling limit when $\alpha = 0$ will be addressed in the Section \ref{DualBrane}.

The one case in which we are interested which cannot be found in the literature is $W_{S;X}$. The computation, using the formalism of the preceding section, results in
\bea
\label{WSXeqn}
\omega_{S;X}({\mu_B}_1,{\mu_B}_2,\mu)=-\frac{1}{10 c_c^2} \log(-1 + 2 z_1^2 + 2 \sqrt{2} z_1 z_2 + 2 z_2^2)
\eea
for the scaling form of $W_{S;X}$, which is in exact agreement with the Liouville result (up to renormalisations of $\mu_B$ and $\mu$). Of course the other cylinder amplitudes may computed in this matrix model or using \cite{Eynard:2002kg}\cite{Eynard:1995nv} and we find the scaling forms to be,
\bea
&&\omega_{X;X}({\mu_B}_1,{\mu_B}_2,\mu)=\frac{1}{10 c_c^2} \ln \frac{z_1-z_2}{x_1-x_2},\\
&&\omega_{X;Y}({\mu_B}_1,{\mu_B}_2,\mu)= -\frac{1}{10 c_c^2} \ln z_1+z_2,\\
&&\omega_{S;S}({\mu_B}_1,{\mu_B}_2,\mu)=\frac{1}{10 c_c^2} \ln \frac{z_1-z_2}{(z_1+z_2)(x_1-x_2)}, 
\eea
again in agreement with the Liouville calculations. The fact that these amplitudes agree with the Liouville calculations \emph{including} terms that were classified in \cite{Kutasov:2004fg} as non-universal suggests they really should have a physical meaning (and shouldn't be thrown away). This issue was raised in \cite{Belavin:2010pj}.

In order to test if P3 is a property of all deviations we need to examine more complicated amplitudes. The $1/N$ correction to the disc amplitude is the amplitude for a disc-with-handle. These computations have already been done in \cite{Eynard:1995nv} and \cite{Eynard:2002kg}, the result being,
\bea
\label{DWHresolv}
{\omega_{M_1}^{(1)}(\mu_B,\mu)}={\omega_{M_2}^{(1)}(\mu_B,\mu)}&=&\frac{1}{2^\frac{1}{3} 5^\frac{2}{3} 36c_c }\frac{z (7 - 24 z^2 + 48 z^4)}{(-1 + 4 z^2)^3}\mu^{-\frac{7}{6}},\\
{\omega_{(M_1+M_2)/\sqrt{2}}^{(1)}(\mu_B,\mu)}&=&\frac{\sqrt{2}}{2^\frac{1}{3} 5^\frac{2}{3} 36c_c }\frac{z (5 - 24 z^2 + 48 z^4)}{(-1 + 4 z^2)^3} \mu^{-\frac{7}{6}}.
\eea
If we now compute the deviations we find,
\bea
\label{DeltaHandle}
\Delta{\omega_{M_2}^{(1)}(\mu_B,\mu)} &=& 0,\\
\Delta{\omega_{(M_1+M_2)/\sqrt{2}}^{(1)}(\mu_B,\mu)} &=& \frac{2^\frac{1}{6} 4}{5^\frac{2}{3} 9 c_c} \frac{z (-3 + 41 z^2 - 208 z^4 + 416 z^6 - 512 z^8 + 256 z^{10})}{ (-1 + 2 z)^3 (1 + 2 z)^3 (1 - 16 z^2 + 16 z^4)^3} \mu^{-\frac{7}{6}}. \nn
\eea

We now want to consider if $\Delta{\omega_{(M_1+M_2)/\sqrt{2}}^{(1)}(\mu_B,\mu)}$ possesses P3.  For it to possess P3 we require that it can be expressed as,
\beq
\Delta{\omega_{(M_1+M_2)/\sqrt{2}}^{(1)}(\mu_B,\mu)}=A(x)+B(\tilde {x}),
\eeq
where $\C{L}^{-1}_x[A]$ and $\C{L}^{-1}_{\tilde{x}}[B]$ \footnote{We denote in the subscript the variable the inverse Laplace transform is taken with respect to.} have point-like support. Since the definition of the deviation is a linear combination of a number of amplitudes, the functions $A$ and $B$ will also be a linear combination of contributions from each amplitude. Since the theory is unitary the contributions to $A$ and $B$ from each amplitude will be positive. From the expressions in \eqref{DWHresolv} we see that each amplitude goes as at most $\sim \frac{1}{z}$ for $z \rightarrow \infty$ and therefore any contribution to $A(x)$ or $B(\tilde{x})$ must go as at most $\sim x^{-1/3}$ or $\sim {\tilde{x}}^{-1/4}$ respectively. We therefore may conclude that $A(x) \sim x^{-1/3}$ and $B(\tilde{x}) \sim {\tilde{x}}^{-1/4}$ as $x \rightarrow \infty$ and hence that their inverse Laplace transforms exist. Furthermore, given that as $z \rightarrow \infty$,
\beq
\label{dwhasymp}
\Delta{\omega_{(M_1+M_2)/\sqrt{2}}^{(1)}(\mu_B,\mu)} \sim \frac{1}{z^7}
\sim\frac{1}{x^{7/3}}\sim\frac{1}{\tilde{ x}^{7/4} },
\eeq
we see that $A$ and $B$ can not fall off faster than $A(x)= a x^{-7/3} +o(x^{-7/3})$ or $B(x)=b{x^{-7/4}} +o(x^{-7/4})$. However it is a well known theorem that a function with only point-like support may be expressed as a sum of derivatives of $\delta$-functions;
\beq
f(t) = \sum^\infty_{m=1}\sum^\infty_{n=1} a_{n,m} \delta^{(n)}(t - t_{n,m}).
\eeq
Such a function will have a Laplace transform of
\beq
\label{inverseLaplace}
\C{L}[f](s) = \sum^\infty_{m=1}\sum^\infty_{n=1} a_{n,m} s^n e^{-s t_{n,m}}.
\eeq
Given that we require that $\C{L}^{-1}_x[A]$ and $\C{L}^{-1}_{\tilde{x}}[B]$ are point-like supported this implies $A(x)$ and $B(x)$ have an asymptotic behaviour as $x \rightarrow \infty$ inconsistent with \eqref{dwhasymp} and hence the deviation can not possess P3.

\subsection{Dual Branes}
\label{DualBrane}
In \cite{Seiberg:2003nm} the two matrix model description of the $(p,q)$ minimal string was considered and it was shown that the continuum limit of the resolvent for one of the matrices was the disc amplitude with a fixed spin boundary. The continuum limit of the resolvent for the other matrix gave the dual brane amplitude. In our matrix model \eqref{ZSX} it is not obvious if we can construct a correlation function which gives the disc with dual brane boundary conditions. In fact it is present in the form,
\beq
W_A(x) = \avg{\frac{1}{N}\Tr\frac{1}{x-(M_1-M_2)}}
\eeq
which may be obtained from the general $X$ resolvent \eqref{genXresolvent} by setting $\alpha = 0$. The resulting loop equation is easily solved as it is quadratic in $x^2$, however the scaling limit is obtained using a different scaling of $x$,
\beq x={x}_c+a^\frac{4}{3}\tilde{\mu}_{B}\sqrt{g_c/c},\label{dualxscale}\eeq
which results in
\beq \tW_{k(M_1-M2)}(\tilde{\mu}_B,\mu)=-\frac{1}{2^{{\frac{5}{3}}}5^\frac{1}{3}c}\left(\left(\tilde{\mu}_B+\sqrt{\tilde{\mu}_B^2-\tilde{\mu}}\right)^\frac{3}{4}+ \left(\tilde{\mu}_B-\sqrt{\tilde{\mu}_B^2-\tilde{\mu}}\right)^\frac{3}{4}\right),\eeq
where $k$ is a numerical constant. If we identify $\tilde{\mu}_B$ in this expression with the dual boundary cosmological constant then this is precisely the amplitude expected for the dual brane obtained from the spin $+$ or spin $-$ boundary state. This explains the $\alpha=0$ possibility we left unexplained earlier in Section \ref{MatrixSec}.

\section{The Seiberg-Shih Relation and Local Operators}
\label{Localstates}
Suppose that the boundaries considered in the previous sections could be expressed as a sum over local operators in the theory, such as,
\beq
\label{localopexp}
\bra{k,l;\sigma_1} \sim \sum^\infty_{n=1} f_n(k,l,\sigma_1) \bra{O_n},
\eeq
where $O_n$ represents some generic local operator and $\sim$ means equivalence up to terms that diverge as $\sigma_1 \rightarrow \infty$. If this were true then the effects of the Seiberg-Shih transformations on a boundary would reduce to studying their effect on the above series. The expansion of a boundary in terms of local operators was investigated in \cite{Moore:1991ir} in which it was used to compute the one and two point correlation numbers of tachyon operators on the sphere for the $(2,2p+1)$ minimal strings. Later it was shown \cite{Anazawa:1997zs} that such a technique could also compute the three point function correctly for the $(p,p+1)$ models.

One technical point raised in \cite{Moore:1991ir} and \cite{Anazawa:1997zs} was that the boundary state could only be expressed as a sum over local operators if we included operators not in the BRST cohomology of \cite{Lian:1992aj}. That the operators in the cohomology of \cite{Lian:1992aj}, which we refer to as LZ-operators, are not all the physical operators of the minimal string is easily verified by noting that the cohomology does not contain the boundary length operator $\oint_{\partial M} dx e^{b\phi}$. Such an operator can be inserted on the disc by differentiating with respect to the boundary cosmological constant, giving the result
\beq
\label{boundaryLOp}
\avg{\oint_{\partial M} dx e^{b\phi}} =\C{A}(p,q) U_{k-1}\left(\cos \frac{\pi q}{p} \right) U_{l-1}\left( \cos \frac{\pi p}{q} \right)\cosh\left(\frac{\pi q \sigma}{\sqrt{pq}}\right),
\eeq
where $\C{A}(p,q)$ is a constant of proportionality that only depends on $p$ and $q$. Indeed there exists a whole family of such operators with Liouville charge $b\frac{p+q-n q}{2p}$ where $n \in \BB{Z}^+$ which we will denote by $\C{U}_{nq}$. 

In order to compute the coefficients $f_n(k,l,\sigma_1)$ appearing in \eqref{localopexp} we need to know the amplitude for any local operator inserted in a disc with a $(k,l)$ FZZT boundary condition. For the tachyon operators the amplitude is \eqref{tachyonDisk}. For the non-tachyon operators we have argued using the mini-superspace approximation and the exact conformal bootstrap on the disc that the functional form of the one-point disc amplitude is \eqref{LZdisc}. Finally, we have seen in \eqref{boundaryLOp} that for the first of the non-LZ operators the disc function has the functional form expected from a direct application of the mini-superspace approximation and also the exact conformal bootstrap. We therefore have,
\beq
\braket{\C{U}_{nq}}{k,l;\sigma} \propto \cosh \left(\frac{\pi n q \sigma}{\sqrt{pq}} \right).
\eeq
If we apply the Liouville duality transformation to the above expression we get
\beq
\braket{\C{U}_{np}}{k,l;\sigma}_{\mathrm{dual}} \propto \cosh \left(\frac{\pi n p \sigma}{\sqrt{pq}} \right),
\eeq
where we have introduced the operators $\C{U}_{np}$ which have Liouville charge $b\frac{p+q-n p}{2p}$. The expansion \eqref{localopexp} can now be given a more concrete form,
\beq
\label{localopexp2}
\bra{k,l;\sigma_1} \sim \sum^\infty_{n=1} f_n(k,l,\sigma_1) \bra{\C{U}_n}.
\eeq

Having introduced the necessary technical results we now want to prove the following claim,
\newline
\newline
{\emph{
If all the FZZT branes in the $(p,q)$ minimal string can be replaced by local operators then the deviation from the Seiberg-Shih relations is caused only by the non-LZ operators.
}}
\newline
\newline
Consider the cylinder amplitude $\C{Z}(k,l;\sigma_1|r,s;\sigma_2)$. If we replace one of the boundaries by a sum of local operators using \eqref{localopexp2} then we get,
\beq
\label{opexpinCyl}
\C{Z}(k,l;\sigma_1|r,s;\sigma_2) = \sum_{n=1}   f_n(k,l,\sigma_1) \braket{\C{U}_n}{r,s;\sigma_2}.
\eeq
By comparison of \eqref{vExplicitZ} with \eqref{opexpinCyl}, we find for the coefficient of the LZ operators,
\bea
\label{fkleqn}
f_n(k,l;\sigma_1) = \frac{A_C}{A_D(n)}\frac{1}{n} U_{k-1}(\cos\frac{\pi n}{p}) U_{l-1}(\cos\frac{\pi n}{q})\mu^{-\frac{n}{2p}}e^{-\frac{\pi n \sigma_1}{\sqrt{pq}}} \equiv \tilde{f}_n(k,l)e^{-\frac{\pi n \sigma_1}{\sqrt{pq}}},
\eea
where $A_C$ is just a numerical constant, $A_D(n) = \Gamma(1-\frac{n}{q})\Gamma(-\frac{n}{p}) \left( \sin(\frac{\pi n}{q}) \sin(\frac{\pi n}{p}) \right)^{1/2}$ and $n$ is not a multiple of $p$ or $q$. Furthermore, note that the LZ operator coefficients satisfy
\beq
f_n(k,l;\sigma_1) = \sum^{k-1}_{a=-(k-1),2} \sum^{l-1}_{b=-(l-1),2} \tilde{f}_n(1,1)e^{-\frac{\pi n}{\sqrt{pq}}\left( \sigma_1 + i\frac{qb+pa}{\sqrt{pq}}
\right).} 
\eeq
This shows that if all boundary states admit an expansion in terms of local operators then the Seiberg-Shih relations transform the coefficients of the LZ-operators correctly. Hence any deviation from the Seiberg-Shih relations must come from the non-LZ operators.
However, it is the non-LZ operators one would expect to give deviations compatible with P3 as they correspond to boundary operators and their duals. One might then suspect that in fact the boundary states may not be expressed as a sum over local operators. We will now see this suspicion is borne out by examination of the cylinder amplitudes. 

If we note that for cylinder amplitudes, $\C{Z}(k,l;\sigma_1|r,s;\sigma_2) = \C{Z}(k,l;\sigma_2|r,s;\sigma_1)$ then together with \eqref{localopexp2}, this implies,
\beq
\sum^\infty_{n=1} f_n(k,l,\sigma_1) \braket{\C{U}_n}{r,s;\sigma_2} = \sum^\infty_{n=1} f_n(r,s,\sigma_1) \braket{\C{U}_n}{k,l;\sigma_2}.
\eeq
Since $\braket{\C{U}_n}{r,s;\sigma} = A^{(r,s)}_n g_n(\sigma)$ where $g_n$ form a set of linearly independent functions then this implies
\beq
f_n(k,l,\sigma) = A^{(k,l)}_n h_n(\sigma),
\eeq
where $h_n$ is some function. We conclude that if all states admit an expansion of the form \eqref{localopexp2} then cylinder amplitudes can be expressed as,
\beq
\label{factorform}
\C{Z}(k,l;\sigma_1|r,s;\sigma_2) = \sum^\infty_{n=1} A^{(k,l)}_n A^{(r,s)}_n h_n(\sigma_1) g_n(\sigma_2).
\eeq
We now want to see if the $\sigma$-independent coefficients of each term in \eqref{vExplicitZ} has the above property. Consider the coefficients appearing in the second sum of \eqref{vExplicitZ},
\bea
F_{k,l,r,s}(\frac{\pi n}{q}) &&\frac{\sin(\pi n /q)}{\sin(\pi p n/q)} = \\
&&\sum^{\lambda_p(k,r)}_{\eta=|k-r|+1,2} \sum^{\lambda_q(l,s)}_{\phi=|l-s|+1,2}
 (-1)^{\eta n} U_{\phi-1}\left( \cos \frac{\pi p n \phi}{q} \right) \left(\Theta(q\eta-p \phi)-\frac{\eta}{p}\right), \nn
\eea
where $\lambda_p(k,r)=\mathrm{min}(k+r-1,2p-1-k-r)$ and $n\neq 0 \mod q$. We want to know whether this can be written in the form $A^{(k,l)}_n A^{(r,s)}_n$. That this is not true in general is shown by considering the case $(p,q)=(4,5)$, for which we have,
\bea
&&F_{2,2,2,1}(\frac{\pi n}{q}) \frac{\sin(\pi n /q)}{\sin(\pi p n/q)}=0,  \\
&&F_{1,1,2,1}(\frac{\pi n}{q}) \frac{\sin(\pi n /q)}{\sin(\pi p n/q)}=\half, 
\eea
which is clearly incompatible with the factorisation required \eqref{factorform}. This lack of factorisation appears to be generic; for instance if we consider the $(3,4)$ model we find that there is not even a subspace of states which admit an expansion in terms of local operators. This can be seen by considering the matrix,
\bea\label{nonnulldet}
&&F_{k,l,r,s}(\frac{\pi n}{q}) \frac{\sin(\pi n /q)}{\sin(\pi p n/q)}=
\left( \begin{array}{ccc}
\frac{2(-1)^n}{3} & \third & \frac{2}{3}\cos \frac{3 \pi n}{4} \\
\third & \frac{2(-1)^n}{3} & \frac{2}{3}(-1)^{n+1}\cos \frac{3 \pi n}{4} \\
\frac{2}{3}\cos \frac{3 \pi n}{4} & \frac{2}{3}(-1)^{n+1}\cos \frac{3 \pi n}{4} & \frac{1}{3}(-1)^n(3 - 4 \cos^2\frac{3 n \pi}{4}) \end{array} \right)\nn \\
\eea
where the columns and rows correspond to $(k,l) = (1,1), (2,1), (2,2)$.  If there exist $A^{(k,l)}_n$ and $A^{(r,s)}_n$ such that $F_{k,l,r,s}(\frac{\pi n}{q}) =A^{(k,l)}_n A^{(r,s)}_n$ for some subset of the states, then at least one of the cofactors of the matrix in \eqref{nonnulldet} must be zero, which is not the case.

The same conclusions can be reached by studying the other coefficients associated to the non-LZ operators in \eqref{vExplicitZ}. This leaves us in an awkward position; there is no subspace of states that can be expanded in terms of local operators unless we allow the subspace to be one-dimensional (in which case why choose a particular boundary as being the one that can be expanded in terms of local operators?), yet from the results of \cite{Moore:1991ir} and \cite{Anazawa:1997zs} it is clear information about insertions of operators can be extracted from loops. It is worth noting that from the examination of cylinder amplitudes in which one of the boundaries is a ZZ brane, it appears that ZZ branes do not couple to non-LZ states \cite{Gesser:2010fi,Ambjorn:2007ge,Ambjorn:2007xe} and so some of the above problem might be avoided if a ZZ brane is present.

\subsection{Extracting Local Operators From Loops}

An interesting observation is that all correlation numbers were extracted from loops with fixed spin boundary conditions. We now will present further evidence that correlation numbers can be extracted from the fixed spin boundary by reproducing the results of \cite{Belavin:2010pj} for the 1-point correlation numbers on the torus in the $(2,5)$ minimal string. This also serves as an independent check on the computation in \cite{Belavin:2010pj} as our methods are quite different.

A trick that makes this computation easy is that all amplitudes for the minimal string are easily expressible in terms of $z$, which is the uniformising parameter of the auxiliary Reimann surface \cite{Seiberg:2003nm} or equivalently spectral curve. If we write these amplitudes in terms of a new variable, $w$, defined by $z = \half(w+w^{-1})$ we may easily compute the large $w$ expansion of the amplitudes. Since $z = \cosh \frac{\pi \sigma}{\sqrt{pq}}$, the large $w$ expansion will be an expansion in terms of the functions $e^{\frac{\pi  \sigma}{\sqrt{pq}}}$ which is exactly what is required to compute the correlation numbers. Explicitly, if we have an amplitude $\braket{1,1;\sigma_1}{X}$ then
\bea
\braket{1,1;\sigma_1}{X} = \sum^{\infty}_{n = 1} A_n w^{-n} = \sum^{\infty}_{n = 1} A_n e^{-\frac{\pi n\sigma_1}{\sqrt{pq}}},
\eea
but we also have
\bea
\braket{1,1;\sigma_1}{X} = \sum^\infty_{n=1}\tilde{f}_n(1,1)  e^{-\frac{\pi n\sigma_1}{\sqrt{pq}}}\braket{\C{U}_n}{X},
\eea
giving,
\bea
\braket{\C{U}_n}{X} = \frac{A_n}{\tilde{f}_n(1,1) }.
\eea

We now have the tools to extract the correlation numbers on the torus beginning with the disc-with-handle amplitude. The calculation in \cite{Belavin:2010pj} of the torus one-point correlation numbers was performed for the $(2,2p+1)$ minimal string. The disc-with-handle amplitude for the $(2,2p+1)$ minimal string has been computed numerous times using matrix model techniques. We will restrict ourselves to the case of the $(2,5)$ model in which the disc-with-handle amplitude takes the form,
\beq
\label{dischandle25}
{\omega_{M}(\mu_B,\mu)}^{(1)} = \frac{1 + 12 z^2}{z^3}\mu^{-\frac{7}{4}}.
\eeq
The boundary condition on the disc is the equivalent of our spin up boundary as it is computed using the resolvent of $M$ where $M$ is the only matrix appearing in the matrix model. We will now use this to compute the torus one point correlation numbers, for tachyon operators, by replacing the boundary with a sum of local operators as in \eqref{localopexp2}. That this computation produces results that match exactly the results of \cite{Belavin:2010pj} is evidence that we may replace this boundary by a sum of local operators. Obviously, in order to compare results we have to renormalise the bulk cosmological constant, boundary cosmological constant and $M$ to be consistent with \cite{Belavin:2010pj}. In order to avoid doing this we note that the effect of such a renormalisation will to be to change the amplitude by a factor dependent only on $p$ and $q$ and so will not affect the ratio of correlation numbers. We therefore shall compare the ratio of correlation numbers. Computing the large $w$ expansion of \eqref{dischandle25} we get,
\beq
{\omega_{M}(\mu_B,\mu)}^{(1)}  = 4 \mu^{-\frac{7}{4}} \sum^{\infty}_{n=0}  (-1)^{n+1} (n-2)(n+3) w^{-2n-1}
\eeq
where the series converges if $w>1$. We therefore find the torus $1$-point correlation number for tachyon operators in the $(2,5)$ minimal string to be,
\bea
\label{1pointfromloop}
\avg{\C{U}_n}^{(1)} = \frac{n}{A_C} \left( \sin\frac{\pi n}{2} \sin\frac{\pi n}{5} \right)^{\half} \Gamma\left(1-\frac{n}{5}\right)\Gamma\left( -\frac{n}{2}\right)(-1)^{\frac{n-1}{2}}(n-5)(n+5) \mu^\frac{n-7}{4} 
\eea
when $n$ is odd and zero if $n$ is even. The 1-point correlation numbers on the torus for the $(2,2p+1)$ models were computed in \cite{Belavin:2010pj}, giving,
\beq
\avg{\C{U}_n}^{(1)} = \avg{\C{U}_1}^{(1)} (-1)^\frac{n-1}{4}\left(\frac{\sin\frac{\pi n}{2p+1}}{\sin\frac{\pi }{2p+1}} \right)^\half \frac{\Gamma(1-\frac{n}{2p+1})\Gamma(\half)}{\Gamma(1-\frac{1}{2p+1})\Gamma(\frac{n}{2})} \frac{(2p+1 -n)(2p+1 +n)}{4p(p+1)}\mu^\frac{n-1}{4},
\eeq
for $n$ odd and zero when $n$ is even. Using standard $\Gamma$-function identities it is then easy to show that the above equation is consistent with our computation \eqref{1pointfromloop}.

The reproduction of the torus correlation numbers is evidence that this procedure is correct and that the fixed spin amplitude can be used to compute correlation numbers of local operators. Furthermore, we can obtain evidence that the non-fixed spin boundary states do not admit an expansion in terms of local operators by applying the above procedure to the disc-with-handle amplitudes we computed for the $(3,4)$ model. We have the expansions,
\bea
{\omega_{M_1}^{(1)}(\mu_B,\mu)}&\sim&\frac{1}{2^\frac{1}{3} 5^\frac{2}{3} c \Lambda^\frac{7}{6}}\left[ \frac{1}{24 w} + \frac{1}{72 w^5} - \frac{5}{72 w^7} + \frac{17}{72 w^{11}} + \C{O}(w^{-13})\right] \\
{\omega_{(M_1+M_2)/\sqrt{2}}^{(1)}(\mu_B,\mu)}&\sim&\frac{\sqrt{2}}{2^\frac{1}{3} 5^\frac{2}{3} c \Lambda^\frac{7}{6}}\left[ \frac{1}{24 w} - \frac{1}{72 w^5} - \frac{1}{72 w^7} + \frac{7}{72 w^{11}}  +\C{O}(w^{-13}) \right]
\eea
If the expansion in terms of local operators is valid for both these boundary states then we would expect the coefficients of each term in the large $w$ expansion to be related by a factor of $\sqrt{2}$ as can be seen from \eqref{fkleqn}; this is clearly not the case. The implication of these results is that the fixed spin boundary condition is special and is the only one that admits an expansion in terms of local operators. 

\section{Discussion}
\label{conclusion}
In this chapter we have examined the Seiberg-Shih relations \eqref{SS1} for cylinders and discs-with-handles. We saw in the previous chapter that for cylinders the identification of FZZT branes is naively spoiled by terms arising from the gravitational sector of the theory, as first seen by \cite{Kutasov:2004fg}\cite{Gesser:2010fi}. We found that the terms spoiling the identification cannot always be associated with degenerate geometries of the worldsheet, in the sense of \cite{Kutasov:2004fg}, as other geometries which are degenerate in a dual sense are also present. This lead us to conjecture that the terms spoiling the FZZT brane identification could always be interpreted as arising from degenerate and dual degenerate geometries. We checked this conjecture by computing the disc-with-handle amplitudes for free and fixed spin boundary states in the $(3,4)$ minimal string using matrix model techniques. We found that the deviation from the Seiberg-Shih relations in this case could not  be written in the way we conjectured. Given that our conjecture was a very generous interpretation of what terms might be unphysical we conclude this is strong evidence that the FZZT brane identification conjectured in \cite{Seiberg:2003nm} does not hold at all levels of perturbation theory.

We also considered an alternative approach to testing the Seiberg-Shih relations by expanding boundaries in terms of local operators. If such an expansion were possible we showed that it would lead to deviations from the Seiberg Shih relations consistent with our conjecture. We then gave explicit examples for cylinder amplitudes where a local operator expansion of boundary states fails. This lead us to a paradox; how could correlation numbers of local operators be extracted from boundary states in \cite{Moore:1991ir} \cite{Anazawa:1997zs} if such boundary states cannot be expanded as local operators? To investigate this we computed the disc-with-handle amplitude in the $(2,5)$ model using matrix model methods and then showed that the fixed spin boundary condition yielded an expansion from which correlation numbers in agreement with the results of \cite{Belavin:2010pj} could be extracted. This lead us to conjecture that the fixed spin boundaries are special in that they allow an expansion in terms of local operators.

There are many obvious generalisations of the work done here, the most obvious being that it would be interesting to extend the above results to the $(p,q)$ case. Since the Liouville techniques are not suited to higher-genus computations such a project would have to be tackled using matrix model methods. The main problem with this approach is that it is difficult to represent all the states for a given minimal string in a matrix model; in the above we chose the Ising model precisely because there was an obvious mapping between the spin degrees of freedom and the matrix fields. However this problem may now be less serious given recent work \cite{Ishiki:2010wb,Bourgine:2010ja} in which boundary states were constructed in the matrix model formulation for $(2,2p+1)$ minimal strings. Another strategy for attacking the same problem may be the geometric recursion techniques developed by Eynard et al. in \cite{Eynard:2008we} as this provides an efficient way of computing many amplitudes for a given matrix model. Such methods would also allow a relatively easy computation of the cylinder-with-handle amplitudes; these would be interesting to study as one could then check if quantum corrections destroy the fact, noted in \cite{Gesser:2010fi,Ambjorn:2007ge,Ambjorn:2007xe}, that FZZT-ZZ cylinder amplitudes are consistent with the Seiberg-Shih relations.

Another generalisation would be to perform a similar investigation in other non-critical string models. An obvious choice would be non-critical superstrings as they perhaps have more physical relevance in addition to being much better behaved. A second and perhaps less obvious choice of string model would be the causal string theories developed in \cite{Ambjorn:2008gk,Ambjorn:2008hv,Ambjorn:2009fm}. Currently these models have only been solved for a target space of zero dimensions and so such a project would require a generalisation of the models to include worldsheet matter. However, even in the zero dimensional case such causal string models are better behaved and many of the odd features of the usual non-critical string are absent due to the lack of baby-universe production; in particular the gravity sector in these theories seems to be much weaker. Since the terms spoiling the FZZT brane identification are due to the gravity sector (at least for the cylinder amplitudes) might they have an interpretation in terms of baby-universe over production and if so might they be absent in these causal string models?

Finally, we saw that the matrix model is able to reproduce all the boundary states found in the $(3,4)$ minimal string apart from the dual of the free spin boundary condition. In \cite{Seiberg:2003nm} this was enough as they claimed that there is only one brane and hence only one dual brane in the theory. However, in light of our results there should exist other dual branes in the theory. Are such states present in the matrix model and if so can they be represented as some form of resolvent? This question and the others outlined we leave to future work.


\chapter{Summary\label{ChapSummary}}
In this thesis we have reviewed how random graphs provide a very useful tool by which observables in a variety of two dimensional quantum gravity theories may be computed. Using these techniques we were able to investigate conjectures related to two of these observables; the dimension of spacetime and Hartle-Hawking wavefunctions, or more generally spacetime containing a number of boundaries. 

In Chapter \ref{ChapBack} after introducing the approach to two dimensional quantum gravity in which the metric is treated as the fundamental degree of freedom we reviewed how the resulting path integral could be defined via a lattice regularisation referred to as Dynamical Triangulation. This lattice regularisation can naturally be interpreted as a random graph and we reviewed how the disc-function observable could be computed in this approach via matrix model and combinatorial methods. We then explored the consequence of the ambiguity inherent in the metric-is-fundamental approach resulting from the necessity of choosing a class of metric to integrate over. In DT, the class of metrics integrated over were the Euclidean metrics on the manifold. We reviewed how the approach of CDT could be use to define path integrals in which only geometries with a well-defined casual structure where included in the integration. These theories differed from DT due to the propensity of DT to generate baby-universes throughout the space; such behaviour ultimately lead to a fractal geometry for the spacetime. 

The question of how to investigate the fractal nature of the spacetimes was taken up in Chapter \ref{ChapDim} in which the concepts of Hausdorff and spectral dimensions were introduced as ways of generalising the notion of dimension beyond smooth manifolds. These measures of the dimension were useful in characterising the properties of the space; we saw that DT gave rise to fractal spacetime since $d_h = 4$ whereas $d_s = 2$. For CDT on the other hand we saw that $d_h = d_s =2$ signalling the emergence of a smooth spacetime in the continuum limit. We then reviewed the recent evidence obtained in a variety of approaches to quantum gravity, most notably CDT, for the spectral dimension being lower in the UV than IR. We referred to this phenomenon as dimensional reduction.

In Chapter \ref{ChapComb} we constructed a toy model based on random comb graphs, in which dimensional reduction occurs. This was useful for a number of reasons; it shows that a scale dependent spectral dimension may be given a rigorous definition and that models can be found in which this rigorous definition results in a phenomenon similar to the dimensional reduction observed in CDT and other approaches. We gave a characterisation of possible continuum limits and the behaviour of the spectral dimension in each case via Theorem \eqref{mainresult}. As discussed in Section \ref{DiscussComb} we hope that these results may be used as the basis for constructing more realistic models in which a similar dimensional reduction can be observed. In particular we hope the work in \cite{Durhuus:2009sm}, in which a type of random multigraph that shares many features with CDT were introduced, may be extended in this direction.

In Chapter \ref{ChapString} we gave a brief review of minimal models, Liouville theory and string theory. This was motivated by the desire to couple matter to the two dimensional gravity models considered in the earlier chapters. After showing how the matrix model techniques are easily extended to describe DT models coupled to matter, we reviewed the equivalence between these models and certain string theories known as $(p,q)$ minimal string theories. We then considered the boundary states, or branes, in these $(p,q)$ minimal strings, in particular focusing on a conjecture of Seiberg and Shih which claimed that all branes occurring in such theories were physically equivalent. This claim is challenged by noting that the simple form of this conjecture fails for cylinder amplitudes. We reviewed the suggestion that since the terms spoiling the conjecture have particular properties, they are non-universal and therefore these examples should not be seen as counter examples. However, for more complicated cylinder amplitudes the spoiling terms fail to have the properties necessary to be classed as non-universal. We introduced a conjecture, referred to as P3, which gives a natural way, based on the duality of Liouville theory, in which these new spoiling terms may also be considered non-universal.

In the final Chapter \ref{ChapBrane} we considered the conjecture P3 by checking if it held for the case of disc with handles in the $(3,4)$ model. This required a matrix model representation of all the conformal boundary states in this model. To accomplish this we introduced a matrix model describing all the Ising model on a random lattice in the presence of a boundary magnetic field. We demonstrated this matrix model could be solved for any potential, not just the cubic one necessary for the Ising model. Using this solution we computed the disc-with-handle amplitude and argued that the terms spoiling the Seiberg-Shih conjecture did not satisfy the extended conjecture. This strongly suggests that the naively different branes in this theory are indeed different. There exist a variety of extensions of this work as discussed in Section \ref{conclusion}.

Finally we considered an alternative approach to testing the Seiberg-Shih conjecture by writing the boundary states as a sum of local operators. This is a particular elegant way of approaching the problem as it negates the need to consider individual amplitudes separately. We found that the local operator representation of the boundary states lead to non-trivial terms which spoil the Seiberg-Shih relation, however, these terms appear to be of exactly the form necessary in order for conjecture P3 to hold. We investigate this apparent contradiction and argue that the expansion in terms of local operators only holds for the identity brane. In doing so we verify certain recent results of Belavin et al. \cite{Belavin:2010pj}.


We hope that this thesis has convinced the reader of the usefulness of random graphs in the study of low dimensional theories of quantum gravity and furthermore that there is a great deal left to explore.




\appendix

\chapter{Mellin Transforms and Asymptotics \label{ChapMellinAsymp:app}}

\section{Standard generating functions}
\label{StandardResults}
We record here a number of standard results for generating functions for random walks on combs; the details of their calculation are given in \cite{Durhuus:2005fq}.

On the empty comb $C=\infty$ we have
\bea P_\infty(x)&=&1-\sqrt{x},\label{Pinf}\\
P_\infty^{<n}(x)&=&(1-x)\frac{(1+\sqrt{x})^{n-1}-(1-\sqrt{x})^{n-1}}{(1+\sqrt{x})^{n}-(1-\sqrt{x})^{n}},\label{PinfL}\\
G_\infty^{(0)}(x;n)&=&(1-x)^{n/2}\frac{2\sqrt{x}}{(1+\sqrt{x})^{n}-(1-\sqrt{x})^{n}}.  \label{GinfL}
\eea
Note that we can promote $n$ to being a continuous positive semi-definite real variable in these expressions; $G_\infty^{(0)}(x;n)$ is then a strictly decreasing function of $n$ and $P_\infty^{<n}(x)$ a strictly increasing function of $n$.
The finite line segment of length $\ell$ has
\bea P_\ell(x)&=& 1-\sqrt{x}\frac{(1+\sqrt{x})^{\ell}-(1-\sqrt{x})^{\ell}}{(1+\sqrt{x})^{\ell}+(1-\sqrt{x})^{\ell}}
\eea
which it is sometimes convenient to write as 
\bea P_\ell(x)
&=&\sqrt{x}\tanh\left(m_\infty(x)\ell\right)\label{Pell}
\eea
where 
\beq m_\infty(x)=\half\log\frac{1+\sqrt{x}}{1-\sqrt{x}}.\eeq
Again, $\ell$ can be promoted to a continuous positive semi-definite real variable of which $P_\ell(x)$ is a strictly increasing function. The first return probability generating function for the comb with teeth of length $\ell+1$ equally spaced at intervals of $n$ is given by 
\bea
P_{\ell, *n} (x)= \frac{3-P_{\ell}(x)}{2} - \half \left[\left(3-P_{\ell}(x)-2P_\infty^{<n}(x)\right)^2-4G_\infty^{(0)}(x;n)^2\right]^\half.\label{Pelln}
\eea
and $P_{*n}(x)$ is obtained by setting $\ell=\infty$ in this formula. $P_{\ell, *n} (x)$ is a strictly decreasing function of $\ell$ and increasing function of $n$, viewed as continuous positive semi-definite real variables.

We also need the scaling limits of some of these quantities. They are
\bea \lim_{a\to 0} a^{-\half}G_\infty^{(0)}(a\xi;a^{-\half}\nu)&=&\xi^\half \mathrm{cosech}( \nu\xi^\half)\eea
and
\bea \fl{\lim_{a\to 0} a^{-\half}\left(1-P_{(a^{-\half}\rho), *(a^{-\half}\nu)} (a\xi)\right)=-\half\xi^\half\tanh(\rho\xi^\half)}\nn\\ \qquad\qquad\qquad \qquad+\half\xi^\half\left[4+4\tanh\rho\xi^\half\coth\nu\xi^\half+\tanh^2\rho\xi^\half\right]^\half.
\eea

\section{Bump functions}
\label{AppendixBump}
A function $\psi : \BB{R} \rightarrow \BB{R}$ is a bump function if $\psi$ is smooth and has compact support. We will now prove some properties concerning the Mellin transformation of a bump function $\psi$, $\C{M}[\psi](s)$, which has support on $[a,b]$ where $b>a>0$.

\begin{lemma}
\label{bump1}
The critical strip of the Mellin transform of the $n$th derivative of $\psi$, $\psi^{(n)}$, is $\BB{C}$ for all $n$.
\end{lemma}
\begin{proof}
Recall that the Mellin transform is defined by, $\C{M}[\psi](s) = \int^{\infty}_0 \psi(x) x^{s-1} dx$. Since $\psi$ has compact support we have,
\beq
\C{M}[\psi^{(n)}](s) = \int^{b}_a \psi^{(n)}(x) x^{s-1} dx
\eeq 
and since $\psi$ is smooth, $|\psi^{(n)}|$ is bounded on $[a,b]$ by some constant $K$, so,
\beq
|\C{M}[\psi^{(n)}](s)| \leq K \int^{b}_a x^{s-1} dx
\eeq
and the RHS is finite for all $s \in \BB{C}$ since $b>a>0$. This also shows that $\C{M}[\psi^{(n)}](s)$ is holomorphic for all $s$.
\end{proof}
\begin{lemma}
\label{bump2}
Given $n \in \BB{Z}^+$, $|\C{M}[\psi](\sigma+i D)| \leq \frac{1}{D^n} \C{M}[|\psi^{(n)}|](\sigma+n)$ for all $s \in \BB{C}$.
\end{lemma}
\begin{proof}
Recall from the previous lemma that the critical strip of the Mellin transform of $\psi$ and its derivatives coincides with $\BB{C}$. We can therefore use the integral representation of the Mellin transform to prove statements valid for all $s \in \BB{C}$. By integration by parts,
\beq
\C{M}[\psi^{(n)}](s) = -\int^{b}_a \psi^{(n+1)}(x) \frac{x^{s}}{s} dx = -\frac{1}{s}\C{M}[\psi^{(n+1)}](s+1)
\eeq
and therefore,
\beq
\fl{\C{M}[\psi](s) = \frac{(-1)^n}{\prod^{n-1}_{k = 0} (s+k)} \int^{b}_a \psi^{(n)}(x) x^{s+n} dx =  \frac{(-1)^n}{\prod^{n-1}_{k = 0} (s+k)}\C{M}[\psi^{(n)}](s+n).}
\eeq
Hence,
\beq
|\C{M}[\psi](\sigma+iT)| \leq \frac{1}{D^n} \int^{b}_a |\psi^{(n)}(x)| x^{\sigma + n} dx = \frac{1}{D^n} \C{M}[|\psi^{(n)}|](\sigma+n).
\eeq
\end{proof}
Given a bump function $\Psi_{\pm\epsilon}$ which is always positive, has support $[1,1\pm\epsilon]$ and is scaled such that its integral is one, we define the cut-off function to be,
\beq
\eta_\pm(x) = 1 - {\int^x_{-\infty} \Psi_{\pm\epsilon}(x) dx }.
\eeq
\begin{lemma}
The critical strip of $\C{M}[\eta_\pm](s)$ is given be $Re[s]> 0$. The analytic continuation of $\C{M}[\eta_\pm](s)$ to all $s$ is given by
\beq
\C{M}[\eta_\pm](s) = \frac{1}{s}\C{M}[\Psi_{\pm\epsilon}](s+1).
\eeq
\end{lemma}
\begin{proof}
The analytic continuation of $\C{M}[\eta_\pm](s)$ may be obtained by applying integration by parts to the Mellin transform of $\eta_\pm$ and recalling by lemma \eqref{bump1} that $\C{M}[\Psi_{\pm\epsilon}](s)$ is holomorphic everywhere.
\end{proof}

\section{Asymptotic Series and Dirichlet Series}
\label{AppendixAsymp}
Starting with 
\beq
{S_\pm}(y) 
= \sum^\infty_{\ell=0} \mu(\ell) \eta_\pm(\ell y)=\mu(0)+ \sum^\infty_{\ell=1} \mu(\ell) \eta_\pm(\ell y)\equiv \mu(0) + S_\pm^{(1)}(y)
\eeq
where $\eta_\pm$ is the smooth cut-off function introduced in \ref{AppendixBump} and $y$ controls where the cut-off occurs we take the Mellin transform of $S_\pm^{(1)}(y)$ with respect to $y$,
\bea
\C M[S^{(1)}_\pm](s) &=& \int^\infty_0 S^{(1)}_\pm(y) y^{s-1} dy \nn \\
&=& \C{D}_\mu(s) \C M[\eta_\pm](s) \nn \\
&=& \C{D}_\mu(s) \C M[\Psi_{\pm\epsilon}](s+1)/s
\eea
where $\C{D}_\mu(s)$ is the Dirichlet series associated to the measure,
\beq
\C{D}_\mu(s) = \sum^{\infty}_{\ell=1} \frac{\mu(\ell)}{\ell^s}.
\eeq
It is easy to see that the fundamental strip of $S^{(1)}_\pm(y)$ is $Re[s]>0$, due to the compact support of $\eta_\pm$ on the positive real axis and so the Mellin transform does indeed exist. We may now invert the Mellin transform to obtain,
\beq
S_\pm(y) = \mu(0)+\frac{1}{2 \pi i}\int^{c+i\infty}_{c-i\infty} y^{-s} \C{D}_\mu(s) \frac{\C M[\Psi_{\pm\epsilon}](s+1)}{s} ds.
\eeq
This may be computed by rewriting the above as,
\beq
S_\pm(y) = \mu(0)+\frac{1}{2 \pi i}\oint_C y^{-s} \C{D}_\mu(s) \frac{\C M[\Psi_{\pm\epsilon}](s+1)}{s} ds -R(y)
\eeq
where the contour $C$ is the rectangle composed of the points $\{c-i\infty,c+i\infty,-N + i\infty,-N-i\infty\}$ with $N>0$, $c$ such the contour is the right of all poles and  
\beq
R(y) = \int^{-N+i\infty}_{-N-i\infty} y^{-s} \C{D}_\mu(s) \frac{\C M[\Psi_{\pm\epsilon}](s+1)}{s} ds.
\eeq
If $\C{D}_\mu$ has slow growth in the strip $-N \leq Re[z]  \leq 0$ then due to lemma \eqref{bump2}, which shows that $\C M[\Psi_{\pm\epsilon}]$ decays faster than any polynomial as $t$ goes to infinity, the contributions from integrating along the contours $c+i\infty$ to $-N + i\infty$ and from $c-i\infty$ to $-N - i\infty$ are zero. Furthermore the remainder term $R(y)$ satisfies,
\beq
|R(y)| \leq \frac{y^{N}}{N} \int^{-N+i\infty}_{-N-i\infty} |\C{D}_\mu(s) ||{\C M[\Psi_{\pm\epsilon}](s+1)}| ds
\eeq
and so will only contribute terms of order $y^N$ to $S_\pm(y)$. We therefore have,
\bea
&&S_\pm(y) = \mu(0)+\frac{1}{2 \pi i}\oint_C y^{-s} \C{D}_\mu(s) \frac{\C M[\Psi_{\pm\epsilon}](s+1)}{s} ds\\
&&= 1+\sum_{s_i \in S \cap \Sigma} \res\left[\C{D}_\mu(s) \frac{\C M[\Psi_{\pm\epsilon}](s+1)}{s}y^{-s};s=s_i\right]-R(y)
\eea
where $S$ is the set of positions of the poles of $\C{D}_\mu$ and we have used the fact that $\C{D}_\mu(0) = 1-\mu(0)$. Finally, define ${\chi_\pm}(u) = \sum^\infty_{l=0} \mu(\ell) \eta_\pm(l/u)$. By relating this function to $S_\pm(y)$ we may write,
\bea
\label{genasymp}
&&\chi_\pm(u) = \mu(0)+\frac{1}{2 \pi i}\oint_C u^{s} \C{D}_\mu(s) \frac{\C M[\Psi_{\pm\epsilon}](s+1)}{s} ds\\
&&= 1+\sum_{s_i \in S \cap \Sigma} \res\left[\C{D}_\mu(s) \frac{\C M[\Psi_{\pm\epsilon}](s+1)}{s}u^{s};s=s_i\right]-R(y).
\eea 

\chapter{Liouville Cylinder Amplitudes \label{ChapLiouvilleCylinder:app}}
\label{Cylindercalc}

The full cylinder amplitudes were given in a usable form recently by \cite{Gesser:2010fi}. We have obtained similar results independently and reproduce  them here for the purpose of completeness. The cylinder amplitude between a $(k,l)$ and $(r,s)$ brane, $\C{Z}(k,l;\sigma_1|r,s;\sigma_2)$, is computed by the integral
\beq
\C{Z}(k,l;\sigma_1|r,s;\sigma_2) = \int^{\infty}_0 d\tau\C{Z}_{\mathrm{ghost}}(\tau) \C{Z}_{\mathrm{liouville}}(\tau) \C{Z}_{\mathrm{matter}}(\tau),
\eeq
where $\tau$ is the modular parameter of the cylinder in the closed string channel. The partition function for the ghosts is well known
\beq
\label{Cylinder1}
\C{Z}_{\mathrm{ghost}}(\tau) =\eta(\B{q})^2,
\eeq
where $\B{q} = \exp(-2 \pi\tau)$ and $\eta(\B{q})$ denotes the Dedekind $\eta$-function. The Liouville and matter contributions are given by
\bea 
\C{Z}_{\mathrm{matter}} = \sum_{(a,b) \in E_{q,p}} \frac{S_{(k,l),(a,b)}S_{(r,s),(a,b)}}{S_{(1,1),(a,b)}} \chi_{(a,b)}(\B{q}),\\
\C{Z}_{\mathrm{liouville}} = \int^{\infty}_0 dP \Psi^{\dagger}_{\sigma_1}(P)\Psi_{\sigma_2}(P) \chi_P (\B{q}), 
\eea
where $E_{q,p}$ is the Kac table for the minimal model, $S$ is the relevant modular $S$-matrix, $\Psi(\sigma)$ is defined in \eqref{FZZTstate} and $\chi_{(a,b)}(\B{q})$ and $ \chi_{P}(\B{q})$ are the characters for the minimal model primary field with Kac index $(a,b)$ and the non-degenerate Liouville primary field respectively. The characters are given by,
\bea
\chi_{(a,b)}(\B{q})&&=\frac{1}{\eta(\B{q})} \sum^{\infty}_{n=-\infty}\left( \B{q}^{(2pqn+qa-pb)^2/4pq}-\B{q}^{(2pqn+qa+pb)^2/4pq}\right),\\
\chi_P (\B{q}) &&= \frac{\B{q}^{P^2}}{\eta(\B{q})}.
\eea
We first concentrate on the $\C{Z}_{\mathrm{matter}}$. Recalling that the modular S-matrix for the $(p,q)$ minimal model is
\beq
S_{(k,l),(r,s)}= 2\sqrt{\frac{2}{pq}}(-1)^{1+lr+ks}\sin(\pi\frac{qkr}{p})\sin(\pi\frac{pls}{q}),
\eeq
the explicit expression for $\C{Z}_{\mathrm{matter}}$ is,
\bea
\C{Z}_{\mathrm{matter}}&&=-2\sqrt{\frac{2}{pq}} \sum^{p-1}_{a=1} \sum^{q-1}_{b=1} \Big[(-1)^{a(l+r+1)+b(k+s+1)}\frac{\sin(\pi\frac{qka}{p})\sin(\pi\frac{plb}{q})\sin(\pi\frac{qra}{p})\sin(\pi\frac{psb}{q})}{\sin(\pi\frac{qa}{p})\sin(\pi\frac{pb}{q})},\nn \\
&&\frac{1}{\eta(\B{q})} \sum^{\infty}_{n=-\infty}\left( \B{q}^{(2pqn+qa-pb)^2/4pq}-\B{q}^{(2pqn+qa+pb)^2/4pq}\right), \Big]
\eea
where we have used the symmetry of the Kac table to rewrite the limits on the $a$ and $b$ summation. Also, note that the quantity in the square brackets is symmetric in $a$ and $b$ and zero if $a=0$ or $b=0$. This allows us to rewrite it as,
\bea
\C{Z}_{\mathrm{matter}}&&=\frac{1}{\sqrt{2pq}} \sum^{p-1}_{a=1} \sum^{q-1}_{b=-(q-1)} \Big[\sin(\frac{\pi t}{p})\sin(\frac{\pi t}{q}) \times \\
&&U_{k-1}(\cos\frac{\pi t}{p})U_{r-1}(\cos\frac{\pi t}{p})U_{l-1}(\cos\frac{\pi t}{q})U_{s-1}(\cos\frac{\pi t}{q}) \frac{1}{\eta(\B{q})} \sum^{\infty}_{n=-\infty}\B{q}^{(2pqn+t)^2/4pq}\Big], \nn
\eea
where $t=qa+pb$. Substituting these expressions into \eqref{Cylinder1} gives,
\bea
&&\C{Z}= \frac{1}{\sqrt{2pq}}\left[ \int^{\infty}_0 dP \frac{4\pi^2 \cos(2\pi\sigma_1 P)\cos(2\pi\sigma_2 P)}{\sinh(2\pi b P)\sinh(\frac{2\pi P}{b})} \right] \nn\\
&& \sum^{p-1}_{a=1} \sum^{q-1}_{b=-(q-1)} \Big[\sin(\frac{\pi t}{p})\sin(\frac{\pi t}{q}) U_{k-1}(\cos\frac{\pi t}{p})U_{r-1}(\cos\frac{\pi t}{p})U_{l-1}(\cos\frac{\pi t}{q})U_{s-1}(\cos\frac{\pi t}{q}),\nn\\
&&\int^{\infty}_{0}d\tau \sum^{\infty}_{n=-\infty}\B{q}^{(2pqn+t)^2/4pq+P^2}\Big].
\eea
Focusing our attention on the integral over the moduli, we compute,
\bea
\int^{\infty}_{0}d\tau \sum^{\infty}_{n=-\infty}\B{q}^{(2pqn+t)^2/4pq+P^2}&&=-\frac{1}{2\pi}\sum^{\infty}_{n=-\infty}\left[ \frac{1}{(2pqn+t)^2/4pq+P^2} \right], \nn \\
&&=\frac{1}{2\sqrt{pq} P} \frac{\sinh(\frac{2\pi P}{\sqrt{pq}})}{\cos{\frac{\pi t}{pq}}-\cosh{\frac{2\pi P}{\sqrt{pq}}}},
\eea
which gives,
\beq
\label{Zintegrala}
\C{Z}= \frac{2\pi^2 }{\sqrt{2}pq} \int^{\infty}_0 \frac{dP}{P} \frac{\cos(2\pi\sigma_1 P)\cos(2\pi\sigma_2 P)\sinh(\frac{2\pi P}{\sqrt{pq}})}{\sinh(2\pi b P)\sinh(\frac{2\pi P}{b})} F_{k,l,r,s}(\frac{2\pi i P}{\sqrt{pq}}),
\eeq
where $F_{k,l,r,s}(z)$ is given by,
\bea 
\label{Fklrs}
F_{k,l,r,s}(z)=\sum^{\lambda_p(k,r)}_{\eta=|k-r|+1,2} \sum^{\lambda_q(l,s)}_{\rho=|l-s|+1,2} \sum^{p-1}_{a=1} \sum^{q-1}_{b=-(q-1)} &&\Bigg[\sin(\frac{\pi t}{p})\sin(\frac{\pi t}{q}) \times \nn \\ 
&&U_{\eta-1}(\cos\frac{\pi t}{p})U_{\rho-1}(\cos\frac{\pi t}{q}) \nn \frac{1}{\cos{\frac{\pi t}{pq}}-\cos{z}}\Bigg],\\
\eea
where $\lambda_p(k,r)=\mathrm{min}(k+r-1,2p-1-k-r)$. To compute \eqref{Zintegrala} we first note that the integrand is symmetric in $P$ and so we can extend the integral to the whole real line. Without loss of generality we assume $\sigma_1 > \sigma_2$, and break up the factor of $\cos(2\pi\sigma_1 P)$ in \eqref{Zintegrala} into its exponentials. For the integral containing, $\exp(2\pi i \sigma_1 P)$ we can close the integral along the real line by a semi-circle to produce the contour $C_+$ in the upper-half plane, for the other term we use a contour, $C_-$ in the lower-half plane. However, the integral using the contour in the lower-half plane, by a change of variables from $P$ to $-P$ is equal to the integral using the contour in the upper half-plane, we therefore get,
\beq
\label{Zintegralb}
\C{Z}(k,l;\sigma_1|r,s;\sigma_2)= A_C\oint_{C_+} dP\left[ G(P,\sigma_1,\sigma_2) F_{k,l,r,s}(\frac{2\pi i P}{\sqrt{pq}})\right],
\eeq
where $A_C$ is a numerical constant and
\beq
G(P,\sigma_1,\sigma_2)= \frac{1}{P} \frac{e^{2\pi i \sigma_1 P}\cos(2\pi\sigma_2 P)\sinh(\frac{2\pi P}{\sqrt{pq}})}{\sinh(2\pi b P)\sinh(\frac{2\pi P}{b})}
\eeq
The poles of $F_{k,l,r,s}(\frac{2 \pi i P}{\sqrt{pq}})$ occur at $P \in S_F = \{\frac{n i}{2\sqrt{p q}}:n\in \BB{Z}, n \neq 0 \mod p, n \neq 0 \mod q\}$. The residue of $F_{k,l,r,s}(\frac{2 \pi i P}{\sqrt{pq}})$ at $P \in S_F$ is,
\bea
\label{Fresidues}
\res(F_{k,l,r,s}(\frac{2 \pi i P}{\sqrt{pq}}); P = \frac{i n}{2\sqrt{p q}}) =&& -2\frac{\sqrt{p q}}{2 \pi i} \sin(\frac{\pi n}{p})\sin(\frac{\pi n}{q}) U_{k-1}(\cos\frac{\pi n}{p})U_{r-1}(\cos\frac{\pi n}{p}) \times \nn \\ 
&&U_{l-1}(\cos\frac{\pi n}{q})U_{s-1}(\cos\frac{\pi n}{q})\frac{1}{\sin \frac{\pi n}{pq}}.
\eea
The residues of the $G(P)$ are,
\bea
\label{Gresidues}
&&2\pi i\res\Big[G(P);P=i\epsilon\Big] = (\frac{1}{2\epsilon}-\pi \sigma_1)\frac{1}{\sqrt{pq}}, \\
&&2\pi i\res\Big[G(P);P=\frac{inp}{2\sqrt{pq}}\Big]=\frac{2}{n} (-1)^n e^{-\frac{\pi p n \sigma_1}{\sqrt{pq}}} \cosh(n \pi b \sigma_2) \frac{\sin(\pi n /q)}{\sin(\pi p n /q)},\nn \\
&&2\pi i\res\Big[G(P);P=\frac{inq}{2\sqrt{pq}}\Big]=\frac{2}{n} (-1)^n e^{-\frac{\pi q n \sigma_1}{\sqrt{pq}}} \cosh(n \pi \sigma_2/b) \frac{\sin(\pi n /p)}{\sin(\pi q n /p)},\nn \\
&&2\pi i\res\Big[G(P);P=\frac{inpq}{2\sqrt{pq}}\Big]=\frac{2}{npq} (-1)^{n(p+q+1)} e^{-\frac{\pi p q n \sigma_1}{\sqrt{pq}}} \cosh(n \pi q b \sigma_2).\nn 
\eea
Finally, a useful identity for computing the sums found in the computation of the cylinder amplitudes is,
\bea
\label{sumident}
\sum^{\infty}_{t=1}\frac{4}{t}\gamma^t e^{-\alpha t}&&\cosh\beta t\cos\phi t = \\
&&2\alpha - \ln(2(\cosh^2\alpha+\cosh^2\beta-2\gamma \cosh\alpha\cosh\beta\cos\phi-\sin^2\phi)).\nn
\eea 

\chapter{Loop equations for Disc with Handle \label{ChapLoopEqn:app}}
\label{MoreLoopEqns}
To complete the calculation of the disc with handle amplitude in the SX matrix model we need a number of additional loop equations.
In order to calculate the quantities $W_{SXSX}$ and $W_{SXSXSX}$ we need,
\bea
X &&\rightarrow X+\epsilon \left( \frac{1}{z-S}K(X,S) \frac{1}{x-X}+\frac{1}{x-X} K^{\dagger}(X,S)\frac{1}{z-S} \right)\\
S &&\rightarrow S+\epsilon \left( \frac{1}{z-S}K(X,S) \frac{1}{x-X}+\frac{1}{x-X} K^{\dagger}(X,S)\frac{1}{z-S} \right)\\
X &&\rightarrow X+\epsilon \left( \frac{1}{z_1-S}\frac{1}{x_1-X}L(X,S)\frac{1}{x_2-X}\frac{1}{z_2-S}+h.c \right).
\eea
Finally the following loop equations allow for quantities of the form, $W_{SX;SX}$ etc. to be found, 
\bea
X &&\rightarrow X+\epsilon \left( \frac{1}{z-S}K(X,S) \frac{1}{x-X}+\frac{1}{x-X} K^{\dagger}(X,S)\frac{1}{z-S} \right)\Tr{H(S,X)} \\
S &&\rightarrow S+\epsilon \left( \frac{1}{z-S}K(X,S) \frac{1}{x-X}+\frac{1}{x-X} K^{\dagger}(X,S)\frac{1}{z-S} \right)\Tr{H(S,X)} \\
X &&\rightarrow X+\epsilon \left( \frac{1}{z_1-S}\frac{1}{x_1-X}L(X,S)\frac{1}{x_2-X}\frac{1}{z_2-S}+h.c \right)\Tr{H(S,X)}.
\eea 

\backmatter


\newpage 

\footnotesize
\addcontentsline{toc}{chapter}{Bibliography}

\begin{thebibliography}{99}

\bibitem{hooft}{G. ’t Hooft and M. J. G. Veltman, One loop divergencies in the theory of gravitation, Annales Poincare Phys. Theor. A 20 (1974) 69-94.}
\bibitem{weinberg}{S. Weinberg, “Ultraviolet divergences in quantum theories of gravitation,” in General Relativity, S. Hawking and W. Israel, eds., p. 790-831, 1980.}

\bibitem{Durhuus:2005fq}
B. Durhuus, T. Jonsson, J.F. Wheater,  {Random walks on combs},
J. Phys. {A39}, 1009--1038 (2006).

\bibitem{Ambjorn:1998} J. Ambj\o rn and R. Loll, Non-perturbative Lorentzian quantum gravity, causality and topology change, Nucl. Phys. B536 (1998) 407-434, ([arXiv: hep-th/9805108).

\bibitem{Ambjorn:2004aa} J. Ambj\o rn, J. Jurkiewicz, and R. Loll, Emergence of a 4D world from causal quantum gravity, Phys. Rev. Lett. 93 (2004) 131301, arXiv: hep-th/0404156.
\bibitem{Ambjorn:2008aa} J. Ambj\o rn, A. Goerlich, J. Jurkiewicz, and R. Loll, Planckian birth of the quantum de Sitter universe, Phys. Rev. Lett. 100 (2008) 091304, arXiv: 0712.2485 [hep-th].

\bibitem{Ambjorn:2005aa} J. Ambj\o rn, J. Jurkiewicz, and R. Loll, Spectral dimension of the universe, Phys. Rev. Lett. 95 (2005) 171301, arXiv: hep-th/0505113.
\bibitem{Lauscher} O. Lauscher and M. Reuter, Fractal spacetime structure in asymptotically safe gravity, JHEP 10 (2005) 050, arXiv: hep-th/0508202.

\bibitem{Horava:2009aa} P. Horava, Quantum Gravity at a Lifshitz Point, Phys. Rev. D79 (2009) 084008, arXiv: 0901.3775 [hep-th].
\bibitem{Horava:2009ab} P. Horava, Spectral dimension of the universe in quantum gravity at a Lifshitz point, Phys. Rev. Lett. 102 (2009) 161301, arXiv: 0901.3775 [hep-th].

\bibitem{benedetti} D. Benedetti and J. Henson, Spectral geometry as a probe of quantum spacetime, Phys. Rev. D20 (2009) 124036, arXiv: 0911.0401 [hep-th].

\bibitem{Carlip:2009aa} S. Carlip, The small scale structure of space-time, arXiv:1009.1136 [hep-th].

\bibitem{tao} T. Tao, 
The Euler-Maclaurin formula, Bernoulli numbers, the zeta function, and real-variable analytic continuation,  http://terrytao.wordpress.com/tag/euler-summation-formula/

\bibitem{Flajolet}

    P. Flajolet, X. Gourdon and P. Dumas,
    Mellin Transforms And Asymptotics: Harmonic Sums,
    Theoretical Computer Science 144, 3-58 (1995).

\bibitem{Durhuus:2006vk}
  B.~Durhuus, T.~Jonsson and J.~F.~Wheater, The spectral dimension of generic trees, J. Stat. Phys. 128, 1237-1260 (2007), 
  arXiv:math-ph/0607020.

\bibitem{Durhuus:2009sm}
  B.~Durhuus, T.~Jonsson and J.~F.~Wheater,
  On the spectral dimension of causal triangulations',  J. Stat. Phys. 139, 859-881 (2007),
  arXiv:0908.3643 [math-ph].

\bibitem{Seiberg:2003nm}
  N.~Seiberg and D.~Shih,
  JHEP {\bf 0402} (2004) 021
  [arXiv:hep-th/0312170].

\bibitem{Fateev:2000ik}
  V.~Fateev, A.~B.~Zamolodchikov and A.~B.~Zamolodchikov,
  arXiv:hep-th/0001012.

\bibitem{Dorn:1994xn}
  H.~Dorn and H.~J.~Otto,
  Nucl.\ Phys.\  B {\bf 429} (1994) 375
  [arXiv:hep-th/9403141].

\bibitem{Zamolodchikov:1995aa}
  A.~B.~Zamolodchikov and A.~B.~Zamolodchikov,
  Nucl.\ Phys.\  B {\bf 477} (1996) 577
  [arXiv:hep-th/9506136].

\bibitem{Ponsot:2001ng}
  B.~Ponsot and J.~Teschner,
  Nucl.\ Phys.\  B {\bf 622} (2002) 309
  [arXiv:hep-th/0110244].

\bibitem{Kutasov:2004fg}
  D.~Kutasov, K.~Okuyama, J.~w.~Park, N.~Seiberg and D.~Shih,
  JHEP {\bf 0408} (2004) 026
  [arXiv:hep-th/0406030].

\bibitem{Basu:2005sda}
  A.~Basu and E.~J.~Martinec,
  Phys.\ Rev.\  D {\bf 72} (2005) 106007
  [arXiv:hep-th/0509142].

\bibitem{Eynard:2002kg}
  B.~Eynard,
  JHEP {\bf 0301} (2003) 051
  [arXiv:hep-th/0210047].

\bibitem{Moore:1991ir}
  G.~W.~Moore, N.~Seiberg and M.~Staudacher,
  Nucl.\ Phys.\  B {\bf 362} (1991) 665.

\bibitem{Belavin:2010pj}
  A.~Belavin and G.~Tarnopolsky,
  arXiv:1006.2056 [hep-th].

\bibitem{Anazawa:1997zs}
  M.~Anazawa,
  Nucl.\ Phys.\  B {\bf 501} (1997) 251
  [arXiv:hep-th/9702126].

\bibitem{Lian:1992aj}
  B.~H.~Lian and G.~J.~Zuckerman,
  Commun.\ Math.\ Phys.\  {\bf 145} (1992) 561.

\bibitem{Eynard:1992cn}
  B.~Eynard and J.~Zinn-Justin,
  Nucl.\ Phys.\  B {\bf 386} (1992) 558
  [arXiv:hep-th/9204082].

\bibitem{Eynard:1995nv}
  B.~Eynard and C.~Kristjansen,
  Nucl.\ Phys.\  B {\bf 455} (1995) 577
  [arXiv:hep-th/9506193].

\bibitem{Carroll:1995qd}
  S.~M.~Carroll, M.~E.~Ortiz and W.~Taylor,
  Nucl.\ Phys.\  B {\bf 468} (1996) 420
  [arXiv:hep-th/9510208].

\bibitem{Ambjorn:2008gk}
  J.~Ambjorn, R.~Loll, Y.~Watabiki, W.~Westra and S.~Zohren,
  Phys.\ Lett.\  B {\bf 670} (2008) 224
  [arXiv:0810.2408 [hep-th]].

\bibitem{Ambjorn:2008hv}
  J.~Ambjorn, R.~Loll, Y.~Watabiki, W.~Westra and S.~Zohren,
  Acta Phys.\ Polon.\  B {\bf 39} (2008) 3355
  [arXiv:0810.2503 [hep-th]].

\bibitem{Ambjorn:2009fm}
  J.~Ambjorn, R.~Loll, W.~Westra and S.~Zohren,
  Phys.\ Lett.\  B {\bf 678} (2009) 227
  [arXiv:0905.2108 [hep-th]].

\bibitem{Martinec:2003ka}
  E.~J.~Martinec,
  arXiv:hep-th/0305148.

\bibitem{Gesser:2010fi}
  J.~A.~Gesser,
  arXiv:1010.5006 [hep-th].

\bibitem{Ishiki:2010wb}
  G.~Ishiki and C.~Rim,
  Phys.\ Lett.\  B {\bf 694} (2010) 272
  [arXiv:1006.3906 [hep-th]].

\bibitem{Bourgine:2010ja}
  J.~E.~Bourgine, G.~Ishiki and C.~Rim,
  JHEP {\bf 1012} (2010) 046
  [arXiv:1010.1363 [hep-th]].

\bibitem{Ambjorn:2007ge}
  J.~Ambjorn and J.~A.~Gesser,
  Acta Phys.\ Polon.\  B {\bf 38} (2007) 3993
  [arXiv:0709.3106 [hep-th]].

\bibitem{Ambjorn:2007xe}
  J.~Ambjorn and J.~A.~Gesser,
  Phys.\ Lett.\  B {\bf 659} (2008) 718
  [arXiv:0707.3431 [hep-th]].

\bibitem{Teschner:2001rv}
  J.~Teschner,
  Class.\ Quant.\ Grav.\  {\bf 18} (2001) R153
  [arXiv:hep-th/0104158].

\bibitem{Ginsparg:1993is}
  P.~H.~Ginsparg and G.~W.~Moore,
  arXiv:hep-th/9304011.

\bibitem{Carroll:1995nj}
  S.~M.~Carroll, M.~E.~Ortiz and W.~Taylor,
  Nucl.\ Phys.\  B {\bf 468} (1996) 383
  [arXiv:hep-th/9510199].

\bibitem{Watabiki:1993ym}
  Y.~Watabiki,
  Nucl.\ Phys.\  B {\bf 441} (1995) 119
  [arXiv:hep-th/9401096].

\bibitem{Zohren:2009dj}
  S.~Zohren,
  arXiv:0905.0213 [hep-th].

\bibitem{Carlip:2008zf}
  S.~Carlip,
  Class.\ Quant.\ Grav.\  {\bf 25} (2008) 154010
  [arXiv:0803.3456 [gr-qc]].

\bibitem{Bojowald:2007ky}
  M.~Bojowald,
  AIP Conf.\ Proc.\  {\bf 910} (2007) 294
  [arXiv:gr-qc/0702144].

\bibitem{QuantumGravity}
  C.Rovelli,
  Quantum Gravity,
  Cambridge University Press, 2004

\bibitem{Polchinski1}
  J. Polchinski,
  String Theory Volume 1,
  Cambridge Universtiy Press, 1998


\bibitem{Polchinski2}
  J. Polchinski,
  String Theory Volume 2,
  Cambridge Universtiy Press, 1998

\bibitem{GSW1}
  M.B. Green, J.H. Schwarz, E. Witten,
  Superstring Theory Volume 1,
  Cambridge University Press, 1987

\bibitem{Nakayama:2004vk}
  Y.~Nakayama,
  Int.\ J.\ Mod.\ Phys.\  A {\bf 19} (2004) 2771
  [arXiv:hep-th/0402009].

\bibitem{EynardReview}
  B. Eynard,
  Random Matrices,
  Lecture course, Saclay, 2000

\bibitem{Makeenkobook}
  Makeenko, Yuri. Methods of Contemporary Gauge Theory. Cambridge University Press, 2002. Cambridge Books Online. Cambridge University Press. 19 May 2011 http://dx.doi.org/10.1017/CBO9780511535147

\bibitem{DTsim}
M. Weigel, Ver tex models on random graphs. PhD thesis, Universit at Leipzig, 2002.

\bibitem{Kawai:1993cj}
  H.~Kawai, N.~Kawamoto, T.~Mogami and Y.~Watabiki,
  Phys.\ Lett.\  B {\bf 306} (1993) 19
  [arXiv:hep-th/9302133].

\bibitem{Ambjorn:1997jf}
  J.~Ambjorn, D.~Boulatov, J.~L.~Nielsen, J.~Rolf and Y.~Watabiki,
  JHEP {\bf 9802} (1998) 010
  [arXiv:hep-th/9801099].

\bibitem{Benedetti:2008hc}
  D.~Benedetti and J.~Henson,
  Phys.\ Lett.\  B {\bf 678} (2009) 222
  [arXiv:0812.4261 [hep-th]].


\bibitem{Ambjorn:2008jf}
  J.~Ambjorn, R.~Loll, Y.~Watabiki, W.~Westra and S.~Zohren,
  Phys.\ Lett.\  B {\bf 665} (2008) 252
  [arXiv:0804.0252 [hep-th]].

\bibitem{YB}
  P. Di Francesco, P. Mathieu, D. Senechal,
  Conformal Field Theory,
  Springer, 1997.

\bibitem{Mansfield:1990tu}
  P.~Mansfield,
  Rept.\ Prog.\ Phys.\  {\bf 53 } (1990)  1183-1251.

\bibitem{Polyakov:1981rd}
  A.~M.~Polyakov,
  Phys.\ Lett.\  B {\bf 103} (1981) 207.

\bibitem{D'Hoker:1988ta}
  E.~D'Hoker and D.~H.~Phong,
  Rev.\ Mod.\ Phys.\  {\bf 60} (1988) 917.

\bibitem{Distler:1988jt}
  J.~Distler and H.~Kawai,
  Nucl.\ Phys.\  B {\bf 321} (1989) 509.

\bibitem{Eynard:2008we}
  B.~Eynard and N.~Orantin,
  arXiv:0811.3531 [math-ph].

\bibitem{Atkin:2011ak}
  M.~R.~Atkin, G.~Giasemidis and J.~F.~Wheater,
  arXiv:1101.4174 [hep-th].

\bibitem{Atkin:2010yv}
  M.~R.~Atkin and J.~F.~Wheater,
  JHEP {\bf 1102} (2011) 084
  [arXiv:1011.5989 [hep-th]].

\bibitem{Benedetti:2009ge}
  D.~Benedetti and J.~Henson,
  Phys.\ Rev.\  D {\bf 80} (2009) 124036
  [arXiv:0911.0401 [hep-th]].

\bibitem{Goulian:1990qr}
  M.~Goulian and M.~Li,
  Phys.\ Rev.\ Lett.\  {\bf 66} (1991) 2051.

\bibitem{Teschner:1995yf}
  J.~Teschner,
  Phys.\ Lett.\  B {\bf 363} (1995) 65
  [arXiv:hep-th/9507109].

\bibitem{Seiberg:1990eb}
  N.~Seiberg,
  Prog.\ Theor.\ Phys.\ Suppl.\  {\bf 102} (1990) 319.

\bibitem{Ambjorn:2009rv}
  J.~Ambjorn, R.~Loll, Y.~Watabiki, W.~Westra and S.~Zohren,
  Acta Phys.\ Polon.\  B {\bf 40} (2009) 3479
  [arXiv:0911.4208 [hep-th]].

\bibitem{Ambjorn:1995dg}
  J.~Ambjorn, Y.~Watabiki,
  Nucl.\ Phys.\  {\bf B445 } (1995)  129-144.
  [hep-th/9501049].




\end{thebibliography}
\newpage

\addcontentsline{toc}{chapter}{Index}
\printindex
\normalsize



\end{document}